\documentclass[12pt]{article}
\usepackage{amsmath,amssymb,mathtools,amsthm}
\usepackage{graphicx,enumitem,color,threeparttable,booktabs,multirow,makecell}
\usepackage[title]{appendix}
\usepackage[hyphens]{url}
\usepackage{hyperref}
\newtheorem{theorem}{Theorem}
\newtheorem{definition}{Definition}

\newtheorem{lemma}{Lemma}
\newtheorem{assumption}{Assumption}

\allowdisplaybreaks[4]
\newcommand{\indep}{\perp \!\!\! \perp}
\newcommand{\E}{\mathrm{E}}
\newcommand{\Var}{\mathrm{Var}}
\newcommand{\Cov}{\mathrm{Cov}}
\renewcommand{\P}{\mathrm{P}}
\newcommand{\bA}{\boldsymbol{A}}
\newcommand{\bX}{\boldsymbol{X}}
\newcommand{\bM}{\boldsymbol{M}}
\newcommand{\bm}{\boldsymbol{m}}
\newcommand{\bD}{\boldsymbol{D}}
\newcommand{\bd}{\boldsymbol{d}}
\newcommand{\bY}{\boldsymbol{Y}}
\newcommand{\be}{\boldsymbol{\varepsilon}}
\newcommand{\bn}{\boldsymbol{\nu}}
\newcommand{\bo}{\boldsymbol{\omega}}
\newcommand{\bdmi}{\boldsymbol{d}_{-i}}
\newcommand{\bmmi}{\boldsymbol{m}_{-i}}
\newcommand{\bDmi}{\boldsymbol{D}_{-i}}
\newcommand{\bMmi}{\boldsymbol{M}_{-i}}
\newcommand{\1}{\boldsymbol{1}}

\usepackage[round]{natbib}
\newcommand{\anon}{1}

\begin{document}

\def\spacingset#1{\renewcommand{\baselinestretch}%
{#1}\small\normalsize} \spacingset{1}


\if1\anon
{
  \title{\bf Causal Mediation Analysis for Network Data with Graph Neural Network}
  \author{Peikai Wu\hspace{.2cm}\\
    Department of Statistics and Data Science, Fudan University, China\\
    and \\
    Zhiguo Xiao \\
    Department of Statistics and Data Science, Fudan University, China}
  \maketitle
} \fi

\if0\anon
{
  \bigskip
  \bigskip
  \bigskip
  \begin{center}
    {\LARGE\bf Causal Mediation Analysis for Network Data with Graph Neural Network}
\end{center}
  \medskip
} \fi

\bigskip
\begin{abstract}
Causal mediation analysis is typically formulated under no interference, an assumption often violated in networked populations. We develop a nonparametric framework for a single large observed network that allows simultaneous treatment and mediator spillovers and high-dimensional network confounding. Exposure and mediator mappings define causal estimands without restricting the true interference mechanism, separating own from spillover effects without prespecified aggregation models. Under strengthened conditional independence conditions, we identify own controlled direct, natural direct, and natural indirect effects and give primitive sufficient conditions in terms of structural errors. We construct augmented inverse probability weighted estimators that are doubly robust for controlled effects and multiply robust for natural effects, using graph neural networks to learn high-dimensional nuisance functions from node features and the adjacency matrix. Under approximate neighborhood interference, weak network dependence, and suitable first-stage rates, we establish asymptotic normality of the effect estimators and consistency of a network HAC variance estimator. In simulations the graph neural network estimator outperforms machine learning methods built on hand-constructed neighborhood features, and a reanalysis of an agricultural insurance experiment in rural China finds insurance knowledge to be a substantive mediating channel while perception-based mediators are not.
\end{abstract}

\noindent%
\vfill

\newpage
\spacingset{1.8} 

\section{Introduction}
Causal mediation analysis aims to open the black box through which a treatment acts on an outcome, decomposing a total effect into an indirect effect transmitted by a mediator and a direct effect that does not operate through that mediator. Ever since \citet{robins1992identifiability} and \citet{pearl2001direct} gave formal definitions of the controlled direct effect and of the natural direct and indirect effects within the potential outcome framework, a substantial literature on identification conditions, semiparametric estimation and sensitivity analysis has accumulated \citep[see, e.g.,][]{imai2010identification,tchetgen2012semiparametric,vanderweele2015explanation}.

This body of theory is, however, built almost without exception on the stable unit treatment value assumption, under which the potential mediator and the potential outcome of an individual depend on that individual's own treatment alone. Among individuals connected by a social network the assumption typically fails, and interference is frequently the very object of study. Retaining definitions of causal effects designed for settings without interference then delivers direct and indirect effects that neither admit a clear causal interpretation nor can be distinguished from spillover mechanisms \citep{ogburn2014causal}, so that a framework for causal mediation analysis designed specifically for networks is called for.

The literature on causal inference in networks is by now rich, ranging from partial interference \citep{hudgens2008toward} to general networks \citep[e.g.,][]{aronow2017estimating,forastiere2021identification,li2022random,leung2022causal,ogburn2024causal}. The overwhelming majority of this work relies on an \emph{exposure mapping}, a known function that aggregates the treatments of all individuals into a low-dimensional vector according to the network structure. \citet{savje2024causal} points out that in the existing literature the exposure mapping discharges two duties at once, namely defining the causal estimand and delimiting the structure of interference, and that the two ought to be separated: a researcher may legitimately use an exposure mapping to \emph{define} the effect of interest, but has no reason to require the true interference mechanism to operate exactly in the form of that mapping. Once the mapping is misspecified, identification results built upon it forfeit their causal interpretation. Following this line of reasoning, \citet{leung2022graph} uses graph neural networks to handle high-dimensional network confounding, but does not consider mediation mechanisms.

By contrast, mediation analysis in network settings remains rather limited, and is for the most part constrained in one of three ways. The first concerns the scope of interference: \citet{vanderweele2013mediation} decompose the natural indirect effect into an own mediation effect and a spillover mediation effect in group-randomized trials, and subsequent work largely stays within the partial interference setting, which is inapplicable when only a single large network is observed. The second concerns functional form: \citet{kundu2024network} extend the exposure mapping paradigm to settings in which the mediator and the outcome are simultaneously network dependent, but their estimation and inference theory is built on a linear random-effects structural equation model and still imposes structural restrictions through exposure mappings. The third concerns the question posed: ``network mediation analysis'' in psychometrics \citep[e.g.,][]{liu2021social,che2021network} treats the network structure itself, or its latent space representation, as the mediator, and thus answers a question different from ours. To the best of our knowledge, no existing work systematically studies identification, estimation and inference for causal mediation analysis in a single large network under a general interference structure and a nonparametric model.

The present paper is intended to fill this gap. We assume that the treatment, the mediator and the outcome are generated by unknown structural functions, where the outcome function is allowed to depend in an arbitrary manner on the treatment vector, the mediator vector, the covariates of all individuals and the adjacency matrix, so that treatment spillover, mediator spillover and network confounding through covariates are accommodated simultaneously. Following the recommendation of \citet{savje2024causal}, we introduce an exposure mapping and a \emph{mediator mapping}, but their sole function is to define the causal estimands: we move an individual's own treatment and own mediator outside the mappings, leaving the mappings to summarize the spillover from the remaining individuals, which separates own effects from spillover effects cleanly. On this basis we define the own controlled direct effect and the own natural direct and indirect effects, and in the appendix we extend the development to interventional effects and to the spillover version of each effect. Throughout, the analysis is carried out conditional on the covariate matrix and the adjacency matrix. This is faithful to the reality that an observational study yields a single realized network, and it avoids the misspecification risk incurred by positing a random graph model.

Since we do not use the exposure mapping to restrict the form in which spillover operates, the sequential ignorability conditions of the setting without interference, which involve only an individual's own variables, no longer suffice for identification. We propose correspondingly strengthened conditional independence assumptions: the potential outcome is conditionally independent of the treatment vector of \emph{all} individuals and, when the natural effects are to be identified, additionally of the own mediator and of the cross-world natural mediator. We also provide primitive sufficient conditions for these assumptions within our model, namely conditional independence among the error terms, so that their plausibility may be assessed from the data generating mechanism itself. Under these conditions every causal estimand equals a functional constructed entirely from conditional moments of observable variables.

Turning to estimation, we construct augmented inverse probability weighted (AIPW) estimators \citep{robins1994estimation,bang2005doubly} of each causal estimand, whose score functions are doubly or multiply robust depending on the effect at hand. The genuine difficulty lies in estimating the nuisance functions entering the scores: in our setting the confounding originates from the covariates of all individuals and from the entire network, and is therefore high-dimensional. Commonly used machine learning methods cannot take a network as a direct input, so that the user must first summarize the network information into features such as neighborhood averages, which in effect reintroduces the risk of misspecification. We therefore estimate all nuisance functions with graph neural networks \citep{kipf2017semi,hamilton2017inductive}, which take node features and the adjacency structure as direct inputs and aggregate higher-order neighborhood information into the representation of the target node by layerwise message passing, without requiring the form of the network features to be specified in advance.

Because a single network is observed, the sample is neither independent nor identically distributed, and the asymptotic theory must restrict dependence at two levels. At the level of causal structure, we extend the approximate neighborhood interference assumption of \citet{leung2022causal} to the mediation setting; at the level of the statistical distribution, we show that the centered individual-level scores form a $\psi$-dependent random field \citep{kojevnikov2021limit} and then invoke a network central limit theorem. The first-order remainder induced by graph neural network estimation must also be controlled, for which we impose the $n^{-1/4}$ rate condition standard in the double machine learning literature \citep{chernozhukov2018double,farrell2021deep}. On this basis we establish the asymptotic normality of the estimators and construct a network HAC variance estimator whose bandwidth adapts to the size and density of the network, and we prove its consistency, thereby making large-sample inference feasible.

We evaluate the proposed framework by simulation and by an empirical application. Under two network generating mechanisms and a data generating process containing higher-order nonlinear spillover, the graph neural network estimator dominates a multilayer perceptron and a random forest built on neighborhood averages in terms of bias, standard deviation and confidence interval coverage. In the empirical part we revisit the agricultural insurance experiment of \citet{cai2015social} and find that the intensive session raises take-up significantly and that a considerable share of the effect is transmitted through insurance knowledge, whereas the natural indirect effect is essentially zero when the perceived take-up rate or the subjective disaster probability serves as the mediator. The informal discussion of mechanisms in the original study is thereby turned into an estimable and testable causal decomposition.

Our contributions may be summarized under three headings. First, we develop a framework for causal mediation analysis in a single large observed network in which the treatment and the mediator spill over simultaneously. More importantly, we extend the distinction drawn by \citet{savje2024causal} between the two duties of the exposure mapping to mediation analysis: both mappings define the causal estimands only and do not restrict the true interference structure, and on this basis we identify a variety of own effects and spillover effects. Second, we combine GNN-based adjustment for high-dimensional network confounding with robust AIPW estimation of mediation effects. The system of nuisance functions here is far more involved than for a general treatment effect; the scores we construct are doubly or multiply robust, and all nuisance functions are learned from the entire network by a PNA-type GNN. To the best of our knowledge, no existing study addresses treatment and mediator spillover, high-dimensional nonlinear confounding, exposure mapping misspecification and machine learning debiasing within a single framework. Third, we develop inference theory for these estimators under a single growing network. Unlike existing work that relies on independent groups under group randomization or on parametric structural equation models, we allow the individual scores to depend on one another along network paths of arbitrary order and let the nuisance functions be estimated nonparametrically by a GNN. We prove asymptotic normality under approximate neighborhood interference, give the conditions that the first-stage convergence rate, the growth of the network and the decay of dependence must jointly satisfy, and construct a consistent network HAC variance estimator, so that the estimators of the various effects are equipped with valid large-sample inference.

The remainder of the paper is organized as follows. Section \ref{ch5:sec:problem-setup} introduces the notation, the network mediation model, the causal estimands and their identification theory. Section \ref{ch5:sec:estimation} presents the AIPW estimators together with the graph neural network estimators of the nuisance functions. Section \ref{ch5:sec:asymptotic-theory} establishes asymptotic normality of the AIPW estimators and the consistency of the variance estimator. Sections \ref{ch5:sec:simulation-study} and \ref{ch5:sec:empirical-application} report the simulation results and the empirical application, respectively, and Section \ref{ch5:sec:conclusion} concludes. Asymptotic theory for other related causal effects estimators, additional simulation results and proofs are provided in the online supplement.

\section{Problem Setup}
\label{ch5:sec:problem-setup}
\subsection{Notation and the Network Mediation Model}
\label{ch5:subsec:network-mediation-model}
Consider $n$ individuals connected by a network. Write $\mathcal{N}_n \coloneq \{1,\ldots,n\}$ for the set of all individuals and let the $n\times n$ binary symmetric matrix $\bA$ represent the network structure, where $A_{ij}=A_{ji}=1$ indicates that individuals $i$ and $j$ are directly connected and $A_{ii}=0$. Let $\ell_{\bA}(i,j)$ denote the length of the shortest path between $i$ and $j$, each edge being of length one, with $\ell_{\bA}(i,j)=\infty$ if the two are not connected. For any individual $i$, define the $K$-neighborhood $\mathcal{N}(i,K) \coloneq \{j: \ell_{\bA}(i,j) \leqslant K\}$ and write $n(i,K) \coloneq |\mathcal{N}(i,K)|$ for its cardinality.

For each individual $i$ we observe a $d_X$-dimensional continuous covariate $X_i\in\mathbb{R}^{d_X}$, a binary treatment $D_i\in\{0,1\}$, a binary mediator $M_i\in\{0,1\}$ and a continuous outcome $Y_i\in\mathbb{R}$;\footnote{The binary treatment and single-mediator specification is adopted for simplicity. Continuous or multi-valued treatments and mediators may be accommodated by replacing, as appropriate, finite sums and probability masses with integrals and the corresponding conditional densities, including generalized propensity scores for continuous treatments. Pointwise effects at fixed continuous values additionally require kernel localization and the associated smoothing asymptotics. The GNN message-passing backbone can be retained, while the classification heads used for the propensity components are replaced by suitable conditional-density heads. A fixed multi-dimensional mediator vector can similarly be handled through its joint conditional law, yielding mediation through the vector as a whole; mediator-specific path effects require additional identifying assumptions.} in addition there are unobserved error terms $\varepsilon_i\in\mathbb{R}^{d_\varepsilon}$, $\nu_i\in\mathbb{R}^{d_\nu}$ and $\omega_i\in\mathbb{R}^{d_\omega}$. Let $d_i$ and $m_i$ denote realized values of $D_i$ and $M_i$, and write $\bd = (d_i)_{i=1}^N$, $\bD = (D_i)_{i=1}^N$, $\bm = (m_i)_{i=1}^N$, $\bM = (M_i)_{i=1}^N$, $\bX = (X_i)_{i=1}^N$, $\bY = (Y_i)_{i=1}^N$, $\be = (\varepsilon_i)_{i=1}^N$, $\bo = (\omega_i)_{i=1}^N$ and $\bn = (\nu_i)_{i=1}^N$. We treat $\bX,\bA,\{\varepsilon_i,\omega_i,\nu_i\}_{i=1}^n$ as random, but conduct the entire development---identification, estimation and asymptotic theory---conditional on $\bX$ and $\bA$. The reason is that an observational study ordinarily delivers a single realized network, and what the researcher cares about is precisely how variation in treatments and mediators within that network changes the outcomes. Conditioning on $\bA$ dispenses with the need to assume that the network is generated by some random graph model, and conditioning on $\bX$ dispenses with the need to assume a joint distribution and dependence structure for the covariates over the network; either assumption would introduce unnecessary misspecification risk.

Following the potential outcome framework of \citet{rubin1974estimating}, we place the treatments and mediators in parentheses after the outcome variable to denote the potential outcome under that particular treatment and mediator configuration. For example, $Y_i(\bd,\bm)$ denotes the potential outcome of individual $i$ when the treatments and mediators of all individuals are set to $\bd$ and $\bm$, respectively. We assume that it obeys the model
\begin{equation}
  Y_i(\bd,\bm) = f_n(i, \bd, \bm, \bX, \bA, \be),
  \label{ch5:eq:potential-outcome-model}
\end{equation}
where $f_n$ is an unknown function with range $\mathbb{R}$ when the total number of individuals is $n$. Letting $M_i(d_i)$ denote the potential mediator of individual $i$ when the treatment of individual $i$ is $d_i$, we assume that it obeys the model
\begin{equation}
  M_i(d_i) = g_n(i, d_i, \bX, \bA, \omega_i),
  \label{ch5:eq:potential-mediator-model}
\end{equation}
where $g_n$ is an unknown function with range $\{0,1\}$ when the total number of individuals is $n$. We further assume that the treatment $D_i$ of individual $i$ obeys the model
\begin{equation}
  D_i = h_n(i, \bX, \bA, \nu_i),
  \label{ch5:eq:treatment-assignment-model}
\end{equation}
where $h_n$ is an unknown function with range $\{0,1\}$ when the total number of individuals is $n$. Note that in $g_n(\cdot)$ and $h_n(\cdot)$ we assume that only the own treatment and own error term of individual $i$ enter the model.\footnote{This restriction is imposed solely in order to establish the identification results of Section \ref{ch5:subsec:causal-estimand-identification}; that is, its purpose is to endow our estimators with a causal interpretation. See Proposition 1 of \citet{leung2022graph} and \citet{leung2024identifying} for related discussion. If causal inference is not the objective, the estimation method and the large-sample theory developed below remain valid without this restriction.}

The specification above accommodates treatment spillover, mediator spillover and network confounding through covariates simultaneously. It is extremely permissive, in that it only stipulates which variables may enter the structural functions and places no restriction whatsoever on the form of their influence.

Under the standard consistency assumption, models (\ref{ch5:eq:potential-outcome-model}) and (\ref{ch5:eq:potential-mediator-model}) deliver the models obeyed by the observed outcome $Y_i$ and the observed mediator $M_i$ of individual $i$.
\begin{assumption}[Consistency] \label{ch5:ass:consistency}
  The observed variables of an individual coincide with the potential outcomes of those variables at the realized treatments, that is,
  \begin{gather*}
  Y_i = Y_i(\bD,\bM) = f_n(i, \bD, \bM, \bX, \bA, \be),\\
  M_i = M_i(D_i) = g_n(i, D_i, \bX, \bA, \omega_i).
\end{gather*}
\end{assumption}

\subsection{Causal Estimands}
\label{ch5:subsec:causal-estimands}
In the setting without interference the potential outcome model takes the general form $Y_i(D_i,M_i) = f(D_i,M_i,X_i,\varepsilon_i)$ and the definitions of the controlled direct effect and of the natural direct and indirect effects are immediate; in the network setting the definition of an effect is considerably more delicate. In studies that do not involve a mediator, a common device is to define an \emph{exposure mapping} $T_i = t_n(i,\bD,\bA)$ \citep{aronow2017estimating} and a \emph{network control} $W_i = w_n(i,\bX,\bA)$, which aggregate the treatments and the covariates of all individuals into low-dimensional vectors. A typical example is
\[
T_i = \bigg(D_i, \sum_{j=1}^n A_{ij}D_j\Big/\sum_{j=1}^n A_{ij}\bigg) \quad \text{and} \quad W_i = \bigg(X_i, \sum_{j=1}^n A_{ij}, \sum_{j=1}^n A_{ij}X_{ij}\Big/\sum_{j=1}^n A_{ij}\bigg).
\]
One then assumes that the true model for the potential outcome takes the form
\begin{equation}
  Y_i(T_i) = f(T_i,W_i,\varepsilon_i),
  \label{ch5:eq:exposure-mapping-outcome-model}
\end{equation}
so that the treatments and covariates of all individuals may affect the outcome only through $T_i$ and $W_i$, and defines the causal effect as $\E[Y_i(t_1) - Y_i(t_0)]$ for prespecified constants $t_1,t_0$.

This device makes it convenient to define a variety of causal estimands, but (\ref{ch5:eq:exposure-mapping-outcome-model}) imposes a strong restriction: the influence of neighbors beyond the order of aggregation is discarded entirely, and the forms of $T_i$ and $W_i$ are chosen by the researcher, whereas the variables of neighbors may very well not act on the outcome in that form. Aggregating higher-order information or adopting a more elaborate mapping does not eliminate the misspecification risk either. In response to this situation, \citet{savje2024causal} observes that the exposure mapping discharges the two duties of defining the causal estimand and of defining the structure through which network effects operate, which in the existing literature frequently appear together \citep[e.g.,][]{forastiere2021identification,leung2022causal,li2022random}, and that the two ought to be separated: it is reasonable to define a causal estimand via an exposure mapping, but one cannot insist that the variables of neighbors act on one's own outcome in the form of that mapping. Accordingly, we use exposure mappings solely to define causal estimands, while the variables always obey models (\ref{ch5:eq:potential-outcome-model})--(\ref{ch5:eq:treatment-assignment-model}) of Section \ref{ch5:subsec:network-mediation-model}; the exposure mapping plays no part whatsoever in determining the mechanism through which the variables act.

We now introduce the network causal estimands studied in this paper. To distinguish an individual's own variables from those of the others, for any vector $(V_i)_{i=1}^n$ define $\boldsymbol{V}_{-i} \coloneq (V_j)_{j=1,j\neq i}^n$, and rewrite $Y_i(\bd,\bm)$ in (\ref{ch5:eq:potential-outcome-model}) as $Y_i(d_i,\bdmi, m_i,\bmmi)$ and $M_i(\bd)$ in (\ref{ch5:eq:potential-mediator-model}) as $M_i(d_i,\bdmi)$. In order to distinguish own effects from spillover effects more sharply, we move an individual's own treatment outside the exposure mapping, leaving the mapping to summarize the spillover from the remaining individuals:
\[
  T_i \coloneq T_i(\bDmi) \coloneq t_n(i,\bDmi,\bA).
\]
Analogously, we define the mediator mapping
\[
  S_i \coloneq S_i(\bMmi) \coloneq s_n(i,\bMmi,\bA).
\]
Here $t_n(i,\bDmi,\bA)$ and $s_n(i,\bMmi,\bA)$ are known vector-valued functions chosen in accordance with the research question, which aggregate $\bDmi$ and $\bMmi$ into low-dimensional vectors according to the network structure when the sample size is $n$.

Let $\mathcal{J}_n \subseteq \mathcal{N}_n$ be a subset of individuals and $j_n \coloneq |\mathcal{J}_n|$. We study the following causal estimands.

\textit{Own controlled direct effect.}\quad In classical mediation analysis the controlled direct effect measures the influence of the treatment on the outcome when the mediator is held fixed at some value. In the network mediation setting this effect extends naturally to the own controlled direct effect, which measures the influence of an individual's own treatment on the outcome when the exposure mapping of the others' treatments, the own mediator and the mediator mapping of the others' mediators are all held fixed. Define
\begin{equation}
  \begin{aligned}
    Y^C_i(d_i,t,m_i,s) \coloneq \sum_{\bdmi:T_i(\bdmi)=t} & \sum_{\bmmi:S_i(\bmmi)=s} \E[Y_i(d_i,\bdmi,m_i,\bmmi)\mid \bX,\bA] \\ 
    & P(\bDmi=\bdmi,\bMmi=\bmmi \mid T_i=t,S_i=s,\bX,\bA),
  \end{aligned}
\end{equation}
which is a weighted average of the conditional expectations of the potential outcomes over all treatment and mediator configurations that produce $(d_i,t,m_i,s)$. The own controlled direct effect is then defined as
\begin{equation}
  OCDE(t,m,s) \coloneq \frac{1}{j_n}\sum_{i\in\mathcal{J}_n} \big[Y^C_i(1,t,m,s) - Y^C_i(0,t,m,s)\big],
\end{equation}
that is, the average effect on the outcome of switching the treatment from $0$ to $1$ among the individuals delimited by $\mathcal{J}_n$, holding $T_i$, $M_i$ and $S_i$ fixed at $t$, $m$ and $s$, respectively.

Restricting attention to $\mathcal{J}_n$ serves two purposes. It indicates that our framework permits the researcher to specify the population of interest, and it allows the overlap condition required by inverse probability weighted or doubly robust estimators to be satisfied. For instance, when $T_i = \1\{\sum_{j=1}^n A_{ij}D_j \allowbreak > 0\}$ and $OCDE(1,m,s)$ is to be estimated, individuals without any neighbor must be excluded: for such individuals $Y_i^C(d_i,1,m,s)$ is not well defined and the requirement that the propensity score be bounded away from zero and one fails. When overlap does hold, taking $\mathcal{J}_n$ to be $\mathcal{N}_n$ corresponds to the causal effect in the entire population.

\textit{Own natural direct and indirect effects.}\quad In classical mediation analysis the natural direct effect measures the influence of the treatment on the outcome when the mediator is set to the level that would naturally arise under a given treatment state, while the natural indirect effect measures the influence on the outcome of moving the mediator from its natural level under no treatment to its natural level under treatment, holding the treatment fixed. In the network mediation setting these extend to the own natural direct effect and the own natural indirect effect. Define
\begin{equation}
  \begin{aligned}
    Y^N_i(d_i,d_i^*,t) \coloneq \sum_{\bdmi:T_i(\bdmi)=t} & \E[Y_i(d_i,\bdmi,M_i(d_i^*),\bMmi(\bdmi)) \mid \bX,\bA] \\
    & \P(\bDmi = \bdmi \mid T_i = t,\bX,\bA),
  \end{aligned}
\end{equation}
where $\bMmi(\bdmi)\coloneq(M_j(d_j))_{j=1,j\neq i}^n$. This is the conditional expectation of the potential outcome when the individual's own treatment is $d_i$, the exposure mapping is $t$, and the own mediator is at its natural level under treatment $d_i^*$. The own natural direct effect and the own natural indirect effect are then defined as
\begin{gather}
  ONDE(t) \coloneq \frac{1}{j_n}\sum_{i\in\mathcal{J}_n} \big[Y_i^N(1,0,t) - Y_i^N(0,0,t)\big], \\
  ONIE(t) \coloneq \frac{1}{j_n}\sum_{i\in\mathcal{J}_n} \big[Y_i^N(1,1,t) - Y_i^N(1,0,t)\big].
\end{gather}
The former measures the average effect of an individual's own treatment on the outcome when the exposure mapping of the others is held fixed at $t$ and the own mediator is set to the level that would naturally arise under treatment $0$; the latter measures the average effect on the outcome of moving the own mediator from its natural level under no treatment to its natural level under treatment, holding the exposure mapping and the own treatment fixed at $t$ and $1$, respectively.

\textit{Other effects.}\quad Besides controlled and natural effects, the interventional effects familiar from mediation analysis extend to the network setting as well; in addition one may study the influence of the treatments and mediators of the others on one's own outcome, that is, the spillover version of each effect. For brevity, we relegate the definition, identification, estimation, and asymptotic theory of these effects to Appendix B.

\subsection{Identification of the Causal Estimands}
\label{ch5:subsec:causal-estimand-identification}
Since only the potential outcome corresponding to the realized treatment and mediator is observed for each individual, causal inference requires identifying assumptions that allow the expectation of an unobserved potential outcome to be replaced by the expectation of the potential outcomes of other individuals under their realized treatment and mediator assignments. In causal mediation analysis without interference the identifying assumptions involve only conditional independence among an individual's own treatment, mediator and outcome, but in the network setting this no longer suffices. In papers that use exposure and mediator mappings to impose structural restrictions, the identifying assumptions are strengthened to conditional independence between the individual outcome and the exposure and mediator mappings \citep[e.g.,][]{forastiere2021identification,kundu2024network}. Since we impose no such restrictions, we propose new identifying assumptions, namely conditional independence between the individual outcome and the treatments of all individuals, supplemented by the mediators of all individuals when some of the effects are to be identified.

We introduce the assumptions required for each effect in turn, beginning with an independence assumption.
\begin{assumption}[Independence of the error terms] \label{ch5:ass:error-term-independence}
  Conditional on $\bX,\bA$, the collection $\{(\varepsilon_i,\omega_i,\nu_i)\}_{i=1}^n$ is mutually independent.
\end{assumption}
This assumption rules out network correlation arising from unobserved network confounding, so that the observed network correlation may be attributed to the observed network structure and to the causal spillover mechanism, which is what permits the spillover mechanism to be identified from the data. It also furnishes the structural basis for the conditional independence assumptions below.

\textit{Own controlled direct effect.}\quad We propose the following conditional independence assumption for the identification of the own controlled direct effect.
\begin{assumption}
  \label{ch5:ass:own-cde-conditional-independence}
  For any $\bd,\bm$ and $i$,
  \[
  Y_i(\bd,\bm) \indep \bD \mid \bX,\bA \quad \text{and} \quad Y_i(\bd,\bm) \indep \bM \mid \bD,\bX,\bA.
  \]
\end{assumption}
This assumption is the network analogue of the sequential conditional independence assumption standard in causal mediation analysis: it requires that all common causes of an individual's potential outcome and of the treatment assignment of all individuals be contained in $\bX,\bA$, and that all common causes of an individual's potential outcome and of the mediators of all individuals be contained in $\bD,\bX,\bA$, so that both the treatment vector and the mediator vector may be regarded as conditionally randomly assigned relative to the potential outcome. Within our model a more primitive sufficient condition is that $\be$ and $\bn$ be independent given $\bX,\bA$ and that $\be$ and $\bo$ be independent given $\bD,\bX,\bA$.

Define the observable $\mu_i^C(d_i,t,m_i,s)$ and the induced observable OCDE as
\begin{equation}
  \mu_i^C(d_i,t,m_i,s) \coloneq \E(Y_i \mid D_i=d_i, T_i=t, M_i=m_i, S_i=s, \bX, \bA),
  \label{ch5:eq:own-cde-observed-functional}
\end{equation}
\[
  \text{and }OCDE^{obs}(t,m,s) \coloneq \frac{1}{j_n}\sum_{i\in\mathcal{J}_n} \big[\mu_i^C(1,t,m,s) - \mu_i^C(0,t,m,s)\big],
\]
respectively. We then have the following identification result.
\begin{theorem} \label{ch5:thm:own-cde-identification}
  Under Assumptions \ref{ch5:ass:consistency}, \ref{ch5:ass:error-term-independence} and \ref{ch5:ass:own-cde-conditional-independence}, $OCDE^{obs}(t,m,s) = OCDE(t,m,s)$ for any $t,s,m$.
\end{theorem}
The importance of this theorem lies in the equality it establishes between $OCDE(t,m,s)$, which involves unobservable potential outcomes, and $OCDE^{obs}(t,m,s)$, which is constructed from conditional expectations of the observed outcome. Estimation of and inference on $OCDE^{obs}(t,m,s)$ therefore accomplish estimation of and inference on $OCDE(t,m,s)$.

\textit{Own natural direct and indirect effects.}\quad For the own natural direct effect and the own natural indirect effect we propose the following identifying assumption.
\begin{assumption} \label{ch5:ass:own-natural-effects-conditional-independence}
  For any $\bd,m_i,d_i^*$ and $i$,
  \begin{gather}
    Y_i(d_i,\bdmi,m_i,\bMmi(\bdmi)) \indep \bD \mid \bX,\bA, \label{ch5:eq:own-natural-treatment-outcome-independence}\\
    Y_i(d_i,\bdmi,m_i,\bMmi(\bdmi)) \indep M_i \mid \bD,\bX,\bA, \label{ch5:eq:own-natural-mediator-outcome-independence}\\
    M_i(d_i) \indep D_i \mid \bX, \bA, \label{ch5:eq:own-natural-treatment-mediator-independence}\\
    Y_i(d_i,\bdmi,m_i,\bMmi(\bdmi)) \indep M_i(d_i^*) \mid \bX,\bA. \label{ch5:eq:own-natural-cross-world-independence}
  \end{gather}
\end{assumption}
This assumption may be viewed as the sequential ignorability assumption used to identify natural direct and indirect effects, transposed to the network mediation setting. The first two conditions parallel Assumption \ref{ch5:ass:own-cde-conditional-independence} and rule out unobserved treatment--outcome and mediator--outcome confounding, the only difference being that the fixed $\bmmi$ is replaced by $\bMmi(\bdmi)$ because natural effects are at issue; the third condition rules out unobserved treatment--mediator confounding. Since natural effects involve the natural mediator levels of the same individual under different treatment states, a fourth, cross-world independence condition is also required: it is what allows the weighted average over $m_i$ to be interpreted as the expected potential outcome when the own mediator is set to its natural level $M_i(d_i^*)$, and is thus the key to passing from the observable mediator distribution to the counterfactual one.

Within our model, a sufficient condition for Assumption \ref{ch5:ass:own-natural-effects-conditional-independence} is that, conditional on $\bX,\bA$, there be no unobserved common cause of the treatment vector, the individual's own mediator and the potential outcome. In particular, the assumption holds if $\{\be,(\omega_j)_{j\neq i}\} \indep \bn \mid \bX,\bA$ and $\omega_i \indep \{\bn, \be, (\omega_j)_{j\neq i}\} \mid \bX,\bA$.

Define the observable quantity $\mu_i^N(d_i,d_i^*,t)$ together with $ONDE^{obs}(t)$ and $ONIE^{obs}(t)$ constructed from it as
\begin{equation}
\begin{aligned}
\mu_i^N(d_i,d_i^*,t) \coloneq \sum_{m_i\in\{0,1\}} \E(Y_i \mid D_i=d_i,& T_i=t,M_i=m_i,\bX,\bA) \\
&\P(M_i = m_i \mid D_i = d_i^*,\bX,\bA),
\end{aligned}
\label{ch5:eq:own-natural-effects-observed-functional}
\end{equation}
\begin{gather*}
  ONDE^{obs}(t) \coloneq \frac{1}{j_n}\sum_{i\in\mathcal{J}_n}\big[\mu_i^N(1,0,t) - \mu_i^N(0,0,t)\big], \\
  \text{and } ONIE^{obs}(t) \coloneq \frac{1}{j_n}\sum_{i\in\mathcal{J}_n}\big[\mu_i^N(1,1,t) - \mu_i^N(1,0,t)\big],
\end{gather*}
respectively. We then have the following identification result.
\begin{theorem} \label{ch5:thm:own-natural-effects-identification}
  Under Assumptions \ref{ch5:ass:consistency}, \ref{ch5:ass:error-term-independence} and \ref{ch5:ass:own-natural-effects-conditional-independence}, 
  $
  \mu_i^N(d_i,d_i^*,t) = Y^N_i(d_i,d_i^*,t)
  $, for any $i,d_i,d_i^*,t$. Consequently, $
  ONDE^{obs}(t) = ONDE(t), ONIE^{obs}(t) = ONIE(t).
  $
\end{theorem}
The significance of this theorem parallels that of Theorem \ref{ch5:thm:own-cde-identification}: estimation of and inference on $ONDE^{obs}(t)$ and $ONIE^{obs}(t)$---defined entirely in terms of observables---serve as estimation of and inference on the causal estimands $ONDE(t)$ and $ONIE(t)$, which involve unobservable potential outcomes.

\section{Estimation}
\label{ch5:sec:estimation}
\subsection{Network Augmented Inverse Probability Weighted Estimators}
\label{ch5:subsec:network-aipw-estimation}
We now construct augmented inverse probability weighted (AIPW) estimators of the causal estimands defined in Section \ref{ch5:subsec:causal-estimands}. 

\textit{AIPW estimator of the OCDE.}\quad We refer to $\mu_i^C(d_i,t,m_i,s)$ in (\ref{ch5:eq:own-cde-observed-functional}) as the outcome regression nuisance function, and further define the propensity score nuisance function
\[
\pi_i^C(d,t,m,s) \coloneq \P(D_i=d, T_i=t, M_i=m, S_i=s \mid \bX, \bA).
\]
Based on these two nuisance functions, for any $(d,t,m,s)$ we define the individual-level AIPW score
\[
\tau_i^C(d,t,m,s) \coloneq \mu_i^C(d,t,m,s) + \frac{\1\{D_i=d,T_i=t,M_i=m,S_i=s\}}{\pi_i^C(d,t,m,s)}\big[Y_i-\mu_i^C(d,t,m,s)\big].
\]
Let $\widehat{\mu}_i^C(d,t,m,s)$ and $\widehat{\pi}_i^C(d,t,m,s)$ be estimators of $\mu_i^C(d,t,m,s)$ and $\pi_i^C(d,t,m,s)$. For the estimation of the OCDE we construct the following estimator of the individual-level AIPW score:
\[
\widehat{\tau}_i^C(d,t,m,s) \coloneq \widehat{\mu}_i^C(d,t,m,s) + \frac{\1\{D_i=d,T_i=t,M_i=m,S_i=s\}}{\widehat{\pi}_i^C(d,t,m,s)}\big[Y_i-\widehat{\mu}_i^C(d,t,m,s)\big].
\]
The score $\widehat{\tau}_i^C(d,t,m,s)$ is doubly robust: as long as either $\widehat{\mu}_i^C(d,t,m,s)=\mu_i^C(d,t,m,s)$ or\\$\widehat{\pi}_i^C(d,t,m,s)=\pi_i^C(d,t,m,s)$, the conditional expectation of $\widehat{\tau}_i^C(d,t,m,s)$ equals $\mu_i^C(d,t,m,s)$.

In view of the relation between $\mu_i^C(d_i,t,m_i,s)$ and $OCDE^{obs}(t,m,s)$, the AIPW estimator of $OCDE^{obs}(t,m,s)$ is naturally constructed as
\[
\widehat{OCDE}(t,m,s) \coloneq \frac{1}{j_n}\sum_{i\in\mathcal{J}_n}\big[\widehat{\tau}_i(1,t,m,s) - \widehat{\tau}_i(0,t,m,s)\big].
\]
When Theorem \ref{ch5:thm:own-cde-identification} applies, this is the AIPW estimator of $OCDE(t,m,s)$.

\textit{AIPW estimators of the ONDE and the ONIE.}\quad We refer to $\mu_i^N(d_i,d_i^*,t)$ in (\ref{ch5:eq:own-natural-effects-observed-functional}) as an outcome regression nuisance function, and define a second outcome regression nuisance function
\[
\eta_i^N(d,t,m) \coloneq \E(Y_i \mid D_i=d,T_i=t,M_i=m,\bX,\bA).
\]
In addition we define three propensity score nuisance functions,
\begin{gather*}
q_i^N(m \mid d) \coloneq \P(M_i=m \mid D_i=d,\bX,\bA),\quad e_i^N(d,t) \coloneq \P(D_i=d,T_i=t \mid \bX,\bA),\\
\text{and } p_i^N(d) \coloneq \P(D_i=d \mid \bX,\bA).
\end{gather*}
It is readily verified from these definitions that $\mu_i^N(d_i,d_i^*,t)$ and $\eta_i^N(d,t,m)$ satisfy
\[
\mu_i^N(d,d^*,t) = \E[\eta_i^N(d,t,M_i) \mid D_i=d^*,\bX,\bA]. 
\]

Based on these nuisance functions, for any $(d,d^*,t)$ we define the individual-level AIPW score
\[
\begin{aligned}
  \tau_i^N(d,d^*,t) \coloneq & \mu_i^N(d,d^*,t) + \frac{\1\{D_i=d^*\}}{p_i^N(d^*)}\big[\eta_i^N(d,t,M_i) - \mu_i^N(d,d^*,t)\big] \\
  & + \frac{\1\{D_i=d,T_i=t\}}{e_i^N(d,t)}\frac{q_i^N(M_i \mid d^*)}{q_i^N(M_i \mid d)}\big[Y_i - \eta_i^N(d,t,M_i)\big].
\end{aligned}
\]
Let $\widehat{\mu}_i^N(d,d^*,t)$, $\widehat{\eta}_i^N(d,t,m)$, $\widehat{q}_i^N(m \mid d)$, $\widehat{e}_i^N(d,t)$ and $\widehat{p}_i^N(d)$ be estimators of $\mu_i^N(d,d^*,t)$, $\eta_i^N(d,t,m)$, $q_i^N(m \mid d)$, $e_i^N(d,t)$ and $p_i^N(d)$, respectively. For the estimation of the ONDE and the ONIE we construct the following estimator of the individual-level AIPW score:
\[
\begin{aligned}
  \widehat{\tau}_i^N(d,d^*,t) \coloneq & \widehat{\mu}_i^N(d,d^*,t) + \frac{\1\{D_i=d^*\}}{\widehat{p}_i^N(d^*)}\big[\widehat{\eta}_i^N(d,t,M_i) - \widehat{\mu}_i^N(d,d^*,t)\big] \\
  & + \frac{\1\{D_i=d,T_i=t\}}{\widehat{e}_i^N(d,t)}\frac{\widehat{q}_i^N(M_i \mid d^*)}{\widehat{q}_i^N(M_i \mid d)}\big[Y_i - \widehat{\eta}_i^N(d,t,M_i)\big].
\end{aligned}
\]
The score $\widehat{\tau}_i^N(d,d^*,t)$ is multiply robust. Under our model and Assumption \ref{ch5:ass:error-term-independence} we have $M_i \indep T_i \mid D_i, \bX,\bA$, from which it is not hard to show that the conditional expectation of $\widehat{\tau}_i^N(d,d^*,t)$ equals $\mu_i^N(d,d^*,t)$ whenever $\widehat{\mu}_i^N(d,d^*,t)$ and $\widehat{\eta}_i^N(d,t,m)$ equal their true values, or $\widehat{q}_i^N(m \mid d)$, $\widehat{e}_i^N(d,t)$ and $\widehat{p}_i^N(d)$ equal their true values, or $\widehat{\eta}_i^N(d,t,m)$ and $\widehat{p}_i^N(d)$ equal their true values.

In view of the relation between $\mu_i^N(d,d^*,t)$ and $ONDE^{obs}(t)$ and $ONIE^{obs}(t)$, we construct the estimators
\begin{gather*}
  \widehat{ONDE}(t) \coloneq \frac{1}{j_n}\sum_{i\in\mathcal{J}_n}\big[\widehat{\tau}_i^N(1,0,t) - \widehat{\tau}_i^N(0,0,t)\big], \\
  \text{and } \widehat{ONIE}(t) \coloneq \frac{1}{j_n}\sum_{i\in\mathcal{J}_n}\big[\widehat{\tau}_i^N(1,1,t) - \widehat{\tau}_i^N(1,0,t)\big].
\end{gather*}
When Theorem \ref{ch5:thm:own-natural-effects-identification} applies, these are the AIPW estimators of $ONDE(t)$ and $ONIE(t)$.

\subsection{Graph Neural Network Estimators of the Nuisance Functions}
\label{ch5:subsec:gnn-nuisance-estimation}
The AIPW estimators of Section \ref{ch5:subsec:network-aipw-estimation} involve a number of nuisance functions, and estimating them accurately is a prerequisite for the consistency of the estimators. Parametric models such as linear or logistic regression are structurally too simple and tend to perform poorly on complex network data; commonly used machine learning methods such as multilayer perceptrons and random forests are more flexible in their specification, but cannot take the network structure as an input, so that the user must first summarize the network into features such as the average covariate of first-order neighbors, the average covariate of second-order neighbors, the number of first-order neighbors, and so on. If the true network effect does not act through these particular forms, such machine learning methods will likewise fail to deliver accurate estimates.

Since network data are a canonical instance of graph-structured data, we estimate all nuisance functions with graph neural networks (GNNs), which are designed for exactly this kind of data. A graph neural network is a nonparametric learning method that takes node features and the adjacency structure as direct inputs and, through a layerwise message passing mechanism, aggregates the information of neighbors and of higher-order neighbors into the representation of the target node; it therefore learns complex nonlinear network dependence automatically, without the form of the network features being specified in advance. This property is particularly important here, since the outcome regressions and the propensity score nuisance functions may all be influenced simultaneously and in complicated nonlinear ways by an individual's own features, the features of neighbors and higher-order network structure.

In detail, a standard graph neural network consists of several layers of neurons, each neuron being a vector-valued function, with the output of one layer serving as the input of the next. Write $h_i^{(l)}$ for the $i$th neuron in layer $l$. Layer $0$ is the input layer, which has $n$ neurons with $h_i^{(0)}=X_i$; layers $1$ through $L$ are hidden layers, each with $n$ neurons whose outputs are $H$-dimensional vectors, so that $h_i^{(l)}\in\mathbb{R}^H$. Their connection to the preceding layer is determined by the message passing architecture
\[
h_i^{(l)} = \Phi_{0l}\Big(h_i^{(l-1)},\Phi_{1l}\big(h_i^{(l-1)},\{h_j^{(l-1)}: A_{ij}=1,j\in\mathcal{N}_n\}\big)\Big),
\]
where $\Phi_{1l}(\cdot)$ aggregates the information of first-order neighbors and allows an individual's own information to interact with that of the neighbors, and $\Phi_{0l}(\cdot)$ maps the own information together with the aggregated neighbor information into the $H$-dimensional output of that neuron. Layer $L+1$ is the output layer, which maps the output of layer $L$ into an $O$-dimensional vector $\Phi_O(h_i^{(L)})$ through a parameterized function $\Phi_O:\mathbb{R}^H \mapsto \mathbb{R}^O$, this being the predicted value of the dependent variable for individual $i$: if the dependent variable is the outcome then $O=1$, whereas if it is a probability then $O$ is the number of discrete probability classes (for instance $O=16$ when $\widehat{\pi}_i^C(d,t,m,s)$ is fitted and $d,t,m,s$ are all binary).

Since each hidden layer performs one round of message passing, a graph neural network with $L$ hidden layers can transmit the information of an individual to its $L$th-order neighbors, so that $L$ determines the range of connected individuals whose information is used in prediction. A small $L$ exploits only very local information, which is appropriate when the influence of higher-order neighbors is weak; a large $L$ exploits the information of higher-order neighbors but markedly increases the complexity of the network and the number of parameters. In practice the choice of $L$ must be weighed against subject-matter knowledge or predictive performance.

The choice of $\Phi_{0l}(\cdot)$ and $\Phi_{1l}(\cdot)$ determines the predictive capacity of the graph neural network. This is especially true of the latter, which determines whether the network can capture the information of neighbors accurately and hence the quality of message passing. A detailed account and comparison of the possible choices of $\Phi_{0l}(\cdot)$ and $\Phi_{1l}(\cdot)$ is beyond the scope of this paper; we describe only the choice used in our simulations and empirical application. We adopt the principal neighborhood aggregation architecture of \citet{corso2020principal}, in which $\Phi_{0l}(\cdot)$ is ordinarily taken to be a multilayer perceptron and $\Phi_{1l}(\cdot)$ to be the composite function
\[
\Phi_{1l}(\cdot) = \Gamma\big(\{\phi_l(h_i^{(l-1)},h_j^{(l-1)}):A_{ij}=1,j\in\mathcal{N}_n\}\big).
\]
Here $\phi_l(\cdot)$ is a vector-valued function, ordinarily taken to be a multilayer perceptron, and $\Gamma(\cdot)=\big(S(\cdot,0),S(\cdot,1),\allowbreak S(\cdot,-1)\big)\otimes\big(\text{mean}(\cdot), \text{sd}(\cdot), \text{sum}(\cdot), \text{min}(\cdot), \text{max}(\cdot)\big)$, where $\otimes$ denotes the tensor product. The scalar function $S(\cdot,\alpha)$ is defined as
\[
S(\cdot,\alpha) = \bigg(\frac{\log(|\cdot|+1)}{\delta}\bigg)^\alpha, \quad \delta = \frac{1}{n}\sum_{i=1}^n \log\bigg(\sum_{j=1}^n A_{ij}+1\bigg),
\]
where $|\cdot|$ returns the number of elements of the input set. The operators $\text{mean}(\cdot)$, $\text{sd}(\cdot)$, $\text{sum}(\cdot)$, $\text{min}(\cdot)$, $\text{max}(\cdot)$ take a set of vectors and compute, componentwise, the mean, standard deviation, sum, minimum and maximum. Consequently, if $\phi_l(\cdot)$ returns an $H$-dimensional vector, then $\big(\text{mean}(\cdot), \text{sd}(\cdot), \text{sum}(\cdot), \text{min}(\cdot), \text{max}(\cdot)\big)$ is a $5H$-dimensional vector and $\Gamma(\cdot)$ is a $15H$-dimensional vector.

We now describe the fitting procedure for the graph neural network estimators, taking $\E(Y_i \mid D_i=d_i, T_i=t, M_i=m_i, S_i=s, \bX, \bA)$ and $\P(D_i=d, T_i=t, M_i=m, S_i=s \mid \bX, \bA)$ as examples. The fitting procedures for the remaining outcome regression and propensity score nuisance functions are entirely analogous and are omitted for brevity.

Let $\theta$ collect all parameters of the chosen graph neural network architecture and let $\widehat{\theta}$ denote their estimated values. Write $\mu(i,D_i,T_i,\allowbreak M_i,S_i,\bX,\bA;\widehat{\theta})$ for the value of the outcome regression above predicted for individual $i$ by the graph neural network estimator, that is, $\mu(i,D_i,T_i,M_i,S_i,\bX,\bA;\widehat{\theta})=\Phi_O(h_i^{(L)})$, where $\widehat{\theta}$ is obtained by minimizing the mean squared error loss:
\[
\widehat{\theta} \in \mathop{\arg\min}_{\theta} \frac{1}{j_n}\sum_{i\in\mathcal{J}_n} \big(Y_i - \mu(i,D_i,T_i,M_i,S_i,\bX,\bA;\theta)\big)^2.
\]

For the propensity score we encode the joint value of the four binary variables $(D_i,T_i,M_i,\allowbreak S_i)$ as a $16$-class label, $C_i = 8D_i + 4T_i + 2M_i + S_i \in \{0,1,\ldots,15\}$. Write $\pi(i,\bX,\bA;\widehat{\theta}) = \big(\pi_0(i,\bX,\bA;\widehat{\theta}),\ldots,\pi_{15}(i,\bX,\bA;\widehat{\theta})\big)$ for the logits of the propensity score predicted for individual $i$ by the graph neural network, that is, $\pi(i,\bX,\bA;\widehat{\theta}) = \Phi_O(h_i^{(L)})$. Applying the softmax transformation, the graph neural network estimator of $\P(C_i=c \mid \bX,\bA)$ is
\[
\frac{\exp\{\pi_{c}(i,\bX,\bA;\widehat{\theta})\}}{\sum_{k=0}^{15} \exp\{\pi_{k}(i,\bX,\bA;\widehat{\theta})\}}.
\]
where $\widehat{\theta}$ is obtained by minimizing the multiclass cross-entropy loss:
\[
\widehat{\theta} \in \mathop{\arg\min}_\theta -\frac{1}{j_n}\sum_{i\in \mathcal{J}_n}\log \Bigg(\frac{\exp\{\pi_{C_i}(i,\bX,\bA;\theta)\}}{\sum_{k=0}^{15} \exp\{\pi_{k}(i,\bX,\bA;\theta)\}}\Bigg).
\]

Finally, we note that a graph neural network designed according to the architecture above is insensitive to the labelling of individuals: interchanging the labels of any two individuals leaves the estimates for all individuals unchanged, a property known as permutation invariance. In most observational network studies the labels of individuals serve no purpose beyond identifying them, and the model specified in Section \ref{ch5:subsec:network-mediation-model} enjoys the same property, so that the graph neural networks of this section are an appropriate choice of estimator for the nuisance functions. See \citet{leung2022graph} for further discussion of permutation invariance and its implications.

\section{Asymptotic Theory}
\label{ch5:sec:asymptotic-theory}
In this section we study the asymptotic properties, as $n\to \infty$, of the AIPW estimators proposed in Section \ref{ch5:subsec:network-aipw-estimation}. Note that as $n$ tends to infinity, the network structure, the distributions of the random variables, the models for the outcome, the mediator and the treatment, the GNN functions and the number of layers $L$ are all allowed to change. Moreover, as noted in Section \ref{ch5:subsec:network-mediation-model}, our asymptotic theory is conditional on $(\bX,\bA)$, which avoids imposing additional restrictions on the network dependence structure of $(\bX,\bA)$.

We begin with the asymptotic properties of the estimator of the own controlled direct effect at a fixed $(t,m,s)$. Since $\bX,\bA,\be$ induce correlation across all individuals, assumptions are required to limit this correlation so that a suitable central limit theorem may be applied, and the central limit theorem itself requires further technical conditions. We introduce the necessary assumptions one at a time.

\textit{Exposure mappings.}\quad The exposure mappings used to define the own controlled direct effect must satisfy the following locality condition.
\begin{assumption}
\label{ch5:ass:own-cde-local-exposure-mappings}
There exist fixed positive integers $K_T$ and $K_S$ such that the exposure mappings $T_i$ and $S_i$ are determined solely by the treatments and mediators of the individuals in the $K_T$-neighborhood and the $K_S$-neighborhood, respectively. Specifically,
\[
T_i = t_n(i,\bD_{\mathcal{N}(i,K_T)\setminus\{i\}},\bA), \quad S_i = s_n(i,\bM_{\mathcal{N}(i,K_S)\setminus\{i\}},\bA).
\]
\end{assumption}
This assumption requires the exposure mappings used to define the causal estimands to be local of fixed order: $T_i$ and $S_i$ may summarize only the treatments and mediators within the fixed $K_T$- and $K_S$-neighborhoods of individual $i$, and $K_T$ and $K_S$ do not grow with the sample size. This is not a strong assumption, as the substantively meaningful exposure mappings encountered in the literature all satisfy it. Moreover, under the approximate neighborhood interference assumption introduced below, the influence of higher-order neighbors on an individual is very weak, so that from a substantive standpoint there is little need to study the effects of neighbors of very high order; choosing $K_T$ and $K_S$ of moderate size suffices for most empirical work.

\textit{Boundedness and overlap.}\quad Next we impose the boundedness and overlap conditions standard in causal inference. Define the individual-level centered score
\begin{equation}
\varphi^C_i(t, m, s) \coloneq \big[\tau^C_i(1, t, m, s) - \tau^C_i(0, t, m, s)\big] - OCDE(t, m, s),
\end{equation}
where $\tau^C_i(d, t, m, s)$ is defined in Section \ref{ch5:subsec:network-aipw-estimation}, and put
\begin{equation}
\sigma_{n,C}^2 \coloneq \Var\bigg(\frac{1}{\sqrt{j_n}} \sum_{i \in \mathcal{J}_n} \varphi^C_i(t, m, s) \mid \bX, \bA \bigg).
\label{ch5:eq:own-cde-asymptotic-variance}
\end{equation}
\begin{assumption}
\label{ch5:ass:own-cde-moments-overlap}
(i) There exist $M < \infty$ and $p > 4$ such that $\E[|Y_i(\bd, \bm)|^p \mid \bX, \bA] < M$ almost surely for all $n \in \mathbb{N}$, $i \in \mathcal{N}_n$, $\bd \in \{0, 1\}^n$ and $\bm \in \{0, 1\}^n$.\\
(ii) There exists $[\underline{\pi}, \overline{\pi}] \subset (0, 1)$ such that, for all $n \in \mathbb{N}$, $i \in \mathcal{N}_n$ and any $(d, t, m, s)$, $\widehat{\pi}_i^C(d, t, m, s),\allowbreak \pi_i^C(d, t, m, s) \in [\underline{\pi}, \overline{\pi}]$ and $j_n / n \geqslant \underline{\pi}$ almost surely.\\
(iii) $\liminf_{n \to \infty} \sigma_{n,C}^2 > 0$ almost surely.
\end{assumption}
Part (i) requires the conditional $p$th moment of the potential outcome to exist, a standard technical condition that is ubiquitous in the double machine learning literature \citep{chernozhukov2018double,farrell2021deep,klyne2026average}. Part (ii) imposes overlap on the propensity score, the standard technical condition in the causal inference literature that employs inverse probability weighted or AIPW estimators; it also requires $\mathcal{J}_n$ to be a nontrivial subset of $\mathcal{N}_n$. Part (iii) is a standard nondegeneracy condition guaranteeing the existence of the asymptotic variance.

\textit{Approximate neighborhood interference.}\quad Since the data come from a single network, some form of restriction on network interference is needed to obtain a weakly dependent sample and thus to obtain asymptotic normality from a central limit theorem for weakly dependent data. Specifically, we establish weak dependence by limiting the strength with which the structural functions $(f_n, g_n, h_n)$ depend on distant variables. For any $S \subseteq \mathcal{N}_n$, let $\bD_S$ be the vector of treatments of the individuals in $S$, that is, $\bD_S \coloneq (D_i)_{i \in S}$, and define $\bM_S, \bX_S, \be_S$ analogously.
\begin{assumption}
\label{ch5:ass:approximate-neighborhood-interference}
There exists a sequence of functions $\{(\gamma_n)\}_{n \in \mathbb{N}}$ with $\gamma_n: \mathbb{R}_+ \to \mathbb{R}_+$ satisfying $\sup_{n \in \mathbb{N}} \gamma_n(r) \xrightarrow{r \to \infty} 0$, and for every $n \in \mathbb{N}$, there is a model using local information only, $
  f_{n,r}(i, \bD_{\mathcal{N}(i,r)}, \bM_{\mathcal{N}(i,r)}, \bX_{\mathcal{N}(i,r)}, \bA_{\mathcal{N}(i,r)},\be_{\mathcal{N}(i,r)})
$, such that
\begin{equation}
  \begin{aligned}
    \max_{i\in\mathcal{N}_n}\E[ & |f_n(i, \bD, \bM, \bX, \bA, \be) \\
  & - f_{n,r}(i, \bD_{\mathcal{N}(i,r)}, \bM_{\mathcal{N}(i,r)}, \bX_{\mathcal{N}(i,r)}, \bA_{\mathcal{N}(i,r)},\be_{\mathcal{N}(i,r)})| \mid \bD,\bM,\bX,\bA] \leqslant \gamma_n(r).
  \end{aligned} \label{ch5:eq:local-outcome-approximation-bound} 
\end{equation}
\end{assumption}
This is the approximate neighborhood interference assumption, introduced by \citet{leung2022causal} and used in \citet{leung2022graph,hoshino2024causal} among others; we employ its extension to the mediation setting. Its essential requirement is that when the true outcome model $f_n$ is approximated by the local $r$-neighborhood outcome model $f_{n,r}$, the approximation error is bounded by $\gamma_n(r)$, which decays as the neighborhood radius grows. The assumption formalizes a plausible network setting in which the outcome of individual $i$ is determined chiefly by the individuals in its local neighborhood, the influence on $i$ weakening the farther away an individual is. In such a setting it is natural to expect $f_{n,r}$ to approximate $f_n$ well once $r$ is large enough, so that the requirement $\sup_{n \in \mathbb{N}} \gamma_n(r) \xrightarrow{r \to \infty} 0$ is also reasonable.

\textit{Convergence rate conditions on nuisance estimation.}\quad Since the nuisance functions in the AIPW estimator are estimated by graph neural networks, requirements on the quality of these estimators are needed for the AIPW estimator to behave well. Specifically, we require the graph neural network estimators to attain the convergence rates demanded by the next assumption. Define $Z_i \coloneq (Y_i, D_i, T_i, M_i, S_i, \bX, \bA)$ and, for any $x \in \mathbb{R}$,
\begin{gather*}
  \begin{aligned}
    \Psi_{\pi^C}(Z_i, x, d, t, m, s) \coloneq & \1\{D_i = d, T_i = t, M_i = m, S_i = s\} \\
    & [Y_i - \mu_i^C(d, t, m, s)]\left[\frac{1}{x} - \frac{1}{\pi_i^C(d, t, m, s)}\right], 
  \end{aligned}\\
    \Psi_{\mu^C}(Z_i, x, d, t, m, s) \coloneq [x - \mu_i^C(d, t, m, s)] \left[1 - \frac{\mathbf{1}\{D_i = d, T_i = t, M_i = m, S_i = s\}}{\pi_i^C(d, t, m, s)}\right].
\end{gather*}
It is readily verified from these definitions that, for any $x$,
\[
\E[\Psi_{\pi^C}(Z_i, x, d, t, m, s) \mid \bX,\bA] = 0, \quad \E[\Psi_{\mu^C}(Z_i, x, d, t, m, s) \mid \bX,\bA] = 0.
\]
These identities are used in the proofs of the theorems and also underpin the plausibility of the stochastic equicontinuity part of the next assumption.
\begin{assumption}
\label{ch5:ass:own-cde-gnn-convergence-rates}
For any $(d,t,m,s)$, the following conditions hold.\\
(i) (Mean squared error)
\begin{gather*}
  \frac{1}{j_n} \sum_{i \in \mathcal{J}_n} [\widehat\pi_i^C(d, t, m, s) - \pi_i^C(d, t, m, s)]^2 = o_p(n^{-1/2}), \\
  \frac{1}{j_n} \sum_{i \in \mathcal{J}_n} [\widehat{\mu}_i^C(d, t, m, s) - \mu_i^C(d, t, m, s)]^2 = o_p(n^{-1/2}).
\end{gather*}
(ii) (Stochastic equicontinuity)
\begin{gather*}
\frac{1}{j_n} \sum_{i \in \mathcal{J}_n} \Psi_{\pi^C}(Z_i, \widehat\pi_i^C(d, t, m, s), d, t, m, s) = o_p(n^{-1/2}), \\
\frac{1}{j_n} \sum_{i \in \mathcal{J}_n} \Psi_{\mu^C}(Z_i, \widehat{\mu}_i^C(d, t, m, s), d, t, m, s) = o_p(n^{-1/2}).
\end{gather*}
\end{assumption}
Part (i) imposes an empirical $L^2$ rate of $o_p(n^{-1/4})$ on each nuisance estimator. For asymptotic normality alone, the componentwise conditions may be replaced by the weaker requirement that the product of the two mean squared errors be $o_p(n^{-1})$, thereby allowing asymmetric convergence rates. We retain the stronger componentwise conditions because they are used to build a consistent estimator for the asymptotic variance. These rate requirements coincide with the standard ones in the double machine learning literature \citep[e.g.,][]{chernozhukov2018double,chernozhukov2022locally}. Part (ii) is a stochastic equicontinuity condition, used to control the first-order remainder in the derivation of asymptotic normality.

Assumption \ref{ch5:ass:own-cde-gnn-convergence-rates} is a high-level condition on the fitted nuisance functions and does not rely on architecture-specific properties of the nuisance function learner. Accordingly, the asymptotic conclusions apply to the GNN implementation considered here, as well as to alternative nuisance learners, provided that these conditions are satisfied. Related results in \citet{leung2022graph} provide further support for the validity of these high-level conditions. In particular, under some regularity conditions, they show that the nuisance functions admit sufficiently accurate low-dimensional approximations, thereby facilitating verification of the MSE rate requirement for the first-stage estimation, and provide routes to establish stochastic equicontinuity under additional restrictions. 

\textit{Technical conditions for $\psi$-dependence.}\quad The approximate neighborhood interference assumption restricts the strength of interference in the network at the level of causal structure, which is weak dependence at the counterfactual level. To apply a central limit theorem we also need a weak dependence restriction at the level of the statistical distribution. This restriction is purely technical and serves mainly to establish that $\{\varphi^C_i\}_{i \in \mathcal{N}_n}$ is $\psi$-dependent \citep{kojevnikov2021limit}. Define the $r$-neighborhood boundary of $i$ and the $k$th moment of the population $r$-neighborhood boundary as
\[
\mathcal{N}^B(i, r) \coloneq \{j \in \mathcal{N}_n : \ell_{\bA}(i, j) = r\} \quad \text{and} \quad \delta_n^B(r, k) \coloneq \frac{1}{n} \sum_{i=1}^n |\mathcal{N}^B(i, r)|^k.
\]
We further define the quantities
\begin{gather}
  \Delta_n(r, m, k) \coloneq \frac{1}{n} \sum_{i=1}^n \max_{j \in \mathcal{N}^B(i, r)} |\mathcal{N}(i, m) \setminus \mathcal{N}(j, r-1)|^k, \notag \\
  c_n(r, m, k) \coloneq \inf_{\alpha > 1} \Delta_n(r, m, k\alpha)^{1/\alpha}\delta_n^B(r, \alpha / (\alpha - 1))^{1 - 1/\alpha}, \notag \\
  K_0 \coloneq \max\{K_T,K_S\}, \quad \psi_n(r) \coloneq \begin{cases}
    1, & 0 \leqslant r \leqslant 4K_0, \\
    \gamma_n(r/4), & r > 4K_0.
  \end{cases} \label{ch5:eq:own-cde-psi-dependence-coefficient}
\end{gather}
Here $\Delta_n(r, m, k)$ measures, on average, the largest expansion that the $m$-neighborhood of a node can produce relative to the $(s-1)$-neighborhood of any node on the boundary of its $s$-neighborhood, and $c_n(r, m, k)$ may be regarded as a measure of the density of the $r$-neighborhood. The quantity $\psi_n(r)$ measures network dependence: it bounds the covariance between $\varphi^C_i$ and $\varphi^C_j$ when $\ell_{\bA}(i,j) \leqslant r$.
\begin{assumption}
\label{ch5:ass:own-cde-weak-dependence}
(i) $\sup_{n \in \mathbb{N}} \max_{r \geqslant 1} \psi_n(r) < \infty$ almost surely. (ii) For the $p$ in Assumption \ref{ch5:ass:own-cde-moments-overlap}(i), for some sequence of positive numbers $v_n \to \infty$ and any $k \in \{1, 2\}$,
\begin{equation}
\frac{1}{n^{k/2}} \sum_{r=0}^{\infty}c_n(r, v_n, k) \psi_n(r)^{1 - (2+k)/p} = o(1) \quad \text{and} \quad n^{3/2} \psi_n(v_n)^{1 - 1/p} = o(1) \quad \text{a.s.}
\label{ch5:eq:network-clt-rate-condition}
\end{equation}
\end{assumption}
Part (i) is a technical condition requiring the $\psi$-dependence coefficients to be uniformly bounded. The essential restriction is part (ii), which corresponds to condition ND of \citet{kojevnikov2021limit}. Condition (\ref{ch5:eq:network-clt-rate-condition}) requires the covariance to decay with distance (as governed by $\psi_n(r)$) faster than network neighborhoods densify (as governed by $c_n(r, v_n, k)$), and is the key condition for the network central limit theorem. \citet{leung2022graph} verifies that (\ref{ch5:eq:network-clt-rate-condition}) holds when the neighborhood size grows with the neighborhood radius at a polynomial or an exponential rate.

\textit{Asymptotic normality.}\quad Under the assumptions above we obtain the asymptotic normality of the AIPW estimator of the OCDE.
\begin{theorem}
\label{ch5:thm:own-cde-asymptotic-normality}
Under Assumptions \ref{ch5:ass:consistency}, \ref{ch5:ass:error-term-independence}, \ref{ch5:ass:own-cde-conditional-independence} and \ref{ch5:ass:own-cde-local-exposure-mappings}--\ref{ch5:ass:own-cde-weak-dependence}, for any
$(t, m, s)$,
\[
\sigma_{n,C}^{-1} \sqrt{j_n} \big(\widehat{OCDE}(t, m, s) - OCDE(t, m, s)\big) \xrightarrow{d} \mathcal{N}(0, 1).
\]
\end{theorem}

\textit{Variance estimator.}\quad Now we construct a consistent estimator of the asymptotic variance $\sigma_{n,C}^2$. Define the bandwidth
\begin{equation}
b_n \coloneq \lceil \widetilde{b}_n \rceil, \quad
\widetilde{b}_n \coloneq \begin{cases}
\dfrac{1}{4} \mathcal{L}(\bA) & \text{if } \mathcal{L}(\bA) < 2\dfrac{\log n}{\log \delta(\bA)}, \\
\mathcal{L}(\bA)^{1/4} & \text{otherwise},
\end{cases}
\label{ch5:eq:network-hac-bandwidth}
\end{equation}
where $\delta(\bA) \coloneq n^{-1} \sum_{i=1}^n n(i, 1)$ is the average degree and $\mathcal{L}(\bA)$ is the average path length on the largest connected subgraph, that is, the average distance over all pairs of individuals in that subgraph. Define the $j_n \times j_n$ kernel matrices $K_U$ and $K_{PD}$ with $(i, j)$ entries $\mathbf{1}\{\ell_{\bA}(i, j) \leqslant b_n\}$ and
\begin{equation}
\frac{|\mathcal{N}(i, b_n / 2) \cap \mathcal{N}(j, b_n / 2)|}{|\mathcal{N}(i, b_n / 2)|^{1/2}\, |\mathcal{N}(j, b_n / 2)|^{1/2}},
\label{ch5:eq:positive-definite-hac-kernel}
\end{equation}
respectively, where $i, j \in \mathcal{J}_n$. Define the residual vector
\begin{equation*}
\widetilde{\tau} \coloneq \big( \widehat{\tau}_i^C(1, t, m, s) - \widehat{\tau}_i^C(0, t, m, s) - \widehat{\mu}_i^C(1, t, m, s) + \widehat{\mu}_i^C(0, t, m, s) \big)_{i \in \mathcal{J}_n},
\end{equation*}
and the HAC variance estimator
\begin{equation}
\widehat{\sigma}_{C}^2 \coloneq \max\{\widehat{\sigma}_U^2, \widehat{\sigma}_{PD}^2\}, \quad \text{where} \quad \widehat{\sigma}_q^2 \coloneq \frac{1}{j_n}\, \widetilde{\tau}' K_q\, \widetilde{\tau}, \quad q \in \{U, PD\}.
\label{ch5:eq:own-cde-hac-variance-estimator}
\end{equation}
Our choice of bandwidth follows \citet{leung2022graph} and is similar to that of \citet{leung2022causal}. The bandwidth adapts to the size and density of the network, and the constant in it takes account of the estimation error of the graph neural network estimators. Here $\widehat{\sigma}_U^2$ is the network HAC estimator based on the uniform kernel $K_U$, the variance estimator developed by \citet{kojevnikov2021limit} to accompany their network central limit theorem, while $\widehat{\sigma}_{PD}^2$ is the network HAC estimator based on the positive definite kernel $K_{PD}$, which avoids the possibility that $\widehat{\sigma}_U^2$ takes a negative value. Taking the larger of $\widehat{\sigma}_U^2$ and $\widehat{\sigma}_{PD}^2$ yields a somewhat conservative network HAC variance estimator that is necessarily positive.

We impose the following assumption on the network in order to obtain a consistent estimator of the asymptotic variance.
\begin{assumption}
\label{ch5:ass:own-cde-hac-consistency}
(i) There exists $M > 0$ such that, for all $n \in \mathbb{N}$, $i \in N_n$ and any $(d, t, m, s)$, $\max\{|Y_i|, |\widehat{\mu}_i^C(d, t, m, s)|\} < M$ almost surely.\\
(ii) For the $p$ of Assumption \ref{ch5:ass:own-cde-moments-overlap}(i) and some positive sequence $b_n \to \infty$, \\$\displaystyle n^{-1} \sum_{r=0}^{\infty} c_n(r, b_n, 2)\, \psi_n(r)^{1 - 4/p} = o(1)$, $\sum_{r>b_n}\delta_n^B(r,1)\psi_n(r)^{1-2/p}=o(1)$ and \\$\sum_{r=0}^{n}\psi_n(r)^{1-2/p} n^{-1}\sum_{i=1}^n\sum_{j\in\mathcal{N}^B(i,r)}|K_{PD,ij}-1| = o(1)$ almost surely.\\
(iii) $\displaystyle n^{-1} \sum_{i=1}^n n(i, b_n)^2 = O_p(\sqrt{n})$.
\end{assumption}
Part (i) is a uniform boundedness condition requiring the outcome and the graph neural network estimators to be bounded, which is stronger than the existence of the $p$th moment in Assumption \ref{ch5:ass:own-cde-moments-overlap}. Part (ii) is essentially a balance condition among the bandwidth, the density of the network and the rate at which dependence decays: at the chosen bandwidth $b_n$, the local correlation retained by the kernel matrix must not accumulate into nonvanishing random error through excessively fast neighborhood growth, while the long-range correlation truncated or understated by the kernel matrix must be negligible in the aggregate. Its first two conditions are taken from Assumption 4.1 of \citet{kojevnikov2021limit} and are used to prove the consistency of the uniform kernel HAC estimator; its third condition is used to prove the consistency of the positive definite kernel HAC estimator. Part (iii) restricts the local density of the network: by limiting the growth of the $b_n$-neighborhood size in the mean square sense, it ensures that when the true residuals are replaced by the estimated ones, the estimation error of the nuisance functions is not amplified excessively by an overly large neighborhood and thereby destroys consistency. \citet{leung2022graph} shows that the conditions in parts (ii) and (iii) hold when the neighborhood size grows with the neighborhood radius at a polynomial or an exponential rate.
\begin{theorem}
\label{ch5:thm:own-cde-variance-consistency}
Under Assumption \ref{ch5:ass:own-cde-hac-consistency} and the assumptions of Theorem \ref{ch5:thm:own-cde-asymptotic-normality}, $\widehat{\sigma}_{C}^2 = \sigma_{n,C}^2 + o_p(1)$ for any fixed $(t,m,s)$.
\end{theorem}

The asymptotic theory for the own natural direct and indirect effects can be analogously derived. For simplicity, we relegate all the details to Appendix A. 

\section{Simulation Study}
\label{ch5:sec:simulation-study}
We design a series of simulation experiments to compare the estimation method proposed here with conventional methods in a variety of settings. Our aim is to show that, in the presence of complex network spillover, our method outperforms conventional machine learning methods in estimating the various effects across a range of settings, in terms of bias, standard deviation and confidence interval coverage. 

\subsection{Simulation Design}
We report results for $n=1000$, $n=2500$ and $n=4000$, and consider two mechanisms for generating the adjacency matrix $\bA$: the random geometric graph model, which sets $A_{ij} = \1\{\|\rho_i - \rho_j\| \leqslant r_n\}, i \neq j$, where $\{\rho_i\}_{i\in\mathcal{N}_n} \stackrel{i.i.d.}{\sim} Unif([0,1]^2)$ and $r_n = (8/(\pi n))^{1/2}$; and the Erd\H{o}s--R\'enyi model, which sets $A_{ij} = A_{ji} \stackrel{i.i.d.}{\sim} Bern(8/n), i < j$. Both have limiting average degree $8$, but the average path length grows polynomially in $n$ under the first model and is of the order of $\log(n)$ under the second; at $n=2500$ it is about $23.8$ and about $4.0$, respectively.

The covariates are drawn independently as $X_i = (X_{i1},X_{i2},\allowbreak X_{i3})' \stackrel{i.i.d.}{\sim} \mathcal{N}(\mathbf{0},\mathbf{I}_3)$ and the error terms as $\{\varepsilon_i\}_{i\in\mathcal{N}_n}$, $\{\nu_i\}_{i\in\mathcal{N}_n}$ and $\{\omega_i\}_{i\in\mathcal{N}_n} \stackrel{i.i.d.}{\sim} \mathcal{N}(0,1)$, all of these variables being mutually independent.

For $X_i$ and $W_i \in \{D_i, M_i, \varepsilon_i\}$ we define the neighborhood averages
\[
\overline{X}_i \coloneq \begin{cases}
  \dfrac{\sum_{j\in\mathcal{N}(i,1)} X_j}{n(i,1)}, & n(i,1)>0, \\
  (0,0,0)' & n(i,1)=0,
\end{cases}
\quad \text{and} \quad 
\overline{W}_i \coloneq \begin{cases}
  \dfrac{\sum_{j\in\mathcal{N}(i,2)} W_j}{n(i,2)}, & n(i,2)>0, \\
  0 & n(i,2)=0.
\end{cases}
\]
and then define the second-order nonlinear network confounder at node $i$,
\[
Z_i \coloneq \frac{1}{n(i,1)}\sum_{j\in\mathcal{N}(i,1)} \tanh\big((2, 2, 2)\overline{X}_j \big).
\]
The treatment, the mediator and the outcome are generated by the following mechanism, which contains complex nonlinear network spillover:
\begin{gather*}
  D_i = \1 \big\{(0.3, 0.3, 0.3)X_i + (0.3, 0.3, 0.3)\overline{X}_i + Z_i + \nu_i > 0\big\}, \\
  M_i = \1 \big\{-0.5 + D_i + (0.3, 0.3, 0.3)X_i + (0.3, 0.3, 0.3)\overline{X}_i + Z_i + \omega_i > 0\big\}, \\
  Y_i = 1 + D_i + \overline{D}_i + M_i + \overline{M}_i + (0.3, 0.3, 0.3)X_i + (0.3, 0.3, 0.3)\overline{X}_i + Z_i + \varepsilon_i + \overline{\varepsilon}_i.
\end{gather*}

The exposure mappings used to define the causal estimands are
\[
T_i = \1 \bigg\{\frac{\sum_{j=1}^n A_{ij}D_j}{\sum_{j=1}^n A_{ij}} \geqslant \frac{1}{2}\bigg\} \ \text{and} \ S_i = \1 \bigg\{\frac{\sum_{j=1}^n A_{ij}M_j}{\sum_{j=1}^n A_{ij}} \geqslant \frac{1}{2}\bigg\},
\]
with $T_i$ and $S_i$ set to $0$ when $\sum_{j=1}^n A_{ij} = 0$.
The subpopulation used for comparison is $\mathcal{J}_n = \{i: 2\leqslant n(i,1) \leqslant 25\}$, the set of individuals with degree between $2$ and $25$.
We estimate the own controlled direct effect $OCDE(0,0,0)$ and the own natural direct effect $ONDE(0)$, whose true values are both equal to $1$ under this data generating mechanism.

For all nuisance functions we consider three estimation methods. The first is the graph neural network recommended here, built with the PNA architecture and the loss functions in Section \ref{ch5:subsec:gnn-nuisance-estimation} through the \texttt{PNAConv} class of the PyTorch Geometric package, with hidden dimension $H=16$, $\phi_l(\cdot)$ a linear layer without activation, $\Phi_{0l}(\cdot)$ a linear layer with ReLU activation, and $\Phi_O(\cdot)$ a two-layer perceptron, linear with ReLU activation followed by linear without activation. The optimizer is Adam with learning rate $0.01$, the number of training epochs is $200$, the weight decay is $10^{-4}$, and the propensity scores are clipped from below at $0.02$. Following the discussion in \citet{leung2022graph}, the number of hidden layers should not be too large, and we report the three small values $L=1,2,3$; since an individual is affected at most by its second-order neighbors here, these correspond to a range of information aggregation that is too small, correctly specified and too large, respectively.

The second and third methods are a multilayer perceptron with the same number of layers and the same hidden dimension as the graph neural network, and a random forest with $500$ trees and a minimum leaf size of $10$. Both use as features an individual's own covariates and the average covariates of first-order neighbors. As conventional machine learning methods, they cannot handle spillover well, and we use them as baselines against which to assess how effectively the graph neural network addresses this problem.

\subsection{Simulation Results}
Tables \ref{ch5:tab:random-geometric-main-simulation} and \ref{ch5:tab:erdos-renyi-main-simulation} summarize the results of $1000$ replications under the random geometric graph and the Erd\H{o}s--R\'enyi model, respectively. The row Est reports the average of the $1000$ point estimates; the row SE reports the average of the $1000$ network HAC standard errors; the row CI reports the coverage of the $95\%$ confidence intervals constructed from our network HAC variance estimator; the row OSE (oracle SE) reports the standard deviation of the $1000$ point estimates; and the row OCI (oracle CI) reports the coverage of the $95\%$ confidence intervals constructed from the true standard deviation in the OSE row. We note that the random forest does not require a number of hidden layers to be specified; for convenience its results are reported in the $L=1$ column, which is not to be read as implying that the random forest has one hidden layer.
\begin{table}[t!]
\small
\renewcommand{\arraystretch}{0.6}
\setlength{\tabcolsep}{3pt}
\centering
\caption{Simulation results: random geometric graph}
\label{ch5:tab:random-geometric-main-simulation}
\resizebox{\textwidth}{!}{
\begin{tabular}{lllccccccccc}
\toprule
& & & \multicolumn{3}{c}{$L=1$} & \multicolumn{3}{c}{$L=2$} & \multicolumn{3}{c}{$L=3$} \\
\cmidrule(lr){4-6} \cmidrule(lr){7-9} \cmidrule(lr){10-12}
Estimand & & $n$ & 1000 & 2500 & 4000 & 1000 & 2500 & 4000 & 1000 & 2500 & 4000 \\
\midrule

\multirow{15}{*}{OCDE} & \multirow{5}{*}{GNN} & Est & 1.0743 & 1.0970 & 1.1012 & 0.9677 & 0.9900 & 0.9951 & 0.9529 & 0.9839 & 0.9947 \\
 &  & OCI & 0.9300 & 0.8800 & 0.8300 & 0.9330 & 0.9450 & 0.9530 & 0.9380 & 0.9360 & 0.9510 \\
 &  & OSE & 0.1817 & 0.1329 & 0.1032 & 0.1365 & 0.0993 & 0.0853 & 0.1507 & 0.1072 & 0.0869 \\
 &  & CI & 0.8770 & 0.8330 & 0.7970 & 0.8040 & 0.9220 & 0.9310 & 0.7940 & 0.8970 & 0.9260 \\
 &  & SE & 0.1552 & 0.1201 & 0.0978 & 0.0981 & 0.0908 & 0.0776 & 0.1049 & 0.0895 & 0.0774 \\
\cmidrule(lr){2-12}
 & \multirow{5}{*}{MLP} & Est & 1.0943 & 1.1149 & 1.1152 & 1.0901 & 1.1107 & 1.1122 & 1.0884 & 1.1127 & 1.1104 \\
 &  & OCI & 0.9180 & 0.8640 & 0.8120 & 0.9310 & 0.8800 & 0.8230 & 0.9510 & 0.8660 & 0.8410 \\
 &  & OSE & 0.1951 & 0.1355 & 0.1054 & 0.2090 & 0.1392 & 0.1101 & 0.2104 & 0.1363 & 0.1113 \\
 &  & CI & 0.9120 & 0.8080 & 0.7720 & 0.9110 & 0.8310 & 0.7710 & 0.9150 & 0.8310 & 0.7880 \\
 &  & SE & 0.1868 & 0.1240 & 0.1002 & 0.1941 & 0.1285 & 0.1023 & 0.1960 & 0.1285 & 0.1029 \\
\cmidrule(lr){2-12}
 & \multirow{5}{*}{RF} & Est & 1.2200 & 1.2170 & 1.1975 &  &  &  &  &  &  \\
 &  & OCI & 0.8860 & 0.7290 & 0.5900 &  &  &  &  &  &  \\
 &  & OSE & 0.2903 & 0.1573 & 0.1168 &  &  &  &  &  &  \\
 &  & CI & 0.2890 & 0.1950 & 0.1310 &  &  &  &  &  &  \\
 &  & SE & 0.0778 & 0.0468 & 0.0362 &  &  &  &  &  &  \\
\midrule

\multirow{15}{*}{ONDE} & \multirow{5}{*}{GNN} & Est & 1.0970 & 1.1076 & 1.1074 & 0.9665 & 0.9894 & 0.9969 & 0.9634 & 0.9909 & 0.9956 \\
 &  & OCI & 0.9260 & 0.8820 & 0.8250 & 0.9480 & 0.9490 & 0.9560 & 0.9530 & 0.9450 & 0.9510 \\
 &  & OSE & 0.2072 & 0.1343 & 0.1023 & 0.1605 & 0.1112 & 0.0888 & 0.1796 & 0.1167 & 0.0887 \\
 &  & CI & 0.8750 & 0.8600 & 0.8020 & 0.8330 & 0.9270 & 0.9490 & 0.8390 & 0.9230 & 0.9410 \\
 &  & SE & 0.1769 & 0.1293 & 0.1040 & 0.1151 & 0.1005 & 0.0849 & 0.1346 & 0.1024 & 0.0837 \\
\cmidrule(lr){2-12}
 & \multirow{5}{*}{MLP} & Est & 1.1354 & 1.1201 & 1.1187 & 1.1333 & 1.1185 & 1.1180 & 1.1226 & 1.1148 & 1.1182 \\
 &  & OCI & 0.8830 & 0.8650 & 0.8000 & 0.8940 & 0.8630 & 0.8040 & 0.9050 & 0.8770 & 0.8040 \\
 &  & OSE & 0.1841 & 0.1367 & 0.1036 & 0.1963 & 0.1371 & 0.1055 & 0.2058 & 0.1412 & 0.1068 \\
 &  & CI & 0.8550 & 0.8140 & 0.7860 & 0.8840 & 0.8220 & 0.7800 & 0.8870 & 0.8400 & 0.7750 \\
 &  & SE & 0.1747 & 0.1281 & 0.1055 & 0.1911 & 0.1333 & 0.1069 & 0.1990 & 0.1342 & 0.1074 \\
\cmidrule(lr){2-12}
 & \multirow{5}{*}{RF} & Est & 1.3114 & 1.2929 & 1.2726 &  &  &  &  &  &  \\
 &  & OCI & 0.7090 & 0.3540 & 0.1400 &  &  &  &  &  &  \\
 &  & OSE & 0.2184 & 0.1266 & 0.0901 &  &  &  &  &  &  \\
 &  & CI & 0.2110 & 0.0700 & 0.0180 &  &  &  &  &  &  \\
 &  & SE & 0.0794 & 0.0492 & 0.0385 &  &  &  &  &  &  \\
\bottomrule
\end{tabular}
}
\end{table}

\begin{table}[t!]
\small
\renewcommand{\arraystretch}{0.6}
\setlength{\tabcolsep}{3pt}
\centering
\caption{Simulation results: Erd\H{o}s--R\'enyi graph}
\label{ch5:tab:erdos-renyi-main-simulation}
\resizebox{\textwidth}{!}{
\begin{tabular}{lllccccccccc}
\toprule
& & & \multicolumn{3}{c}{$L=1$} & \multicolumn{3}{c}{$L=2$} & \multicolumn{3}{c}{$L=3$} \\
\cmidrule(lr){4-6} \cmidrule(lr){7-9} \cmidrule(lr){10-12}
Estimand & & $n$ & 1000 & 2500 & 4000 & 1000 & 2500 & 4000 & 1000 & 2500 & 4000 \\
\midrule

\multirow{15}{*}{OCDE} & \multirow{5}{*}{GNN} & Est & 1.0451 & 1.0615 & 1.0611 & 0.9821 & 0.9980 & 0.9935 & 0.9651 & 0.9879 & 0.9927 \\
 &  & OCI & 0.9430 & 0.9180 & 0.9080 & 0.9530 & 0.9570 & 0.9490 & 0.9360 & 0.9530 & 0.9510 \\
 &  & OSE & 0.1766 & 0.1262 & 0.1055 & 0.1534 & 0.1121 & 0.0938 & 0.1693 & 0.1179 & 0.0987 \\
 &  & CI & 0.9190 & 0.9220 & 0.9020 & 0.9000 & 0.9530 & 0.9470 & 0.8830 & 0.9470 & 0.9460 \\
 &  & SE & 0.1643 & 0.1257 & 0.1039 & 0.1284 & 0.1098 & 0.0932 & 0.1338 & 0.1115 & 0.0946 \\
\cmidrule(lr){2-12}
 & \multirow{5}{*}{MLP} & Est & 1.0687 & 1.0717 & 1.0684 & 1.0588 & 1.0664 & 1.0641 & 1.0563 & 1.0655 & 1.0638 \\
 &  & OCI & 0.9410 & 0.9140 & 0.9020 & 0.9460 & 0.9290 & 0.9080 & 0.9410 & 0.9170 & 0.9100 \\
 &  & OSE & 0.2059 & 0.1333 & 0.1082 & 0.2126 & 0.1379 & 0.1098 & 0.2086 & 0.1377 & 0.1084 \\
 &  & CI & 0.9290 & 0.9130 & 0.9000 & 0.9290 & 0.9120 & 0.9100 & 0.9440 & 0.9110 & 0.9000 \\
 &  & SE & 0.2007 & 0.1318 & 0.1057 & 0.2043 & 0.1335 & 0.1071 & 0.2027 & 0.1341 & 0.1074 \\
\cmidrule(lr){2-12}
 & \multirow{5}{*}{RF} & Est & 1.1698 & 1.1523 & 1.1391 &  &  &  &  &  &  \\
 &  & OCI & 0.8310 & 0.6840 & 0.6350 &  &  &  &  &  &  \\
 &  & OSE & 0.1690 & 0.1012 & 0.0853 &  &  &  &  &  &  \\
 &  & CI & 0.4060 & 0.2410 & 0.1850 &  &  &  &  &  &  \\
 &  & SE & 0.0733 & 0.0442 & 0.0349 &  &  &  &  &  &  \\
\midrule

\multirow{15}{*}{ONDE} & \multirow{5}{*}{GNN} & Est & 1.0596 & 1.0668 & 1.0706 & 0.9781 & 0.9962 & 1.0005 & 0.9681 & 0.9919 & 0.9947 \\
 &  & OCI & 0.9310 & 0.9190 & 0.8670 & 0.9510 & 0.9480 & 0.9560 & 0.9500 & 0.9420 & 0.9530 \\
 &  & OSE & 0.1804 & 0.1135 & 0.0852 & 0.1638 & 0.1109 & 0.0863 & 0.1855 & 0.1139 & 0.0839 \\
 &  & CI & 0.9140 & 0.8890 & 0.8680 & 0.9160 & 0.9380 & 0.9510 & 0.8960 & 0.9330 & 0.9420 \\
 &  & SE & 0.1649 & 0.1093 & 0.0865 & 0.1318 & 0.1043 & 0.0834 & 0.1484 & 0.1049 & 0.0844 \\
\cmidrule(lr){2-12}
 & \multirow{5}{*}{MLP} & Est & 1.0681 & 1.0721 & 1.0731 & 1.0689 & 1.0740 & 1.0732 & 1.0668 & 1.0707 & 1.0720 \\
 &  & OCI & 0.9240 & 0.8880 & 0.8650 & 0.9350 & 0.8840 & 0.8490 & 0.9330 & 0.9020 & 0.8610 \\
 &  & OSE & 0.1482 & 0.1000 & 0.0813 & 0.1588 & 0.1025 & 0.0826 & 0.1538 & 0.1024 & 0.0813 \\
 &  & CI & 0.9190 & 0.8840 & 0.8440 & 0.9210 & 0.8850 & 0.8370 & 0.9310 & 0.8840 & 0.8500 \\
 &  & SE & 0.1454 & 0.0989 & 0.0808 & 0.1508 & 0.1014 & 0.0823 & 0.1542 & 0.1012 & 0.0813 \\
\cmidrule(lr){2-12}
 & \multirow{5}{*}{RF} & Est & 1.2023 & 1.1698 & 1.1562 &  &  &  &  &  &  \\
 &  & OCI & 0.5730 & 0.3230 & 0.2420 &  &  &  &  &  &  \\
 &  & OSE & 0.1124 & 0.0708 & 0.0586 &  &  &  &  &  &  \\
 &  & CI & 0.2660 & 0.1100 & 0.0590 &  &  &  &  &  &  \\
 &  & SE & 0.0683 & 0.0419 & 0.0333 &  &  &  &  &  &  \\
\bottomrule
\end{tabular}
}
\end{table}

We begin with the finite-sample behavior of the GNN estimator, which Tables \ref{ch5:tab:random-geometric-main-simulation} and \ref{ch5:tab:erdos-renyi-main-simulation} show to be closely tied to its range of information aggregation. With $L=2$ the GNN aggregates information up to second-order neighbors, matching the data generating mechanism of this section, and its point estimates approach the true value rapidly in every setting: the downward bias visible at $n=1000$ diminishes markedly at $n=2500$ and $n=4000$, while the OSE decreases steadily with the sample size. Increasing the number of layers to $L=3$ introduces no appreciable additional bias and performs similarly, so that moderately widening the aggregation range is fairly robust. The OSE at $L=3$ is nevertheless slightly higher than at $L=2$ in most settings, most visibly at $n=1000$, which shows that aggregating beyond the actual range of dependence does not improve efficiency further and may instead induce extra finite-sample variability. By contrast, $L=1$ uses only first-order neighbors and cannot capture the second-order network dependence of the data generating mechanism; its point estimates display a persistent upward bias under both network models, and as $n$ grows the OSE keeps decreasing while the estimates do not approach the true value correspondingly, so that the OCI in fact deteriorates. An insufficient number of layers thus produces an approximation bias that cannot be removed by increasing the sample size, and the problem is more pronounced under the random geometric graph, so that the consequences of omitting higher-order network information also vary with the network topology.

Second, compared with the MLP and the RF, the GNN exhibits a more stable overall advantage in point estimation bias, sampling variability and oracle coverage. The MLP results scarcely change with the number of layers: adding hidden layers cannot compensate for the absence of a mechanism for propagating information along the network. Its bias does not diminish as $n$ grows and its OSE is usually larger than that of the GNN, the efficiency advantage of GNN being particularly clear under the random geometric graph and for the $OCDE$. Even where the MLP has a slightly smaller OSE, as for the $ONDE$ under the Erd\H{o}s--R\'enyi graph, its bias leaves its OCI well below that of the GNN. The RF performs worse overall: its point estimates exceed the true value substantially in every setting, the bias remains pronounced as the sample size grows, and although its OSE is at times smaller than that of the GNN under the Erd\H{o}s--R\'enyi graph, the smaller sampling variability does not translate into better inference and its OCI falls far short of the nominal level. In the presence of higher-order and nonlinear network confounding, conventional machine learning methods that do not exploit the network structure explicitly can therefore hardly deliver reliable estimates of causal effects, whereas the GNN, which propagates information layer by layer along the graph, performs markedly better.

Finally, comparisons among SE, OSE, CI and OCI bear on the finite-sample behavior of the network HAC variance estimator. For the GNN with $L=2$ and $L=3$, the OCI is close to $0.95$ in every setting, so that the point estimator is well centered and well approximated by a normal distribution. The CI falls below the OCI in smaller samples, chiefly because the network HAC standard error understates the OSE to some degree. The gap narrows appreciably as the sample size grows: at $n=4000$ the SE for $L=2$ and $L=3$ reaches roughly $89\%$ to $96\%$ of the OSE under the random geometric graph, and the two are nearly identical under the Erd\H{o}s--R\'enyi graph, their ratio lying between about $96\%$ and $101\%$; the variance estimator is thus fairly accurate in large samples. The MLP presents a different picture: its SE and OSE are usually close and the discrepancy between CI and OCI comparatively small, but both fall below the nominal level as the sample size grows, which shows that its undercoverage stems from a nonvanishing point estimation bias rather than from the variance estimator itself. The RF suffers from both problems at once, with a markedly low OCI and an SE that substantially understates the OSE. On the whole, when the nuisance functions are estimated accurately the network HAC estimator improves steadily with the sample size, whereas poor fits may produce bias and variance understatement simultaneously and hence severe undercoverage.

Appendix C reports further results for the own natural indirect effect, the spillover controlled direct effect and the spillover interventional direct and indirect effects.

\section{Empirical Application}
\label{ch5:sec:empirical-application}
We use the method proposed here to assess the influence of an intensive information session and of a default purchase option on farmers' insurance decisions when a new rice insurance product is marketed in rural China, and we decompose that influence into an indirect effect transmitted through farmers' perceptions and a direct effect that does not operate through this channel.

The data used in this section come from \citet{cai2015social}, who designed a randomized field experiment in rural China; the design is intricate, and we describe only what is needed here. The authors organized two rounds of insurance information sessions, the second three days after the first. Each round comprised a simple session, presenting only basic product information, and an intensive session, which in addition provided financial education on how insurance works and how to compute expected returns; every participating farmer was randomly assigned to exactly one of the four sessions. Because rural social networks are dense, first-round participants shared and discussed what they had learned with farmers who had not yet attended the second round, so that the decisions of the latter are influenced not only by their own session but also by spillover from the first round.

Of the more than $5000$ individuals in the original data, $4848$ remain after those with missing values are dropped, and all of them are included in our study, that is, $\mathcal{J}_n=\mathcal{N}_n$. The exposure mapping and the mediator mapping of every individual are set directly to $1$, that is, $T_i=S_i=1$, and we estimate $OCDE(1,0,1)$, $OCDE(1,1,1)$, $ONDE(1)$ and $ONIE(1)$, controlled and natural effects that place no restriction on the treatments and mediators of the others. Fixing $T_i$ and $S_i$ at constants facilitates comparison with \citet{cai2015social}, whose effects of interest place no restriction on the variables of others as well; our estimands are then very similar to what the coefficients of their model capture, although strictly speaking their paper does not adopt a causal inference framework and those coefficients do not admit a causal interpretation. Other specifications of $T_i$ and $S_i$ would deliver many further effects that their study does not consider, an important advantage of our framework.

We first study the effect of attending the intensive session on the insurance purchase decision. We consider three possible mediators: whether the individual answered at least $40\%$ of the post-session insurance knowledge test correctly, believes that more than $50\%$ of farmers purchase the insurance, and assesses the probability of a severe disaster in the coming year at more than $40\%$. This parallels the question studied in Table 2 of \citet{cai2015social}, with the addition of mediators. We then study the effect of the default purchase option, under which a farmer must sign to decline insurance when purchase is the default and to purchase when no purchase is the default, taking the latter two of the three variables above as possible mediators. This effect is not studied in \citet{cai2015social} and may be regarded as an additional contribution of our empirical analysis. Following the original paper, we use gender, age, household size, rice planting area in 2010 and literacy as covariates.

Table \ref{ch5:tab:agricultural-insurance-estimates} summarizes the empirical estimates; each causal estimand is estimated by the same three methods as in the simulation study, with the number of layers of the graph neural network and the multilayer perceptron set to $2$. The GNN estimates show that the intensive session has a stable and significantly positive influence on insurance purchase under every mediation mechanism. When the mediator is the insurance knowledge test score, the estimate of $OCDE(1,0,1)$, $0.0628$, is markedly larger than that of $OCDE(1,1,1)$, $0.0352$, and both are significant: the influence of the intensive session is heterogeneous across mediator levels, and those with a poorer command of insurance knowledge are more inclined to be persuaded by the intensive session to purchase. The estimates of $ONDE(1)$ and $ONIE(1)$ are $0.0383$ and $0.0351$, respectively, and both are significant, which indicates that financial education influences the decision not only by improving farmers' understanding of insurance, but also directly or through other channels this mediator does not fully capture. However, measuring the total effect by the sum of the natural direct and indirect effects, the insurance knowledge channel accounts for a considerable share of it.

\begin{table}[t!]
\small
\renewcommand{\arraystretch}{0.6}
\setlength{\tabcolsep}{3pt}
\centering
\caption{Summary of the empirical estimates}
\label{ch5:tab:agricultural-insurance-estimates}
\begin{threeparttable}
\resizebox{\textwidth}{!}{
\begin{tabular}{ccccccccc}
\toprule
Treatment & Mediator & Estimand & \multicolumn{2}{c}{GNN} & \multicolumn{2}{c}{MLP} & \multicolumn{2}{c}{RF} \\
\midrule
\multirow{12}{*}{\makecell{Intensive \\session}} & \multirow{4}{*}{Test score} & $OCDE(1,0,1)$ & 0.0628 & (0.0230) & 0.0432 & (0.0197) & 0.0368 & (0.0126) \\
& & $OCDE(1,1,1)$ & 0.0352 & (0.0188) & 0.0361 & (0.0225) & 0.0316 & (0.0139) \\
& & $ONDE(1)$ & 0.0383 & (0.0148) & 0.0435 & (0.0157) & 0.0341 & (0.0106) \\
& & $ONIE(1)$ & 0.0351 & (0.0055) & 0.0325 & (0.0046) & 0.0368 & (0.0039) \\
\cmidrule{2-9}
& \multirow{4}{*}{\makecell{Perceived \\take-up}} & $OCDE(1,0,1)$ & 0.0660 & (0.0247) & 0.0855 & (0.0236) & 0.0670 & (0.0139) \\
& & $OCDE(1,1,1)$ & 0.0965 & (0.0159) & 0.0784 & (0.0198) & 0.0718 & (0.0124) \\
& & $ONDE(1)$ & 0.0803 & (0.0138) & 0.0770 & (0.0159) & 0.0707 & (0.0112) \\
& & $ONIE(1)$ & -0.0009 & (0.0011) & -0.0002 & (0.0012) & -0.0004 & (0.0007) \\
\cmidrule{2-9}
& \multirow{4}{*}{\makecell{Disaster \\probability}} & $OCDE(1,0,1)$ & 0.1017 & (0.0192) & 0.0987 & (0.0196) & 0.0873 & (0.0125) \\
& & $OCDE(1,1,1)$ & 0.0527 & (0.0200) & 0.0490 & (0.0216) & 0.0459 & (0.0131) \\
& & $ONDE(1)$ & 0.0703 & (0.0143) & 0.0771 & (0.0159) & 0.0692 & (0.0112) \\
& & $ONIE(1)$ & 0.0014 & (0.0010) & 0.0005 & (0.0010) & 0.0008 & (0.0006) \\
\midrule
\multirow{8}{*}{\makecell{Default \\purchase}} & \multirow{4}{*}{\makecell{Perceived \\take-up}} & $OCDE(1,0,1)$ & 0.0609 & (0.0311) & 0.0905 & (0.0274) & 0.0841 & (0.0159) \\
& & $OCDE(1,1,1)$ & 0.1301 & (0.0274) & 0.0952 & (0.0277) & 0.0912 & (0.0173) \\
& & $ONDE(1)$ & 0.0694 & (0.0250) & 0.0891 & (0.0242) & 0.0897 & (0.0162) \\
& & $ONIE(1)$ & -0.0012 & (0.0043) & -0.0032 & (0.0020) & -0.0001 & (0.0007) \\
\cmidrule{2-9}
& \multirow{4}{*}{\makecell{Disaster \\probability}} & $OCDE(1,0,1)$ & 0.1327 & (0.0261) & 0.0974 & (0.0260) & 0.1001 & (0.0165) \\
& & $OCDE(1,1,1)$ & 0.0674 & (0.0315) & 0.0746 & (0.0303) & 0.0791 & (0.0169) \\
& & $ONDE(1)$ & 0.0793 & (0.0242) & 0.0923 & (0.0242) & 0.0896 & (0.0162) \\
& & $ONIE(1)$ & -0.0003 & (0.0032) & -0.0007 & (0.0011) & 0.0003 & (0.0005) \\
\bottomrule
\end{tabular}
}
\begin{tablenotes}[flushleft]
\footnotesize
\item \textit{Note}: Point estimates are reported outside the parentheses and standard errors inside.
\end{tablenotes}
\end{threeparttable}
\end{table}

These results agree in direction with the main findings of \citet{cai2015social}, though not in magnitude. Table 2 of the original paper reports that the intensive session raises the insurance take-up rate by about $14$ to $15$ percentage points in the first-round sample, an increase the authors attribute largely to a better understanding of the returns to the product and of how it works. The total effect obtained here, about $7$ to $8$ percentage points, is smaller, chiefly because we use the full sample and adjust simultaneously for covariates, network dependence and the mediation mechanism, whereas their estimate rests on the first-round contrast between the simple and the intensive sessions alone. Since second-round farmers may already have obtained information from first-round participants through the social network, the marginal influence of the intensive session is naturally attenuated in the full sample. Sign and statistical significance nevertheless agree with the original paper.

Turning to the other two mediators, when the mediator is the perceived insurance take-up rate the two controlled direct effects are both significantly positive with $OCDE(1,1,1)$ the larger, so that the effect of the intensive session is stronger among subjects who believe that most people have purchased insurance; $ONDE(1)$ is significantly positive but $ONIE(1)$ is essentially zero and insignificant, so that the positive influence on purchase is not achieved by altering farmers' beliefs about the take-up of others. This echoes the mechanism analysis of the original paper: the network transmits knowledge about the product more effectively than information about the purchase decisions of others. When the mediator is the subjective disaster probability the pattern is the same except that $OCDE(1,0,1)$ is now the larger, so that the intensive session matters mainly among subjects who believe a severe disaster unlikely; with $ONDE(1)$ again significantly positive and $ONIE(1)$ again essentially zero, the session appears to raise take-up by improving farmers' appreciation of the product itself rather than by changing their assessment of future disaster risk.

We next consider the default purchase option. Since \citet{cai2015social} do not study it as a principal treatment, the following results extend our empirical analysis beyond the original paper. The GNN estimates show that the default option likewise raises the probability of purchase significantly. Both controlled direct effects are significantly positive under either mediation mechanism, and their ordering repeats the pattern found for the intensive session, so that the effect of the default option is again stronger among subjects who believe that most people have purchased insurance or that a severe disaster is unlikely. The natural direct effect is significantly positive under both mechanisms, raising the probability of purchase by about $7$ to $8$ percentage points, whereas both estimates of $ONIE(1)$ are close to zero and insignificant. The default option therefore appears to act directly, by lowering the cost of active choice, shifting the reference point or exploiting behavioral inertia, rather than through either perception. This echoes the results for the intensive session: for both treatments, the influence on purchase bypasses the two considered perception-based mediators.

Finally, the MLP and RF estimates agree with the GNN estimates in overall direction: most direct effect estimates are positive, and the natural indirect effect is generally close to zero under the perceived take-up and disaster probability mediators, so that the main conclusions do not depend on the particular method used to estimate the nuisance functions. The GNN estimates do differ somewhat for several controlled direct effects, most visibly under the default purchase treatment with the disaster probability mediator, as one would expect if individual decisions were jointly influenced by network position, neighbors' characteristics and the local structure of information transmission. Since the GNN, unlike the MLP and the RF, incorporates network neighborhood information into the estimation of the nuisance functions, we take it as the primary basis for our empirical analysis. Overall, our method reproduces the basic finding of the original paper that the intensive session raises the insurance take-up rate and further distinguishes the mediation channels: insurance knowledge is a key mechanism through which the intensive session operates, whereas perceived take-up and perceived disaster probability are not, and the default option changes purchase behavior directly rather than through these perceptions.

\section{Conclusion}
\label{ch5:sec:conclusion}
This paper develops a framework for causal mediation analysis in a single large observed network, allowing the treatment and the mediator to spill over simultaneously and the outcome to depend on individual characteristics and network structure in a general nonlinear fashion. Exposure and mediator mappings are used solely to define the own controlled direct effect and the own natural direct and indirect effects, without requiring the true interference mechanism to conform to a prespecified aggregation form. For estimation and inference, we construct doubly or multiply robust augmented inverse probability weighted estimators and use graph neural networks to learn the high-dimensional nuisance functions directly from the node features and the network. Under approximate neighborhood interference, weak network dependence and suitable first-stage conditions, we establish asymptotic normality and the consistency of a network HAC variance estimator. The simulations demonstrate the advantages of our GNN estimator over conventional ML estimators, and the empirical analysis validates its practical applicability in real-world settings. 

Several extensions are worth pursuing. First, a longitudinal version of the framework could accommodate time-varying treatments, mediators and outcomes, as well as an evolving network, thereby tracing how direct, indirect and spillover effects accumulate over time. Second, the framework could be extended to continuous or multi-valued treatments and mediators and to a fixed-dimensional vector of mediators, with corresponding modifications to the nuisance functions and score construction; further decomposition into mediator-specific pathways would permit a more detailed analysis of multiple mechanisms. Third, developing uniform inference, simultaneous testing and methods for heterogeneous mediation effects across network positions or subpopulations would further broaden the range of empirical questions addressed by the framework.

\begin{appendices}
\section{Asymptotic Theory of Natural Direct and Indirect Effects}
\label{ch5:subsec:asymptotic-of-natural-effects}

We now study the asymptotic properties of the estimators of the own natural direct and indirect effects at a fixed $t$. Since the estimators of the own natural direct effect and of the own natural indirect effect are both constructed from the same AIPW score estimator $\widehat{\tau}_i^N(d,d^*,t)$, their asymptotic theory is entirely analogous, and for brevity we discuss only the former. As in the main text, the following assumptions are required. The discussion of these assumptions parallels that in Section \ref{ch5:sec:asymptotic-theory} exactly, and to avoid repetition we omit it and state the assumptions alone.
\begin{assumption}[Exposure mapping]
\label{ch5:ass:own-natural-local-exposure-mapping}
There exists a fixed positive integer $K_T$ such that the exposure mapping $T_i$ is determined solely by the treatments of the individuals in the $K_T$-neighborhood. Specifically,
\[
T_i = t_n(i,\bD_{\mathcal{N}(i,K_T)\setminus\{i\}},\bA).
\]
\end{assumption}
 
To state the next assumption, define
\[
\varphi_i^N(t) \coloneq [\tau_i^N(1,0,t) - \tau_i^N(0,0,t)] - NDE(t),
\]
and
\[
\sigma_{n,N}^2 \coloneq \Var\bigg(\frac{1}{\sqrt{j_n}} \sum_{i\in\mathcal{J}_n}\varphi_i^N(t) \mid \bX,\bA \bigg).
\]
\begin{assumption}[Moment conditions]
  \label{ch5:ass:own-natural-effects-moments-overlap}
  (i) There exists a constant $M < \infty$ such that $|Y_i(\bd, \bm)| \leqslant M$ almost surely for all $n\in\mathbb{N}$, $i\in\mathcal{N}_n$, $d\in\{0,1\}^n$ and $m\in\{0,1\}^n$.\\
  (ii) There exists an interval $[\underline{\pi},\overline{\pi}] \subset (0,1)$ such that, for all $n\in\mathbb{N}$, $i\in\mathcal{N}_n$ and any $(d,d^*,t,m)$,
  \[
  e_i^N(d,t),\ \widehat{e}_i^N(d,t),\ p_i^N(d),\ \widehat{p}_i^N(d),\ q_i^N(m \mid d),\ \widehat{q}_i^N(m \mid d) \in [\underline{\pi},\overline{\pi}]
  \]
  and $j_n/n \geqslant \underline{\pi}$ almost surely.\\
  (iii) $\liminf_{n\to\infty}\sigma_{n,N}^2 > 0$ almost surely.
\end{assumption}
 
For any $x\in\mathbb{R}$ we define
\begin{gather*}
  \Psi_{\eta^N}(Z_i,x,d,d^*,t) \coloneq \big[x - \eta_i^N(d,t,M_i)\big]\bigg[\frac{\1\{D_i=d^*\}}{p_i^N(d^*)} - \frac{\1\{D_i=d,T_i=t\}}{e_i^N(d,t)}\frac{q_i^N(M_i \mid d^*)}{q_i^N(M_i \mid d)}\bigg], \\
  \Psi_{q^N}(Z_i,x_1,x_0,d,d^*,t) \coloneq \frac{\1\{D_i=d,T_i=t\}}{e_i^N(d,t)}\big[Y_i - \eta_i^N(d,t,M_i)\big]\bigg[\frac{x_1}{x_0} - \frac{q_i^N(M_i \mid d^*)}{q_i^N(M_i \mid d)}\bigg], \\
  \Psi_{e^N}(Z_i,x,d,d^*,t) \coloneq \1\{D_i=d,T_i=t\}\big[Y_i - \eta_i^N(d,t,M_i)\big]\frac{q_i^N(M_i \mid d^*)}{q_i^N(M_i \mid d)}\bigg[\frac{1}{x} - \frac{1}{e_i^N(d,t)}\bigg], \\
  \Psi_{p^N}(Z_i,x,d,d^*,t) \coloneq \1\{D_i=d^*\}\big[\eta_i^N(d,t,M_i) - \mu_i^N(d,d^*,t)\big]\bigg[\frac{1}{x} - \frac{1}{p_i^N(d^*)}\bigg], \\
  \Psi_{\mu^N}(Z_i,x,d,d^*,t) \coloneq \big[x - \mu_i^N(d,d^*,t)\big]\bigg[1 - \frac{1\{D_i=d^*\}}{p_i^N(d^*)}\bigg].
\end{gather*}
 
\begin{assumption}[GNN convergence rates]
  \label{ch5:ass:own-natural-effects-gnn-convergence-rates}
  For any $(d,d^*,t,m)$, the following conditions hold:\\
  (i) (Mean squared error)
  \begin{gather*}
    \frac{1}{j_n} \sum_{i\in\mathcal{J}_n}[\widehat{\eta}_i^N(d,t,m) - \eta_i^N(d,t,m)]^2 = o_p(n^{-1/2}), \\
    \frac{1}{j_n} \sum_{i\in\mathcal{J}_n}[\widehat{q}_i^N(m \mid d) - q_i^N(m \mid d)]^2 = o_p(n^{-1/2}), \\
    \frac{1}{j_n} \sum_{i\in\mathcal{J}_n}[\widehat{e}_i^N(d,t) - e_i^N(d,t)]^2 = o_p(n^{-1/2}), \\
    \frac{1}{j_n} \sum_{i\in\mathcal{J}_n}[\widehat{p}_i^N(d) - p_i^N(d)]^2 = o_p(n^{-1/2}), \\
    \frac{1}{j_n} \sum_{i\in\mathcal{J}_n}[\widehat{\mu}_i^N(d,d^*,t) - \mu_i^N(d,d^*,t)]^2 = o_p(n^{-1/2}).
  \end{gather*}
  (ii) (Stochastic equicontinuity)
  \begin{gather*}
    \frac{1}{j_n}\sum_{i\in\mathcal{J}_n}\Psi_{\eta^N}\big(Z_i,\widehat{\eta}_i^N(d,t,M_i),d,d^*,t\big) = o_p(n^{-1/2}), \\
    \frac{1}{j_n}\sum_{i\in\mathcal{J}_n}\Psi_{q^N}\big(Z_i,\widehat{q}_i^N(M_i \mid d^*),\widehat{q}_i^N(M_i \mid d),d,d^*,t\big) = o_p(n^{-1/2}), \\
    \frac{1}{j_n}\sum_{i\in\mathcal{J}_n}\Psi_{e^N}\big(Z_i,\widehat{e}_i^N(d,t),d,d^*,t\big) = o_p(n^{-1/2}), \\
    \frac{1}{j_n}\sum_{i\in\mathcal{J}_n}\Psi_{p^N}\big(Z_i,\widehat{p}_i^N(d^*),d,d^*,t\big) = o_p(n^{-1/2}), \\
    \frac{1}{j_n}\sum_{i\in\mathcal{J}_n}\Psi_{\mu^N}\big(Z_i,\widehat{\mu}_i^N(d,d^*,t),d,d^*,t\big) = o_p(n^{-1/2}).
  \end{gather*}
\end{assumption}
The $\psi$-dependence coefficient required by the next assumption is defined as
\begin{equation}
  \psi_n(r) \coloneq \begin{cases}
    1, & 0 \leqslant r \leqslant 4K_T, \\
    \gamma_n(r/4), & r > 4K_T.
  \end{cases}
  \label{ch5:eq:own-natural-effects-psi-dependence-coefficient}
\end{equation}
\begin{assumption}[$\psi$-dependence]
  \label{ch5:ass:own-natural-effects-weak-dependence}
  (i) $\sup_{n \in \mathbb{N}} \max_{r \geqslant 1} \psi_n(r) < \infty$ almost surely.\\
(ii) For some fixed $p>4$, some sequence of positive numbers $v_n \to \infty$ and any $k \in \{1, 2\}$,
\[
\frac{1}{n^{k/2}} \sum_{r=0}^{\infty}c_n(r, v_n, k) \psi_n(r)^{1 - (2+k)/p} = o(1) \quad \text{and} \quad n^{3/2} \psi_n(v_n)^{1 - 1/p} = o(1) \quad \text{a.s.}
\]
\end{assumption}
Under the assumptions above we obtain the asymptotic normality of the AIPW estimator of the NDE.
\begin{theorem}
  \label{ch5:thm:own-nde-asymptotic-normality}
  Under Assumptions \ref{ch5:ass:consistency}, \ref{ch5:ass:error-term-independence}, \ref{ch5:ass:own-natural-effects-conditional-independence}, \ref{ch5:ass:own-natural-local-exposure-mapping}, \ref{ch5:ass:own-natural-effects-moments-overlap}, \ref{ch5:ass:approximate-neighborhood-interference}, \ref{ch5:ass:own-natural-effects-gnn-convergence-rates} and \ref{ch5:ass:own-natural-effects-weak-dependence}, for any fixed $t$,
\[
\sigma_{n,N}^{-1}\sqrt{j_n}\big(\widehat{NDE}(t)-NDE(t)\big) \xrightarrow{d} \mathcal{N}(0,1).
\]
\end{theorem}
 
We now construct an estimator of the asymptotic variance. Define the residual vector
\begin{equation*}
\widetilde{\tau} = \big(\widehat{\tau}_i^N(1,0,t) - \widehat{\tau}_i^N(0,0,t) - \widehat{\mu}_i^N(1,0,t) + \widehat{\mu}_i^N(0,0,t) \big)_{i \in \mathcal{J}_n},
\end{equation*}
and the HAC variance estimator
\[
\widehat{\sigma}_{N}^2 \coloneq \max\{\widehat{\sigma}_U^2, \widehat{\sigma}_{PD}^2\} \quad \text{where} \quad \widehat{\sigma}_q^2 \coloneq \frac{1}{j_n}\, \widetilde{\tau}' K_q\, \widetilde{\tau}, \quad q \in \{U, PD\},
\]
where $K_U$ and $K_{PD}$ are as defined for the variance estimator of the previous section. We impose the following assumption in order to obtain a consistent estimator of the asymptotic variance.
\begin{assumption}
\label{ch5:ass:own-natural-effects-hac-consistency}
(i) There exists $M > 0$ such that, for all $n \in \mathbb{N}$, $i \in N_n$ and any $(d,d^*,t,m)$, $\max\{|Y_i|, |\widehat{\mu}_i^N(d,d^*,t)|, |\widehat{\eta}_i^N(d, t, m)|\} < M$ almost surely.\\
(ii) For some $p>4$ and some positive sequence $b_n \to \infty$, $\displaystyle n^{-1} \sum_{r=0}^{\infty} c_n(r, b_n, 2)\, \psi_n(r)^{1 - 4/p} = o(1)$, $\sum_{r>b_n}\delta_n^B(r,1)\psi_n(r)^{1-2/p}=o(1)$ and $\sum_{r=0}^{n}\psi_n(r)^{1-2/p} n^{-1}\sum_{i=1}^n\sum_{j\in\mathcal{N}^B(i,r)}|K_{PD,ij}-1| = o(1)$ almost surely.\\
(iii) $\displaystyle n^{-1} \sum_{i=1}^n n(i, b_n)^2 = O_p(\sqrt{n})$.
\end{assumption}
 
\begin{theorem}
  \label{ch5:thm:own-nde-variance-consistency}
Under Assumption \ref{ch5:ass:own-natural-effects-hac-consistency} and the assumptions of Theorem \ref{ch5:thm:own-nde-asymptotic-normality}, $\widehat{\sigma}_{N}^2 = \sigma_{n,N}^2+o_p(1)$ for any fixed $t$.
\end{theorem}

\section{Identification, Estimation, and Asymptotic Theory for Additional Effects}
\label{ch5:subsec:other-effects}
\subsection{Own interventional direct/indirect effects}
In addition to controlled and natural effects, we consider the own interventional direct and indirect effects. For any individual $i$ and $d\in\{0,1\}$, let $G_i(d)$ be an auxiliary random variable such that
\[
G_i(d)\mid\bX,\bA\stackrel{d}{=}M_i(d)\mid\bX,\bA,
\]
and, conditional on $\bX,\bA$, $G_i(d)$ is independent of all other variables. Thus, conditional on $\bX,\bA$, $G_i(d)$ represents an independent stochastic intervention drawn from the conditional distribution of $M_i(d)$. Define the individual-level interventional potential outcome
\[
\begin{aligned}
Y_i^I(d_i,d_i^*,t) \coloneq{}& \sum_{\bdmi:T_i(\bdmi)=t} \E\!\left[Y_i\!\left(d_i,\bdmi,G_i(d_i^*),\bMmi(\bdmi)\right) \mid\bX,\bA \right]  \\
&\qquad\times \P(\bDmi=\bdmi\mid T_i=t,\bX,\bA),
\end{aligned}
\]
and let
\begin{gather*}
OIDE(t)\coloneq \frac{1}{j_n}\sum_{i\in\mathcal{J}_n} \big[Y_i^I(1,0,t)-Y_i^I(0,0,t)\big],\\
OIIE(t)\coloneq \frac{1}{j_n}\sum_{i\in\mathcal{J}_n} \big[Y_i^I(1,1,t)-Y_i^I(1,0,t)\big].
\end{gather*}
Here, $OIDE(t)$ is the average effect of changing own treatment from 0 to 1 while stochastically intervening on the individual's own mediator according to its conditional distribution under no treatment; $OIIE(t)$ is the average effect, with own treatment fixed at 1, of changing the stochastic intervention distribution of the own mediator from its distribution under no treatment to that under treatment.

\begin{assumption}
\label{ch5:ass:own-interventional-effects-independence}
For every admissible $i,d_i,d_i^*,\bdmi$, and $m_i$, conditions (\ref{ch5:eq:own-natural-treatment-outcome-independence})--(\ref{ch5:eq:own-natural-treatment-mediator-independence}) hold.
\end{assumption}

Unlike identification of the natural direct and indirect effects, Assumption \ref{ch5:ass:own-interventional-effects-independence} does not require the cross-world conditional independence condition in (\ref{ch5:eq:own-natural-cross-world-independence}). Define the observed-data functional
\[
\mu_i^I(d_i,d_i^*,t) \coloneq \sum_{m_i\in\{0,1\}} \E(Y_i\mid D_i=d_i,T_i=t,M_i=m_i,\bX,\bA) \P(M_i=m_i\mid D_i=d_i^*,\bX,\bA),
\]
and
\begin{gather*}
OIDE^{obs}(t)\coloneq \frac{1}{j_n}\sum_{i\in\mathcal{J}_n} \big[\mu_i^I(1,0,t)-\mu_i^I(0,0,t)\big],\\
OIIE^{obs}(t)\coloneq \frac{1}{j_n}\sum_{i\in\mathcal{J}_n} \big[\mu_i^I(1,1,t)-\mu_i^I(1,0,t)\big].
\end{gather*}

\begin{theorem}
\label{ch5:thm:own-interventional-effects-identification}
Under Assumptions \ref{ch5:ass:consistency}, \ref{ch5:ass:error-term-independence}, and \ref{ch5:ass:own-interventional-effects-independence}, for any $i,d_i,d_i^*,t$,
\[
\mu_i^I(d_i,d_i^*,t)=Y_i^I(d_i,d_i^*,t),
\]
and hence
\[
OIDE^{obs}(t)=OIDE(t), \qquad OIIE^{obs}(t)=OIIE(t).
\]
\end{theorem}

Because $\mu_i^I(d,d^*,t)$ has the same observed-data representation as $\mu_i^N(d,d^*,t)$ in (\ref{ch5:eq:own-natural-effects-observed-functional}), the AIPW score constructed in the main text for the natural direct and indirect effects can be used directly under the corresponding overlap conditions. To emphasize that the score is now associated with an interventional effect, write
\[
\widehat{\tau}_i^I(d,d^*,t) \coloneq \widehat{\tau}_i^N(d,d^*,t).
\]
Accordingly, the AIPW estimators of $OIDE(t)$ and $OIIE(t)$ are, respectively,
\begin{gather*}
\widehat{OIDE}(t)\coloneq \frac{1}{j_n}\sum_{i\in\mathcal{J}_n} \big[\widehat{\tau}_i^I(1,0,t) - \widehat{\tau}_i^I(0,0,t)\big],\\
\widehat{OIIE}(t)\coloneq \frac{1}{j_n}\sum_{i\in\mathcal{J}_n} \big[\widehat{\tau}_i^I(1,1,t) - \widehat{\tau}_i^I(1,0,t)\big],
\end{gather*}
and inherit the multiple-robustness property of $\widehat{\tau}_i^N(d,d^*,t)$.

\subsection{Spillover controlled direct effect}
The spillover controlled direct effect measures the average effect on an individual's outcome of changing the treatment exposure mapping for other individuals from $t^*$ to $t$, while holding the individual's own treatment, own mediator, and mediator exposure mapping fixed at $d$, $m$, and $s$, respectively. It is defined as
\[
SCDE(d,m,s,t,t^*) \coloneq \frac{1}{j_n}\sum_{i\in\mathcal{J}_n} \big[Y_i^C(d,t,m,s)-Y_i^C(d,t^*,m,s)\big].
\]

The conditions required to identify the spillover controlled direct effect are the same as those required for the own controlled direct effect. For completeness, we restate them here.

\begin{assumption}
\label{ch5:ass:spillover-cde-conditional-independence}
For any $\bd,\bm$, and $i$,
\begin{gather*}
Y_i(\bd,\bm)\indep\bD\mid\bX,\bA,\\
Y_i(\bd,\bm)\indep\bM\mid\bD,\bX,\bA.
\end{gather*}
\end{assumption}

Assumption \ref{ch5:ass:spillover-cde-conditional-independence} requires that, conditional on $\bX,\bA$, there be no unobserved common causes of an individual's potential outcome and the treatment vector and, conditional further on the treatment vector $\bD$, no unobserved common causes of the individual's potential outcome and the mediator vector. Define the observed-data functional
\[
\mu_i^C(d,t,m,s) \coloneq \E\!\left(Y_i \mid D_i=d,T_i=t,M_i=m,S_i=s,\bX,\bA \right),
\]
and
\[
SCDE^{obs}(d,m,s,t,t^*) \coloneq \frac{1}{j_n}\sum_{i\in\mathcal{J}_n} \big[\mu_i^C(d,t,m,s) - \mu_i^C(d,t^*,m,s)\big].
\]

\begin{theorem}
\label{ch5:thm:spillover-cde-identification}
Under Assumptions \ref{ch5:ass:consistency}, \ref{ch5:ass:error-term-independence}, and \ref{ch5:ass:spillover-cde-conditional-independence}, for any $i,d,t,m,s$,
\[
\mu_i^C(d,t,m,s)=Y_i^C(d,t,m,s).
\]
Consequently, for any $d,m,s,t,t^*$,
\[
SCDE^{obs}(d,m,s,t,t^*) = SCDE(d,m,s,t,t^*).
\]
\end{theorem}

Theorem \ref{ch5:thm:spillover-cde-identification} shows that the spillover controlled direct effect is identified by the difference between observed conditional means at two treatment-exposure levels. Under the corresponding overlap conditions, using the individual-level AIPW score $\widehat{\tau}_i^C(d,t,m,s)$ defined in Section \ref{ch5:subsec:network-aipw-estimation}, we obtain
\[
\widehat{SCDE}(d,m,s,t,t^*) \coloneq \frac{1}{j_n}\sum_{i\in\mathcal{J}_n} \big[\widehat{\tau}_i^C(d,t,m,s) - \widehat{\tau}_i^C(d,t^*,m,s)\big].
\]
This estimator inherits the double-robustness property of $\widehat{\tau}_i^C(d,t,m,s)$: if either the outcome-regression nuisance function or the corresponding propensity-score nuisance function is correctly specified, the conditional expectation of the corresponding AIPW score equals the spillover controlled direct effect.

\subsection{Spillover natural direct/indirect effects}
Fix the individual's own treatment at $d$. The spillover natural direct effect measures the average effect on the individual's outcome of changing the treatment exposure of other individuals while holding their mediators at the natural levels induced by a specified treatment exposure. The spillover natural indirect effect measures the average effect generated by changes in the natural levels of other individuals' mediators as treatment exposure changes, while holding the treatment exposure of others fixed. Define
\begin{equation}
\begin{aligned}
Y_i^{SN}(d_i,t,t^*) \coloneq & \sum_{\bdmi:T_i(\bdmi)=t}\sum_{\bdmi^*:T_i(\bdmi^*)=t^*}\E[Y_i(d_i,\bdmi,M_i(d_i),\bMmi(\bdmi^*)) \mid \bX,\bA] \\
& \times \P(\bDmi=\bdmi\mid T_i=t,\bX,\bA) \P(\bDmi=\bdmi^*\mid T_i=t^*,\bX,\bA).
\end{aligned}
\end{equation}
Correspondingly, define
\begin{gather*}
SNDE(d,t_1,t_0,t^*) \coloneq \frac{1}{j_n}\sum_{i\in\mathcal{J}_n} \big[Y_i^{SN}(d,t_1,t^*) - Y_i^{SN}(d,t_0,t^*)\big], \\
SNIE(d,t,t_1^*,t_0^*) \coloneq \frac{1}{j_n}\sum_{i\in\mathcal{J}_n} \big[Y_i^{SN}(d,t,t_1^*) - Y_i^{SN}(d,t,t_0^*)\big].
\end{gather*}
Here, $SNDE(d,t_1,t_0,t^*)$ compares the average effect of changing the treatment exposure of other individuals from $t_0$ to $t_1$ while holding their mediators at the natural distribution under treatment exposure $t^*$; $SNIE(d,t,t_1^*,t_0^*)$ compares the average effect of changing the natural distribution of other individuals' mediators from that induced by $t_0^*$ to that induced by $t_1^*$ while holding their treatment exposure fixed at $t$.

\begin{assumption}
\label{ch5:ass:spillover-natural-effects-independence}
For any $i\in\mathcal{N}_n$, $d_i,m_i\in\{0,1\}$, $\bdmi,\bdmi^*\in\{0,1\}^{n-1}$, $\bmmi\in\{0,1\}^{n-1}$, and any $j\in\mathcal{N}_n$ and $d_j\in\{0,1\}$,
\begin{gather}
Y_i(d_i,\bdmi,m_i,\bmmi) \indep \bD \mid \bX,\bA, \label{ch5:eq:spillover-natural-treatment-outcome-independence} \\
Y_i(d_i,\bdmi,m_i,\bmmi) \indep \bM \mid \bD,\bX,\bA, \label{ch5:eq:spillover-natural-mediator-outcome-independence} \\
M_j(d_j) \indep D_j \mid \bX,\bA, \label{ch5:eq:spillover-natural-treatment-mediator-independence} \\
Y_i(d_i,\bdmi,m_i,\bmmi) \indep \big(M_i(d_i),\bMmi(\bdmi^*) \big) \mid \bX,\bA.
\label{ch5:eq:spillover-natural-cross-world-independence}
\end{gather}
\end{assumption}

The first three conditions in Assumption \ref{ch5:ass:spillover-natural-effects-independence} rule out unobserved treatment--outcome, mediator--outcome, and treatment--mediator confounding, respectively. Condition (\ref{ch5:eq:spillover-natural-cross-world-independence}) is a joint cross-world independence condition that permits the own natural mediator $M_i(d_i)$ and the natural mediators of other individuals under a different treatment assignment, $\bMmi(\bdmi^*)$, to enter the potential outcome simultaneously. Define the following observed-data nuisance functions:
\begin{gather*}
\eta_i^{SN}(d,\bdmi,m,\bmmi) \coloneq \E\!\left(Y_i \mid D_i=d,\bDmi=\bdmi,M_i=m,\bMmi=\bmmi,\bX,\bA \right), \\
q_i^{SN}(m\mid d) \coloneq \P(M_i=m\mid D_i=d,\bX,\bA), \\
\alpha_i^{SN}(\bdmi\mid t) \coloneq \P(\bDmi=\bdmi\mid T_i=t,\bX,\bA), \\
r_i^{SN}(\bmmi\mid t) \coloneq \P(\bMmi=\bmmi\mid T_i=t,\bX,\bA).
\end{gather*}
Let
\begin{equation}
\begin{aligned}
\mu_i^{SN}(d,t,t^*) \coloneq{}\, &\sum_{\bdmi:T_i(\bdmi)=t} \sum_{m\in\{0,1\}} \sum_{\bmmi\in\{0,1\}^{n-1}} \eta_i^{SN}(d,\bdmi,m,\bmmi) \\
&\quad\times q_i^{SN}(m\mid d) \alpha_i^{SN}(\bdmi\mid t) r_i^{SN}(\bmmi\mid t^*),
\end{aligned}
\end{equation}
and define
\begin{gather*}
SNDE^{obs}(d,t_1,t_0,t^*) \coloneq \frac{1}{j_n}\sum_{i\in\mathcal{J}_n} \big[\mu_i^{SN}(d,t_1,t^*) - \mu_i^{SN}(d,t_0,t^*)\big], \\
SNIE^{obs}(d,t,t_1^*,t_0^*) \coloneq \frac{1}{j_n}\sum_{i\in\mathcal{J}_n} \big[\mu_i^{SN}(d,t,t_1^*) - \mu_i^{SN}(d,t,t_0^*)\big].
\end{gather*}

\begin{theorem}
\label{ch5:thm:spillover-natural-effects-identification}
Under Assumptions \ref{ch5:ass:consistency}, \ref{ch5:ass:error-term-independence}, and \ref{ch5:ass:spillover-natural-effects-independence}, for any $i,d,t,t^*$,
\[
\mu_i^{SN}(d,t,t^*) = Y_i^{SN}(d,t,t^*).
\]
Consequently,
\begin{gather*}
SNDE^{obs}(d,t_1,t_0,t^*) = SNDE(d,t_1,t_0,t^*), \\
SNIE^{obs}(d,t,t_1^*,t_0^*) = SNIE(d,t,t_1^*,t_0^*).
\end{gather*}
\end{theorem}

Under the corresponding overlap conditions, AIPW estimators can be constructed from the identification functional above. To this end, define
\begin{align*}
\xi_i^{SN}(d,t^*,\bdmi,m) &\coloneq \sum_{\bmmi\in\{0,1\}^{n-1}} \eta_i^{SN}(d,\bdmi,m,\bmmi) r_i^{SN}(\bmmi\mid t^*), \\
\zeta_i^{SN}(d,t,\bmmi) &\coloneq \sum_{\bdmi:T_i(\bdmi)=t} \sum_{m\in\{0,1\}} \eta_i^{SN}(d,\bdmi,m,\bmmi) \alpha_i^{SN}(\bdmi\mid t) q_i^{SN}(m\mid d),
\end{align*}
and
\begin{gather*}
\rho_i^{SN}(\bmmi\mid d,\bdmi,m) \coloneq \P\!\left(\bMmi=\bmmi \mid D_i=d,\bDmi=\bdmi,M_i=m,\bX,\bA \right), \\
e_i^{SN}(d,t) \coloneq \P(D_i=d,T_i=t\mid\bX,\bA), \\
p_i^T(t) \coloneq \P(T_i=t\mid\bX,\bA).
\end{gather*}
By the preceding definitions,
\begin{align*}
\mu_i^{SN}(d,t,t^*) ={}& \sum_{\bdmi:T_i(\bdmi)=t} \sum_{m\in\{0,1\}} \xi_i^{SN}(d,t^*,\bdmi,m) \alpha_i^{SN}(\bdmi\mid t) q_i^{SN}(m\mid d) \\
={}& \sum_{\bmmi\in\{0,1\}^{n-1}} \zeta_i^{SN}(d,t,\bmmi) r_i^{SN}(\bmmi\mid t^*).
\end{align*}
Define the individual-level AIPW score
\begin{align*}
\tau_i^{SN}(d,t,t^*) \coloneq & \mu_i^{SN}(d,t,t^*) + \frac{\1\{D_i=d,T_i=t\}}{e_i^{SN}(d,t)}\big[\xi_i^{SN}(d,t^*,\bDmi,M_i)-\mu_i^{SN}(d,t,t^*)\big]\\
& + \frac{\1\{T_i=t^*\}}{p_i^T(t^*)}\big[\zeta_i^{SN}(d,t,\bMmi)-\mu_i^{SN}(d,t,t^*)\big]\\
& + \frac{\1\{D_i=d,T_i=t\}}{e_i^{SN}(d,t)}\frac{r_i^{SN}(\bMmi\mid t^*)}{\rho_i^{SN}(\bMmi\mid d,\bDmi,M_i)}\big[Y_i-\eta_i^{SN}(d,\bDmi,M_i,\bMmi)\big].
\end{align*}
By iterated expectations, one can verify that
\[
\E\!\left[\tau_i^{SN}(d,t,t^*) \mid \bX,\bA \right] = \mu_i^{SN}(d,t,t^*).
\]
Let $\widehat{\tau}_i^{SN}(d,t,t^*)$ denote the score obtained by replacing all nuisance functions in the preceding display with their corresponding estimators. The AIPW estimators of the spillover natural direct and indirect effects are then, respectively,
\begin{gather*}
\widehat{SNDE}(d,t_1,t_0,t^*) \coloneq \frac{1}{j_n}\sum_{i\in\mathcal{J}_n} \big[\widehat{\tau}_i^{SN}(d,t_1,t^*) - \widehat{\tau}_i^{SN}(d,t_0,t^*)\big], \\
\widehat{SNIE}(d,t,t_1^*,t_0^*) \coloneq \frac{1}{j_n}\sum_{i\in\mathcal{J}_n} \big[\widehat{\tau}_i^{SN}(d,t,t_1^*) - \widehat{\tau}_i^{SN}(d,t,t_0^*)\big].
\end{gather*}
It is important to note that $\eta_i^{SN}$, $\rho_i^{SN}$, and $r_i^{SN}$ involve the full vectors $\bDmi$ or $\bMmi$, so the identification formula and estimator above involve high-dimensional joint distributions. In large networks, implementation therefore requires suitable structured modeling or dimension-reduction restrictions on these high-dimensional nuisance functions.

\subsection{Spillover interventional direct/indirect effects}
Identification of the spillover natural direct and indirect effects relies on cross-world independence. To avoid this condition and obtain causal estimands based on a low-dimensional mediator mapping, let $\mathcal{S}_i$ denote the support of $S_i$. For any $t$, let the auxiliary random variable $H_i(t)$ satisfy
\[
H_i(t)\mid\bX,\bA \stackrel{d}{=} S_i\mid T_i=t,\bX,\bA,
\]
and suppose that, conditional on $\bX,\bA$, $H_i(t)$ is independent of all other variables. Thus, $H_i(t)$ is an independent stochastic mediator intervention drawn from the conditional distribution of the mediator mapping under treatment exposure $t$. Define
\begin{equation}
\begin{aligned}
Y_i^{SI}(d_i,t,t^*) \coloneq &\E[Y_i^C(d_i,t,M_i(d_i),H_i(t^*)) \mid \bX,\bA] \\
= & \sum_{m\in\{0,1\}} \sum_{s\in\mathcal{S}_i} Y_i^C(d_i,t,m,s)\P(M_i(d_i)=m\mid\bX,\bA) \P(H_i(t^*)=s\mid\bX,\bA).
\end{aligned}
\end{equation}
Correspondingly, define
\begin{gather*}
SIDE(d,t_1,t_0,t^*) \coloneq \frac{1}{j_n}\sum_{i\in\mathcal{J}_n} \big[Y_i^{SI}(d,t_1,t^*) - Y_i^{SI}(d,t_0,t^*)\big], \\
SIIE(d,t,t_1^*,t_0^*) \coloneq \frac{1}{j_n}\sum_{i\in\mathcal{J}_n} \big[Y_i^{SI}(d,t,t_1^*) - Y_i^{SI}(d,t,t_0^*)\big].
\end{gather*}
Here, $SIDE(d,t_1,t_0,t^*)$ compares treatment exposures $t_1$ and $t_0$ for other individuals while holding the distribution of the stochastic mediator intervention fixed at that associated with treatment exposure $t^*$; $SIIE(d,t,t_1^*,t_0^*)$ measures the average effect of changing the stochastic mediator-intervention distribution from that associated with $t_0^*$ to that associated with $t_1^*$ while holding the treatment exposure of other individuals fixed at $t$.

\begin{assumption}
\label{ch5:ass:spillover-interventional-effects-independence}
For any $i\in\mathcal{N}_n$, $d_i,m_i\in\{0,1\}$, $\bdmi\in\{0,1\}^{n-1}$, and $\bmmi\in\{0,1\}^{n-1}$,
\begin{gather}
Y_i(d_i,\bdmi,m_i,\bmmi) \indep \bD \mid \bX,\bA, \label{ch5:eq:spillover-interventional-treatment-outcome-independence} \\
Y_i(d_i,\bdmi,m_i,\bmmi) \indep \bM \mid \bD,\bX,\bA, \label{ch5:eq:spillover-interventional-mediator-outcome-independence} \\
M_i(d_i) \indep D_i \mid \bX,\bA.
\label{ch5:eq:spillover-interventional-treatment-mediator-independence}
\end{gather}
\end{assumption}

Unlike Assumption \ref{ch5:ass:spillover-natural-effects-independence}, Assumption \ref{ch5:ass:spillover-interventional-effects-independence} does not require cross-world independence. This is because $H_i(t^*)$ is not a natural mediator from another counterfactual world, but an independent draw from the observed distribution $S_i\mid T_i=t^*,\bX,\bA$. Define
\begin{gather*}
\eta_i^{SI}(d,t,s) \coloneq \E(Y_i\mid D_i=d,T_i=t,S_i=s,\bX,\bA), \\
r_i^{SI}(s\mid t) \coloneq \P(S_i=s\mid T_i=t,\bX,\bA), \\
\mu_i^{SI}(d,t,t^*) \coloneq \sum_{s\in\mathcal{S}_i} \eta_i^{SI}(d,t,s) r_i^{SI}(s\mid t^*),
\end{gather*}
and
\begin{gather*}
SIDE^{obs}(d,t_1,t_0,t^*) \coloneq \frac{1}{j_n}\sum_{i\in\mathcal{J}_n} \big[\mu_i^{SI}(d,t_1,t^*) - \mu_i^{SI}(d,t_0,t^*)\big], \\
SIIE^{obs}(d,t,t_1^*,t_0^*) \coloneq \frac{1}{j_n}\sum_{i\in\mathcal{J}_n} \big[\mu_i^{SI}(d,t,t_1^*) - \mu_i^{SI}(d,t,t_0^*)\big].
\end{gather*}

\begin{theorem}
\label{ch5:thm:spillover-interventional-effects-identification}
Under Assumptions \ref{ch5:ass:consistency}, \ref{ch5:ass:error-term-independence}, and \ref{ch5:ass:spillover-interventional-effects-independence}, for any $i,d,t,t^*$,
\[
\mu_i^{SI}(d,t,t^*) = Y_i^{SI}(d,t,t^*).
\]
Consequently,
\begin{gather*}
SIDE^{obs}(d,t_1,t_0,t^*) = SIDE(d,t_1,t_0,t^*), \\
SIIE^{obs}(d,t,t_1^*,t_0^*) = SIIE(d,t,t_1^*,t_0^*).
\end{gather*}
\end{theorem}

Under the corresponding overlap conditions, further define
\begin{gather*}
\rho_i^{SI}(s\mid d,t) \coloneq \P(S_i=s\mid D_i=d,T_i=t,\bX,\bA), \\
e_i^{SI}(d,t) \coloneq \P(D_i=d,T_i=t\mid\bX,\bA), \\
p_i^T(t) \coloneq \P(T_i=t\mid\bX,\bA).
\end{gather*}
By Assumption \ref{ch5:ass:error-term-independence} and the structural model in this paper,
\[
\rho_i^{SI}(s\mid d,t) = r_i^{SI}(s\mid t).
\]
Define the individual-level AIPW score
\begin{align*}
\tau_i^{SI}(d,t,t^*) \coloneq & \mu_i^{SI}(d,t,t^*) + \frac{\1\{T_i=t^*\}}{p_i^T(t^*)} \big[\eta_i^{SI}(d,t,S_i) - \mu_i^{SI}(d,t,t^*) \big] \\
& + \frac{\1\{D_i=d,T_i=t\}}{e_i^{SI}(d,t)} \frac{r_i^{SI}(S_i\mid t^*)} {\rho_i^{SI}(S_i\mid d,t)} \big[Y_i-\eta_i^{SI}(d,t,S_i) \big].
\end{align*}
By iterated expectations, one can verify that
\[
\E\!\left[\tau_i^{SI}(d,t,t^*) \mid \bX,\bA \right] = \mu_i^{SI}(d,t,t^*).
\]
Let $\widehat{\tau}_i^{SI}(d,t,t^*)$ denote the score obtained by replacing all nuisance functions in the preceding display with their corresponding estimators. The AIPW estimators of the spillover interventional direct and indirect effects are then, respectively,
\begin{gather*}
\widehat{SIDE}(d,t_1,t_0,t^*) \coloneq \frac{1}{j_n}\sum_{i\in\mathcal{J}_n} \big[\widehat{\tau}_i^{SI}(d,t_1,t^*) - \widehat{\tau}_i^{SI}(d,t_0,t^*)\big], \\
\widehat{SIIE}(d,t,t_1^*,t_0^*) \coloneq \frac{1}{j_n}\sum_{i\in\mathcal{J}_n} \big[\widehat{\tau}_i^{SI}(d,t,t_1^*) - \widehat{\tau}_i^{SI}(d,t,t_0^*)\big].
\end{gather*}

\subsection{Asymptotic properties of the preceding estimators}
Each estimator above is a linear combination of finitely many individual-level AIPW scores, so their asymptotic theory can be stated in a unified form. Let
\[
\mathcal{Q} \coloneq \{\mathrm{OIDE},\mathrm{OIIE},\mathrm{SCDE}, \mathrm{SNDE},\mathrm{SNIE},\mathrm{SIDE},\mathrm{SIIE}\}.
\]
For any $Q\in\mathcal{Q}$ and any fixed treatment, mediator, and exposure levels, let $\theta_Q$ and $\widehat{\theta}_Q$ denote the corresponding causal estimand and estimator, respectively, and write
\[
\widehat{\theta}_Q = \frac{1}{j_n} \sum_{i\in\mathcal{J}_n} \widehat{\Delta}_{i,Q},
\]
where $\widehat{\Delta}_{i,Q}$ denotes the individual-level score contrast appearing in brackets in the corresponding estimator. Let $\Delta_{i,Q}$ be the score contrast obtained by replacing all estimated nuisance functions with their true values, and define
\[
m_{i,Q} \coloneq \E(\Delta_{i,Q}\mid\bX,\bA), \qquad \varphi_{i,Q} \coloneq \Delta_{i,Q}-m_{i,Q},
\]
and
\[
\sigma_{n,Q}^2 \coloneq \Var\bigg(\frac{1}{\sqrt{j_n}} \sum_{i\in\mathcal{J}_n} \varphi_{i,Q} \mid \bX,\bA \bigg).
\]
By the identification results above,
\[
\theta_Q = \frac{1}{j_n} \sum_{i\in\mathcal{J}_n} m_{i,Q}.
\]

\begin{theorem}
\label{ch5:thm:other-effects-asymptotic-normality}
Suppose that the corresponding identification conditions hold and that the effect-specific analogues of the moment, overlap, approximate-neighborhood-interference, nuisance-function convergence-rate, and $\psi$-dependence conditions used in Theorem \ref{ch5:thm:own-cde-asymptotic-normality} or Theorem \ref{ch5:thm:own-nde-asymptotic-normality} hold, so that
\[
\frac{1}{\sqrt{j_n}} \sum_{i\in\mathcal{J}_n} \big(\widehat{\Delta}_{i,Q} - \Delta_{i,Q} \big) = o_p(1), \qquad \liminf_{n\to\infty} \sigma_{n,Q}^2 > 0.
\]
Then, conditional on $(\bX,\bA)$, for any $Q\in\mathcal{Q}$ and any fixed treatment, mediator, and exposure levels,
\[
\sigma_{n,Q}^{-1} \sqrt{j_n} \big(\widehat{\theta}_Q-\theta_Q \big) \xrightarrow{d} \mathcal{N}(0,1).
\]
\end{theorem}

Let $\widehat{m}_{i,Q}$ denote the conditional mean contrast obtained by replacing the nuisance functions in $m_{i,Q}$ with their corresponding estimators. Replace the residual vector in (\ref{ch5:eq:own-cde-hac-variance-estimator}) by
\[
\widetilde{\boldsymbol{\Delta}}_Q \coloneq \big(\widehat{\Delta}_{i,Q} - \widehat{m}_{i,Q} \big)_{i\in\mathcal{J}_n},
\]
and denote the resulting network HAC variance estimator by $\widehat{\sigma}_Q^2$. Under the effect-specific analogues of the network HAC conditions in Assumption \ref{ch5:ass:own-cde-hac-consistency} or Assumption \ref{ch5:ass:own-natural-effects-hac-consistency}, we have
\[
\widehat{\sigma}_Q^2 = \sigma_{n,Q}^2 + o_p(1),
\]
and therefore
\[
\widehat{\theta}_Q \pm z_{1-\alpha/2} \frac{\widehat{\sigma}_Q}{\sqrt{j_n}}
\]
is an asymptotic $(1-\alpha)$ confidence interval. Because $\widehat{\tau}_i^I(d,d^*,t)=\widehat{\tau}_i^N(d,d^*,t)$, $OIDE$ and $OIIE$ directly inherit the asymptotic theory for the own natural-effect estimators. Likewise, $SCDE$, being a contrast of two $\widehat{\tau}_i^C(d,t,m,s)$ scores, directly inherits the theory for the own controlled-direct-effect estimator. For $SNDE$, $SNIE$, $SIDE$, and $SIIE$, the same linearization and network central-limit-theorem arguments are applied to $\widehat{\tau}_i^{SN}$ and $\widehat{\tau}_i^{SI}$, respectively. In particular, because $SNDE$ and $SNIE$ involve high-dimensional joint distributions of $\bDmi$ and $\bMmi$, their asymptotic results additionally require these high-dimensional nuisance functions to satisfy the preceding overlap, convergence-rate, and weak-dependence conditions under the structured or dimension-reduction model employed.

\section{Supplementary Simulation Experiments}
\label{ch5:subsec:supplementary-simulations}
This section uses the same data-generating processes, nuisance-function estimators, and network HAC variance estimators as in the main text to further examine the finite-sample performance of the own natural indirect effect, the spillover controlled direct effect, and the spillover interventional direct and indirect effects. Tables \ref{ch5:tab:supplementary-random-geometric-nie-scde} and \ref{ch5:tab:supplementary-erdos-renyi-nie-scde} report the results for $ONIE$ and $SCDE$ under random geometric and Erd\H{o}s--R\'enyi graphs, respectively, whereas Tables \ref{ch5:tab:supplementary-random-geometric-side-siie} and \ref{ch5:tab:supplementary-erdos-renyi-side-siie} report the corresponding results for $SIDE$ and $SIIE$. The row labeled ``Truth'' reports the average true effect across the 1,000 simulation replications. Est, OCI, CI, and SE are defined as in the main text. Because the true values of the four effects considered here vary with $\bX,\bA$, OSE is instead defined as the standard deviation of the estimation error across the 1,000 replications. (When the true value is constant, as in the main-text simulations, the standard deviation of the estimation error is identical to the standard deviation of the 1,000 effect estimates.)

\begin{table}[t!]
\small
\renewcommand{\arraystretch}{0.55}
\setlength{\tabcolsep}{3pt}
\centering
\caption{Simulation results: random geometric graph}
\label{ch5:tab:supplementary-random-geometric-nie-scde}
\resizebox{\textwidth}{!}{
\begin{tabular}{lllccccccccc}
\toprule
& & & \multicolumn{3}{c}{$L=1$} & \multicolumn{3}{c}{$L=2$} & \multicolumn{3}{c}{$L=3$} \\
\cmidrule(lr){4-6} \cmidrule(lr){7-9} \cmidrule(lr){10-12}
Estimand & & $n$ & 1000 & 2500 & 4000 & 1000 & 2500 & 4000 & 1000 & 2500 & 4000 \\
\midrule

\multirow{16}{*}{ONIE} &  & Truth & 0.2762 & 0.2769 & 0.2772 & 0.2762 & 0.2769 & 0.2772 & 0.2762 & 0.2769 & 0.2772 \\
\cmidrule(lr){2-12}
 & \multirow{5}{*}{GNN} & Est & 0.3641 & 0.3781 & 0.3870 & 0.2443 & 0.2675 & 0.2755 & 0.2427 & 0.2658 & 0.2736 \\
 &  & OCI & 0.9490 & 0.8090 & 0.6540 & 0.9640 & 0.9460 & 0.9460 & 0.9470 & 0.9410 & 0.9460 \\
 &  & OSE & 0.1517 & 0.0875 & 0.0666 & 0.1252 & 0.0698 & 0.0491 & 0.1216 & 0.0629 & 0.0523 \\
 &  & CI & 0.7940 & 0.7170 & 0.5680 & 0.7370 & 0.9090 & 0.9540 & 0.8170 & 0.9130 & 0.9370 \\
 &  & SE & 0.1093 & 0.0817 & 0.0650 & 0.0714 & 0.0590 & 0.0466 & 0.0855 & 0.0563 & 0.0467 \\
\cmidrule(lr){2-12}
 & \multirow{5}{*}{MLP} & Est & 0.4049 & 0.4001 & 0.4063 & 0.4047 & 0.3990 & 0.4061 & 0.4085 & 0.3997 & 0.4033 \\
 &  & OCI & 0.7530 & 0.7120 & 0.4880 & 0.8640 & 0.7620 & 0.5430 & 0.8730 & 0.7670 & 0.5480 \\
 &  & OSE & 0.0970 & 0.0800 & 0.0639 & 0.1266 & 0.0890 & 0.0687 & 0.1390 & 0.0897 & 0.0674 \\
 &  & CI & 0.6700 & 0.5790 & 0.4120 & 0.7490 & 0.6320 & 0.4480 & 0.7650 & 0.6060 & 0.4550 \\
 &  & SE & 0.0922 & 0.0742 & 0.0622 & 0.1131 & 0.0816 & 0.0660 & 0.1185 & 0.0828 & 0.0656 \\
\cmidrule(lr){2-12}
 & \multirow{5}{*}{RF} & Est & 0.5708 & 0.5378 & 0.5332 &  &  &  &  &  &  \\
 &  & OCI & 0.3640 & 0.0610 & 0.0040 &  &  &  &  &  &  \\
 &  & OSE & 0.1247 & 0.0733 & 0.0562 &  &  &  &  &  &  \\
 &  & CI & 0.0280 & 0.0000 & 0.0000 &  &  &  &  &  &  \\
 &  & SE & 0.0446 & 0.0276 & 0.0216 &  &  &  &  &  &  \\
\midrule

\multirow{16}{*}{SCDE} &  & Truth & 0.1382 & 0.1345 & 0.1332 & 0.1382 & 0.1345 & 0.1332 & 0.1382 & 0.1345 & 0.1332 \\
\cmidrule(lr){2-12}
 & \multirow{5}{*}{GNN} & Est & 0.3684 & 0.3841 & 0.3769 & 0.1629 & 0.1682 & 0.1579 & 0.1470 & 0.1471 & 0.1392 \\
 &  & OCI & 0.7950 & 0.5200 & 0.3620 & 0.9490 & 0.9340 & 0.9350 & 0.9610 & 0.9410 & 0.9500 \\
 &  & OSE & 0.1949 & 0.1316 & 0.1084 & 0.1498 & 0.1101 & 0.0913 & 0.1623 & 0.1118 & 0.0966 \\
 &  & CI & 0.6280 & 0.4540 & 0.3080 & 0.7610 & 0.8690 & 0.8730 & 0.7820 & 0.8790 & 0.8850 \\
 &  & SE & 0.1543 & 0.1220 & 0.1004 & 0.0945 & 0.0904 & 0.0786 & 0.1035 & 0.0907 & 0.0781 \\
\cmidrule(lr){2-12}
 & \multirow{5}{*}{MLP} & Est & 0.4243 & 0.4126 & 0.3988 & 0.4131 & 0.4054 & 0.3925 & 0.4045 & 0.4023 & 0.3884 \\
 &  & OCI & 0.7280 & 0.4800 & 0.3460 & 0.7670 & 0.5350 & 0.3950 & 0.7810 & 0.5420 & 0.3980 \\
 &  & OSE & 0.2144 & 0.1400 & 0.1143 & 0.2232 & 0.1442 & 0.1179 & 0.2253 & 0.1474 & 0.1176 \\
 &  & CI & 0.6580 & 0.4200 & 0.3010 & 0.6820 & 0.4680 & 0.3260 & 0.7000 & 0.4820 & 0.3440 \\
 &  & SE & 0.1966 & 0.1320 & 0.1063 & 0.2014 & 0.1351 & 0.1082 & 0.2029 & 0.1361 & 0.1082 \\
\cmidrule(lr){2-12}
 & \multirow{5}{*}{RF} & Est & 0.3836 & 0.3524 & 0.3300 &  &  &  &  &  &  \\
 &  & OCI & 0.8250 & 0.7940 & 0.7790 &  &  &  &  &  &  \\
 &  & OSE & 0.2572 & 0.1926 & 0.1663 &  &  &  &  &  &  \\
 &  & CI & 0.3980 & 0.2900 & 0.2490 &  &  &  &  &  &  \\
 &  & SE & 0.0773 & 0.0483 & 0.0378 &  &  &  &  &  &  \\
\bottomrule
\end{tabular}
}
\end{table}

\begin{table}[t!]
\small
\renewcommand{\arraystretch}{0.55}
\setlength{\tabcolsep}{3pt}
\centering
\caption{Simulation results: Erd\H{o}s--R\'enyi graph}
\label{ch5:tab:supplementary-erdos-renyi-nie-scde}
\resizebox{\textwidth}{!}{
\begin{tabular}{lllccccccccc}
\toprule
& & & \multicolumn{3}{c}{$L=1$} & \multicolumn{3}{c}{$L=2$} & \multicolumn{3}{c}{$L=3$} \\
\cmidrule(lr){4-6} \cmidrule(lr){7-9} \cmidrule(lr){10-12}
Estimand & & $n$ & 1000 & 2500 & 4000 & 1000 & 2500 & 4000 & 1000 & 2500 & 4000 \\
\midrule

\multirow{16}{*}{ONIE} &  & Truth & 0.3014 & 0.3017 & 0.3018 & 0.3014 & 0.3017 & 0.3018 & 0.3014 & 0.3017 & 0.3018 \\
\cmidrule(lr){2-12}
 & \multirow{5}{*}{GNN} & Est & 0.3542 & 0.3696 & 0.3730 & 0.2779 & 0.3010 & 0.3029 & 0.2882 & 0.2998 & 0.3053 \\
 &  & OCI & 0.9440 & 0.8230 & 0.7430 & 0.9700 & 0.9480 & 0.9470 & 0.9560 & 0.9450 & 0.9520 \\
 &  & OSE & 0.1324 & 0.0657 & 0.0517 & 0.1409 & 0.0641 & 0.0474 & 0.1217 & 0.0645 & 0.0498 \\
 &  & CI & 0.8680 & 0.7870 & 0.6810 & 0.7570 & 0.9470 & 0.9520 & 0.8470 & 0.9390 & 0.9320 \\
 &  & SE & 0.1009 & 0.0634 & 0.0493 & 0.0776 & 0.0585 & 0.0457 & 0.0869 & 0.0582 & 0.0466 \\
\cmidrule(lr){2-12}
 & \multirow{5}{*}{MLP} & Est & 0.3761 & 0.3760 & 0.3768 & 0.3752 & 0.3742 & 0.3756 & 0.3717 & 0.3737 & 0.3771 \\
 &  & OCI & 0.8430 & 0.6990 & 0.6050 & 0.8530 & 0.7410 & 0.6500 & 0.8860 & 0.7600 & 0.6370 \\
 &  & OSE & 0.0772 & 0.0513 & 0.0442 & 0.0861 & 0.0542 & 0.0460 & 0.0884 & 0.0554 & 0.0458 \\
 &  & CI & 0.7850 & 0.6900 & 0.5940 & 0.8260 & 0.7370 & 0.6230 & 0.8490 & 0.7390 & 0.5990 \\
 &  & SE & 0.0695 & 0.0522 & 0.0443 & 0.0787 & 0.0548 & 0.0458 & 0.0812 & 0.0561 & 0.0455 \\
\cmidrule(lr){2-12}
 & \multirow{5}{*}{RF} & Est & 0.4715 & 0.4485 & 0.4399 &  &  &  &  &  &  \\
 &  & OCI & 0.2910 & 0.0570 & 0.0210 &  &  &  &  &  &  \\
 &  & OSE & 0.0681 & 0.0419 & 0.0351 &  &  &  &  &  &  \\
 &  & CI & 0.0670 & 0.0070 & 0.0020 &  &  &  &  &  &  \\
 &  & SE & 0.0383 & 0.0236 & 0.0187 &  &  &  &  &  &  \\
\midrule

\multirow{16}{*}{SCDE} &  & Truth & 0.0350 & 0.0343 & 0.0342 & 0.0350 & 0.0343 & 0.0342 & 0.0350 & 0.0343 & 0.0342 \\
\cmidrule(lr){2-12}
 & \multirow{5}{*}{GNN} & Est & 0.0763 & 0.0839 & 0.0776 & 0.0377 & 0.0386 & 0.0357 & 0.0470 & 0.0477 & 0.0364 \\
 &  & OCI & 0.9460 & 0.9300 & 0.9260 & 0.9470 & 0.9430 & 0.9470 & 0.9460 & 0.9410 & 0.9530 \\
 &  & OSE & 0.1862 & 0.1258 & 0.0987 & 0.1567 & 0.1108 & 0.0919 & 0.1621 & 0.1149 & 0.0952 \\
 &  & CI & 0.9060 & 0.9110 & 0.9190 & 0.8880 & 0.9250 & 0.9460 & 0.8920 & 0.9230 & 0.9450 \\
 &  & SE & 0.1602 & 0.1185 & 0.0970 & 0.1266 & 0.1056 & 0.0899 & 0.1344 & 0.1072 & 0.0906 \\
\cmidrule(lr){2-12}
 & \multirow{5}{*}{MLP} & Est & 0.1002 & 0.0947 & 0.0798 & 0.0833 & 0.0877 & 0.0741 & 0.0836 & 0.0849 & 0.0756 \\
 &  & OCI & 0.9410 & 0.9170 & 0.9260 & 0.9440 & 0.9280 & 0.9380 & 0.9350 & 0.9280 & 0.9360 \\
 &  & OSE & 0.1955 & 0.1304 & 0.1024 & 0.2003 & 0.1348 & 0.1043 & 0.2045 & 0.1334 & 0.1034 \\
 &  & CI & 0.9170 & 0.8920 & 0.9050 & 0.9260 & 0.9080 & 0.9230 & 0.9170 & 0.9120 & 0.9220 \\
 &  & SE & 0.1818 & 0.1229 & 0.0990 & 0.1876 & 0.1256 & 0.1012 & 0.1873 & 0.1259 & 0.1007 \\
\cmidrule(lr){2-12}
 & \multirow{5}{*}{RF} & Est & 0.0611 & 0.0660 & 0.0618 &  &  &  &  &  &  \\
 &  & OCI & 0.9470 & 0.9150 & 0.9060 &  &  &  &  &  &  \\
 &  & OSE & 0.0909 & 0.0586 & 0.0480 &  &  &  &  &  &  \\
 &  & CI & 0.8470 & 0.8030 & 0.7800 &  &  &  &  &  &  \\
 &  & SE & 0.0676 & 0.0416 & 0.0330 &  &  &  &  &  &  \\
\bottomrule
\end{tabular}
}
\end{table}

\begin{table}[t!]
\small
\renewcommand{\arraystretch}{0.55}
\setlength{\tabcolsep}{3pt}
\centering
\caption{Simulation results: random geometric graph}
\label{ch5:tab:supplementary-random-geometric-side-siie}
\resizebox{\textwidth}{!}{
\begin{tabular}{lllccccccccc}
\toprule
& & & \multicolumn{3}{c}{$L=1$} & \multicolumn{3}{c}{$L=2$} & \multicolumn{3}{c}{$L=3$} \\
\cmidrule(lr){4-6} \cmidrule(lr){7-9} \cmidrule(lr){10-12}
Estimand & & $n$ & 1000 & 2500 & 4000 & 1000 & 2500 & 4000 & 1000 & 2500 & 4000 \\
\midrule

\multirow{16}{*}{SIDE} &  & Truth & 0.1611 & 0.1569 & 0.1553 & 0.1611 & 0.1569 & 0.1553 & 0.1611 & 0.1569 & 0.1553 \\
\cmidrule(lr){2-12}
 & \multirow{5}{*}{GNN} & Est & 0.4778 & 0.4793 & 0.4719 & 0.1834 & 0.1960 & 0.1903 & 0.1531 & 0.1604 & 0.1586 \\
 &  & OCI & 0.7530 & 0.3550 & 0.1940 & 0.9500 & 0.9410 & 0.9390 & 0.9540 & 0.9490 & 0.9460 \\
 &  & OSE & 0.2306 & 0.1400 & 0.1140 & 0.1805 & 0.1153 & 0.0991 & 0.2172 & 0.1331 & 0.1095 \\
 &  & CI & 0.5190 & 0.3090 & 0.1660 & 0.7880 & 0.8710 & 0.8780 & 0.8240 & 0.8710 & 0.8740 \\
 &  & SE & 0.1837 & 0.1322 & 0.1064 & 0.1163 & 0.0976 & 0.0827 & 0.1395 & 0.1017 & 0.0841 \\
\cmidrule(lr){2-12}
 & \multirow{5}{*}{MLP} & Est & 0.5131 & 0.5090 & 0.5026 & 0.5119 & 0.5050 & 0.4986 & 0.5110 & 0.4974 & 0.4944 \\
 &  & OCI & 0.6090 & 0.2850 & 0.1340 & 0.6780 & 0.3510 & 0.1630 & 0.7150 & 0.3730 & 0.1810 \\
 &  & OSE & 0.2101 & 0.1412 & 0.1133 & 0.2350 & 0.1495 & 0.1181 & 0.2424 & 0.1511 & 0.1197 \\
 &  & CI & 0.5180 & 0.2800 & 0.1260 & 0.6020 & 0.3290 & 0.1510 & 0.6440 & 0.3480 & 0.1640 \\
 &  & SE & 0.1991 & 0.1406 & 0.1116 & 0.2257 & 0.1475 & 0.1151 & 0.2362 & 0.1497 & 0.1152 \\
\cmidrule(lr){2-12}
 & \multirow{5}{*}{RF} & Est & 0.5284 & 0.4787 & 0.4545 &  &  &  &  &  &  \\
 &  & OCI & 0.7530 & 0.6920 & 0.6270 &  &  &  &  &  &  \\
 &  & OSE & 0.2859 & 0.2065 & 0.1657 &  &  &  &  &  &  \\
 &  & CI & 0.2850 & 0.1200 & 0.0450 &  &  &  &  &  &  \\
 &  & SE & 0.0829 & 0.0526 & 0.0411 &  &  &  &  &  &  \\
\midrule

\multirow{16}{*}{SIIE} &  & Truth & 0.0191 & 0.0185 & 0.0183 & 0.0191 & 0.0185 & 0.0183 & 0.0191 & 0.0185 & 0.0183 \\
\cmidrule(lr){2-12}
 & \multirow{5}{*}{GNN} & Est & 0.1387 & 0.1588 & 0.1657 & 0.0145 & 0.0173 & 0.0212 & 0.0055 & 0.0079 & 0.0124 \\
 &  & OCI & 0.9110 & 0.5490 & 0.2880 & 0.9630 & 0.9600 & 0.9410 & 0.9670 & 0.9500 & 0.9530 \\
 &  & OSE & 0.1409 & 0.0734 & 0.0598 & 0.0633 & 0.0405 & 0.0294 & 0.0854 & 0.0429 & 0.0339 \\
 &  & CI & 0.6190 & 0.4170 & 0.2120 & 0.8590 & 0.9410 & 0.9370 & 0.8790 & 0.9130 & 0.9100 \\
 &  & SE & 0.1006 & 0.0692 & 0.0543 & 0.0392 & 0.0340 & 0.0271 & 0.0491 & 0.0329 & 0.0261 \\
\cmidrule(lr){2-12}
 & \multirow{5}{*}{MLP} & Est & 0.2137 & 0.2067 & 0.2049 & 0.2075 & 0.2015 & 0.2009 & 0.1997 & 0.2037 & 0.1973 \\
 &  & OCI & 0.6860 & 0.2880 & 0.1210 & 0.8070 & 0.4300 & 0.1920 & 0.8570 & 0.5000 & 0.2020 \\
 &  & OSE & 0.1212 & 0.0744 & 0.0607 & 0.1515 & 0.0842 & 0.0655 & 0.1669 & 0.0921 & 0.0667 \\
 &  & CI & 0.4730 & 0.2830 & 0.1150 & 0.6510 & 0.3610 & 0.1830 & 0.6700 & 0.3660 & 0.1820 \\
 &  & SE & 0.1053 & 0.0766 & 0.0590 & 0.1358 & 0.0843 & 0.0636 & 0.1455 & 0.0866 & 0.0637 \\
\cmidrule(lr){2-12}
 & \multirow{5}{*}{RF} & Est & 0.3516 & 0.3480 & 0.3433 &  &  &  &  &  &  \\
 &  & OCI & 0.3810 & 0.1210 & 0.0370 &  &  &  &  &  &  \\
 &  & OSE & 0.1466 & 0.1019 & 0.0817 &  &  &  &  &  &  \\
 &  & CI & 0.0380 & 0.0020 & 0.0000 &  &  &  &  &  &  \\
 &  & SE & 0.0478 & 0.0307 & 0.0241 &  &  &  &  &  &  \\
\bottomrule
\end{tabular}
}
\end{table}

\begin{table}[t!]
\small
\renewcommand{\arraystretch}{0.55}
\setlength{\tabcolsep}{3pt}
\centering
\caption{Simulation results: Erd\H{o}s--R\'enyi graph}
\label{ch5:tab:supplementary-erdos-renyi-side-siie}
\resizebox{\textwidth}{!}{
\begin{tabular}{lllccccccccc}
\toprule
& & & \multicolumn{3}{c}{$L=1$} & \multicolumn{3}{c}{$L=2$} & \multicolumn{3}{c}{$L=3$} \\
\cmidrule(lr){4-6} \cmidrule(lr){7-9} \cmidrule(lr){10-12}
Estimand & & $n$ & 1000 & 2500 & 4000 & 1000 & 2500 & 4000 & 1000 & 2500 & 4000 \\
\midrule

\multirow{16}{*}{SIDE} &  & Truth & 0.0383 & 0.0374 & 0.0373 & 0.0383 & 0.0374 & 0.0373 & 0.0383 & 0.0374 & 0.0373 \\
\cmidrule(lr){2-12}
 & \multirow{5}{*}{GNN} & Est & 0.0966 & 0.0962 & 0.0927 & 0.0439 & 0.0444 & 0.0439 & 0.0462 & 0.0444 & 0.0419 \\
 &  & OCI & 0.9360 & 0.9050 & 0.8790 & 0.9550 & 0.9450 & 0.9440 & 0.9510 & 0.9510 & 0.9480 \\
 &  & OSE & 0.1773 & 0.1082 & 0.0795 & 0.1756 & 0.1069 & 0.0820 & 0.1924 & 0.1129 & 0.0923 \\
 &  & CI & 0.9070 & 0.8960 & 0.9000 & 0.8740 & 0.9340 & 0.9530 & 0.9060 & 0.9380 & 0.9370 \\
 &  & SE & 0.1598 & 0.1041 & 0.0837 & 0.1343 & 0.1002 & 0.0812 & 0.1497 & 0.1033 & 0.0842 \\
\cmidrule(lr){2-12}
 & \multirow{5}{*}{MLP} & Est & 0.0994 & 0.0960 & 0.0934 & 0.1004 & 0.0970 & 0.0936 & 0.1022 & 0.0932 & 0.0932 \\
 &  & OCI & 0.9310 & 0.9110 & 0.8910 & 0.9350 & 0.9110 & 0.8890 & 0.9290 & 0.9230 & 0.8780 \\
 &  & OSE & 0.1496 & 0.1014 & 0.0796 & 0.1624 & 0.1024 & 0.0799 & 0.1643 & 0.1046 & 0.0792 \\
 &  & CI & 0.9220 & 0.9040 & 0.8940 & 0.9190 & 0.9130 & 0.8960 & 0.9270 & 0.9040 & 0.8890 \\
 &  & SE & 0.1495 & 0.0996 & 0.0815 & 0.1564 & 0.1016 & 0.0822 & 0.1611 & 0.1021 & 0.0816 \\
\cmidrule(lr){2-12}
 & \multirow{5}{*}{RF} & Est & 0.1022 & 0.0956 & 0.0923 &  &  &  &  &  &  \\
 &  & OCI & 0.8720 & 0.8250 & 0.7590 &  &  &  &  &  &  \\
 &  & OSE & 0.0820 & 0.0552 & 0.0433 &  &  &  &  &  &  \\
 &  & CI & 0.7600 & 0.6660 & 0.6070 &  &  &  &  &  &  \\
 &  & SE & 0.0649 & 0.0405 & 0.0326 &  &  &  &  &  &  \\
\midrule

\multirow{16}{*}{SIIE} &  & Truth & 0.0074 & 0.0073 & 0.0073 & 0.0074 & 0.0073 & 0.0073 & 0.0074 & 0.0073 & 0.0073 \\
\cmidrule(lr){2-12}
 & \multirow{5}{*}{GNN} & Est & 0.0270 & 0.0253 & 0.0271 & 0.0155 & 0.0094 & 0.0090 & 0.0163 & 0.0106 & 0.0108 \\
 &  & OCI & 0.9420 & 0.9240 & 0.9070 & 0.9530 & 0.9570 & 0.9490 & 0.9470 & 0.9510 & 0.9510 \\
 &  & OSE & 0.0823 & 0.0418 & 0.0325 & 0.0741 & 0.0442 & 0.0327 & 0.0773 & 0.0450 & 0.0344 \\
 &  & CI & 0.9080 & 0.9070 & 0.9030 & 0.9140 & 0.9540 & 0.9400 & 0.9240 & 0.9320 & 0.9380 \\
 &  & SE & 0.0663 & 0.0392 & 0.0316 & 0.0513 & 0.0385 & 0.0303 & 0.0569 & 0.0389 & 0.0309 \\
\cmidrule(lr){2-12}
 & \multirow{5}{*}{MLP} & Est & 0.0347 & 0.0259 & 0.0264 & 0.0299 & 0.0249 & 0.0260 & 0.0270 & 0.0264 & 0.0257 \\
 &  & OCI & 0.9200 & 0.9230 & 0.9120 & 0.9440 & 0.9240 & 0.9160 & 0.9440 & 0.9270 & 0.9160 \\
 &  & OSE & 0.0604 & 0.0383 & 0.0309 & 0.0669 & 0.0406 & 0.0324 & 0.0757 & 0.0420 & 0.0320 \\
 &  & CI & 0.8500 & 0.9170 & 0.8810 & 0.9110 & 0.9190 & 0.8880 & 0.9110 & 0.9130 & 0.8910 \\
 &  & SE & 0.0521 & 0.0370 & 0.0304 & 0.0619 & 0.0394 & 0.0315 & 0.0666 & 0.0406 & 0.0314 \\
\cmidrule(lr){2-12}
 & \multirow{5}{*}{RF} & Est & 0.0343 & 0.0312 & 0.0306 &  &  &  &  &  &  \\
 &  & OCI & 0.8750 & 0.7910 & 0.7090 &  &  &  &  &  &  \\
 &  & OSE & 0.0299 & 0.0191 & 0.0158 &  &  &  &  &  &  \\
 &  & CI & 0.7320 & 0.5680 & 0.4940 &  &  &  &  &  &  \\
 &  & SE & 0.0214 & 0.0138 & 0.0113 &  &  &  &  &  &  \\
\bottomrule
\end{tabular}
}
\end{table}

The results in Tables \ref{ch5:tab:supplementary-random-geometric-nie-scde}--\ref{ch5:tab:supplementary-erdos-renyi-side-siie} show that the main findings from the simulation study in the main text extend to the own natural indirect effect and the three spillover effects. Overall, when the GNN has $L=2$ or $L=3$ layers, its estimates are substantially closer to the corresponding true values under both network designs, and OSE declines steadily as the sample size increases. In contrast, the $L=1$ GNN typically exhibits a pronounced positive bias that does not vanish with the sample size, particularly for $SCDE$, $SIDE$, and $SIIE$ under the random geometric graph. Although the standard deviation of the estimation error still decreases with $n$, the center of the estimator does not move correspondingly toward the truth, causing OCI to deteriorate markedly as the sample size grows. This pattern is consistent with the findings for $OCDE$ and $ONDE$ in the main text.

More specifically, for $ONIE$, the $L=2$ and $L=3$ GNNs perform well under both network designs. A similar pattern appears for $SIIE$, for which $L=2$ generally yields smaller point-estimation bias and lower OSE. For the spillover direct effects, finite-sample performance is more sensitive to the choice of depth. Under the Erd\H{o}s--R\'enyi graph, $L=2$ generally gives estimates of $SCDE$ and $SIDE$ closest to the truth, whereas under the random geometric graph, $L=3$ produces smaller biases for these two estimands than $L=2$. Thus, consistent with the main-text results, moderately increasing the number of GNN aggregation layers generally does not introduce substantial additional bias and may improve finite-sample performance for some more complex spillover estimands. Nevertheless, OSE is slightly larger for $L=3$ than for $L=2$ in most settings, suggesting that aggregating information beyond the effective dependence range can increase both the complexity of nuisance-function estimation and sampling variability. Overall, matching the GNN aggregation range to the plausible range of network dependence remains a prudent choice, while modest over-specification of the depth is typically more robust than under-specification.

The conventional machine-learning methods display the same qualitative patterns as in the main text. The MLP results vary little across different numbers of hidden layers, and its estimates are systematically above the truth for nearly all four estimands; this bias shows little attenuation as the sample size increases. This finding indicates that increasing neural-network depth alone cannot substitute for information propagation along network edges. The RF generally exhibits even larger point-estimation bias, especially for $ONIE$, $SIDE$, and $SIIE$ under the random geometric graph, where its estimates remain far from the truth even at $n=4000$. Although the OSE of MLP or RF is smaller than that of the GNN in a few settings, lower sampling variability does not translate into more reliable inference. Under the Erd\H{o}s--R\'enyi graph, the sampling variability of some spillover effects is large relative to the magnitude of the true effect, so OCI for MLP and RF can occasionally remain relatively high. Their point estimates, however, remain clearly displaced from the truth, and coverage typically decreases as the sample size grows. These apparently favorable coverage rates therefore mainly reflect wide sampling distributions that mask point-estimation bias, rather than evidence that conventional methods have adequately captured the network-confounding structure.

Finally, comparison of SE, OSE, CI, and OCI shows that the network HAC variance estimator exhibits essentially the same finite-sample behavior for these four supplementary causal estimands as in the main text. For the $L=2$ and $L=3$ GNNs, OCI remains close to the nominal coverage probability across both network designs and all sample sizes, indicating satisfactory centering of the point estimator and a good normal approximation. By contrast, CI is typically below OCI in smaller samples, primarily because the network HAC standard error somewhat understates the true sampling variability of the estimation error. As the sample size grows, the gap between SE and OSE generally narrows, and CI moves correspondingly toward the nominal level. At $n=4000$, for the Erd\H{o}s--R\'enyi graph and $L=2$ or $L=3$, SE is about $90\%$--$99\%$ of OSE and CI ranges from $0.932$ to $0.953$, indicating that the network HAC variance estimator is already quite accurate. Convergence is slower under the random geometric graph: at $n=4000$, SE is about $77\%$--$95\%$ of OSE. In this case, CI is already close to the nominal level for $ONIE$ and for $SIIE$ with $L=2$, whereas $SCDE$, $SIDE$, and $SIIE$ with $L=3$ still exhibit some undercoverage. Because OCI remains close to $0.95$ in these settings, the undercoverage primarily reflects finite-sample downward bias in the network HAC standard errors rather than material centering bias of the point estimator or a failure of the normal approximation. For MLP, SE and OSE are generally close and the gap between CI and OCI is relatively small; however, for estimands with substantial bias, both CI and OCI fall increasingly below the nominal level as the sample size grows, indicating that undercoverage is driven mainly by point-estimation bias that does not vanish with sample size rather than by variance estimation. RF, in most settings, suffers from both point-estimation bias and variance underestimation: OCI is already below the nominal level, while SE substantially understates OSE, causing CI to fall even further below OCI.

In summary, the supplementary simulations show that the proposed method remains effective for more complex spillover causal estimands, and the qualitative conclusions from the main text regarding the aggregation range, the failure modes of the benchmark methods, and the quality of network HAC inference continue to hold for these estimands.

\section{Proofs}
\subsection{Proof of Theorem \ref{ch5:thm:own-cde-identification}}
For any $i,d_i,t,m_i,s$, by the law of total probability and noting that $T_i$ and $S_i$ are determined by $\bDmi$ and $\bMmi$, respectively (so the sums need only range over $(\bdmi,\bmmi)$ satisfying $T_i(\bdmi)=t$ and $S_i(\bmmi)=s$, in which case conditioning additionally on $T_i=t$ and $S_i=s$ is redundant), we have
\begin{align*}
& \mu_i^C(d_i,t,m_i,s) \\
= & \sum_{\bdmi:T_i(\bdmi)=t} \sum_{\bmmi:S_i(\bmmi)=s} \E(Y_i \mid D_i=d_i, \bDmi=\bdmi, M_i=m_i, \bMmi=\bmmi, \bX, \bA)\\
& \hspace{2cm} \P(\bDmi=\bdmi,\bMmi=\bmmi \mid D_i=d_i, T_i=t, M_i=m_i, S_i=s, \bX, \bA).
\end{align*}
By consistency, $Y_i=Y_i(D_i,\bDmi,M_i,\bMmi)$, and Assumption \ref{ch5:ass:own-cde-conditional-independence},
\[
\begin{aligned}
  & \E(Y_i \mid D_i=d_i, \bDmi=\bdmi, M_i=m_i, \bMmi=\bmmi, \bX, \bA)\\
  & = \E[Y_i(d_i,\bdmi,m_i,\bmmi) \mid \bX, \bA].
\end{aligned}
\]
On the other hand, Assumption \ref{ch5:ass:error-term-independence} together with models (\ref{ch5:eq:potential-mediator-model}) and (\ref{ch5:eq:treatment-assignment-model}) implies $(D_i,M_i) \indep (\bDmi,\bMmi) \mid \bX,\bA$. Combining this with the definitions of the exposure mappings $T_i$ and $S_i$ yields $(D_i,M_i) \indep (\bDmi,\bMmi) \mid T_i,S_i,\bX,\bA$, so
\begin{align*}
  & \P(\bDmi=\bdmi,\bMmi=\bmmi \mid D_i=d_i, T_i=t, M_i=m_i, S_i=s, \bX, \bA) \\
  = & \P(\bDmi=\bdmi,\bMmi=\bmmi \mid T_i=t, S_i=s, \bX, \bA).
\end{align*}
Substituting the preceding two equalities gives $\mu_i^C(d_i,t,m_i,s)=Y_i^C(d_i,t,m_i,s)$, and hence
\[
  OCDE^{obs}(t,m,s) = \frac{1}{j_n}\sum_{i\in\mathcal{J}_n} \big[Y^C_i(1,t,m,s) - Y^C_i(0,t,m,s)\big] = OCDE(t,m,s).
\]
This completes the proof.

\subsection{Proof of Theorem \ref{ch5:thm:own-natural-effects-identification}}
For any $i,d_i,t,m$, the law of total probability, together with the fact that $T_i$ is determined by $\bDmi$, gives
\begin{align*}
& \E(Y_i \mid D_i=d_i,T_i=t,M_i=m,\bX,\bA) \\
= & \sum_{\bdmi:T_i(\bdmi)=t} \E(Y_i \mid D_i=d_i,\bDmi=\bdmi,M_i=m,\bX,\bA) \\
& \hspace{3cm}\P(\bDmi=\bdmi \mid D_i=d_i,T_i=t,M_i=m,\bX,\bA).
\end{align*}
By consistency ($Y_i=Y_i(D_i,\bDmi,M_i,\bMmi)$ and $\bMmi=\bMmi(\bDmi)$) and conditions (\ref{ch5:eq:own-natural-treatment-outcome-independence}) and (\ref{ch5:eq:own-natural-mediator-outcome-independence}) in Assumption \ref{ch5:ass:own-natural-effects-conditional-independence}, we have
\begin{equation}
  \E(Y_i \mid D_i=d_i,\bDmi=\bdmi,M_i=m_i,\bX,\bA) = \E[Y_i(d_i,\bdmi,m_i,\bMmi(\bdmi)) \mid \bX,\bA].
 \label{ch5:eq:proof-own-natural-outcome-regression-identification}
\end{equation}
Assumption \ref{ch5:ass:error-term-independence} together with models (\ref{ch5:eq:potential-mediator-model}) and (\ref{ch5:eq:treatment-assignment-model}) implies $(D_i,M_i) \indep \bDmi \mid \bX,\bA$. Combining this with the definition of the exposure mapping $T_i$ yields $(D_i,M_i) \indep \bDmi \mid T_i,\bX,\bA$, and therefore
\begin{equation}
  \P(\bDmi=\bdmi \mid D_i=d_i, T_i=t, M_i=m_i, \bX, \bA) = \P(\bDmi=\bdmi \mid T_i=t, \bX, \bA).
  \label{ch5:eq:proof-own-natural-peer-treatment-distribution}
\end{equation}
Moreover, by consistency, $M_i=M_i(D_i)$, and condition (\ref{ch5:eq:own-natural-treatment-mediator-independence}) in Assumption \ref{ch5:ass:own-natural-effects-conditional-independence},
\begin{equation}
  \P(M_i = m_i \mid D_i = d_i^*,\bX,\bA) = \P(M_i(d_i^*)=m_i \mid \bX,\bA).
  \label{ch5:eq:proof-own-natural-mediator-distribution}
\end{equation}
Substituting (\ref{ch5:eq:proof-own-natural-outcome-regression-identification}), (\ref{ch5:eq:proof-own-natural-peer-treatment-distribution}), and (\ref{ch5:eq:proof-own-natural-mediator-distribution}) into the definition of $\mu_i^N(d_i,d_i^*,t)$ yields
\begin{align*}
  \mu_i^N(d_i,d_i^*,t) & = \sum_{m_i\in\{0,1\}} \sum_{\bdmi:T_i(\bdmi)=t} \E[Y_i(d_i,\bdmi,m_i,\bMmi(\bdmi)) \mid \bX,\bA] \\
  & \hspace{2cm}\P(\bDmi=\bdmi \mid T_i=t, \bX, \bA)\P(M_i(d_i^*)=m_i \mid \bX,\bA).
\end{align*}
By condition (\ref{ch5:eq:own-natural-cross-world-independence}) in Assumption \ref{ch5:ass:own-natural-effects-conditional-independence},
\[
\begin{aligned}
  & \E[Y_i(d_i,\bdmi,m_i,\bMmi(\bdmi)) \mid \bX,\bA] \\
  & \hspace{3cm} = \E[Y_i(d_i,\bdmi,M_i(d_i^*),\bMmi(\bdmi)) \mid M_i(d_i^*)=m_i,\bX,\bA],
\end{aligned}
\]
Thus, the weighted sum over $m_i$ in the preceding expression is precisely the expectation with respect to the conditional distribution of $M_i(d_i^*)$, so that
\begin{align*}
\mu_i^N(d_i,d_i^*,t) & = \sum_{\bdmi:T_i(\bdmi)=t} \E[Y_i(d_i,\bdmi,M_i(d_i^*),\bMmi(\bdmi)) \mid \bX,\bA]\\
&\hspace{5cm} \P(\bDmi=\bdmi \mid T_i=t, \bX, \bA) \\
& = Y_i^N(d_i,d_i^*,t).
\end{align*}
This completes the proof.

\subsection{Proof of Theorem \ref{ch5:thm:own-interventional-effects-identification}}
Following the proof of Theorem \ref{ch5:thm:own-natural-effects-identification}, we obtain
\begin{align*}
  \mu_i^I(d_i,d_i^*,t) & = \sum_{m_i\in\{0,1\}} \sum_{\bdmi:T_i(\bdmi)=t} \E[Y_i(d_i,\bdmi,m_i,\bMmi(\bdmi)) \mid \bX,\bA] \\
  & \hspace{2cm}\P(\bDmi=\bdmi \mid T_i=t, \bX, \bA)\P(M_i(d_i^*)=m_i \mid \bX,\bA).
\end{align*}
By the definition of $G_i(d)$, $\P(M_i(d_i^*)=m_i \mid \bX,\bA)=\P(G_i(d_i^*)=m_i \mid \bX,\bA)$. Substituting this equality into the right-hand side above gives $\mu_i^I(d_i,d_i^*,t)=Y_i^I(d_i,d_i^*,t)$. This completes the proof.

\subsection{Proof of Theorem \ref{ch5:thm:spillover-natural-effects-identification}}
By consistency, $Y_i=Y_i(D_i,\bDmi,M_i,\bMmi)$, and conditions (\ref{ch5:eq:spillover-natural-treatment-outcome-independence}) and (\ref{ch5:eq:spillover-natural-mediator-outcome-independence}) in Assumption \ref{ch5:ass:spillover-natural-effects-independence}, we have
\begin{equation}
\begin{aligned}
& \E[Y_i \mid D_i=d_i,\bDmi=\bdmi,M_i=m_i,\bMmi=\bmmi,\bX,\bA]\\
& = \E[Y_i(d_i,\bdmi,m_i,\bmmi) \mid \bX,\bA];
\end{aligned}
\label{ch5:eq:proof-spillover-natural-outcome-regression-identification}
\end{equation}
By consistency, $M_i=M_i(D_i)$, and condition (\ref{ch5:eq:spillover-natural-treatment-mediator-independence}) in Assumption \ref{ch5:ass:spillover-natural-effects-independence}, we have
\begin{equation}
  \P(M_i=m_i \mid D_i=d_i, \bX, \bA) = \P(M_i(d_i)=m_i \mid \bX,\bA).
  \label{ch5:eq:proof-spillover-natural-own-mediator-distribution}
\end{equation}
On the other hand, by the law of total probability and consistency, $\bMmi=\bMmi(\bDmi)$, and noting that Assumption \ref{ch5:ass:error-term-independence} and condition (\ref{ch5:eq:spillover-natural-treatment-mediator-independence}) imply $\bMmi(\bdmi^*) \indep \bDmi \mid \bX,\bA$, we obtain
\begin{equation}
  \begin{aligned}
  & \P(\bMmi=\bmmi \mid T_i=t^*,\bX,\bA) \\
  = & \sum_{\bdmi^*:T_i(\bdmi^*)=t^*} \P(\bMmi(\bdmi^*)=\bmmi \mid \bX,\bA) \P(\bDmi=\bdmi^* \mid T_i=t^*,\bX,\bA).
  \end{aligned}
  \label{ch5:eq:proof-spillover-natural-peer-mediator-distribution}
\end{equation}
Substituting (\ref{ch5:eq:proof-spillover-natural-outcome-regression-identification}), (\ref{ch5:eq:proof-spillover-natural-own-mediator-distribution}), and (\ref{ch5:eq:proof-spillover-natural-peer-mediator-distribution}) into the definition of $\mu_i^{SN}(d_i,t,t^*)$ gives
\begin{align*}
  & \mu_i^{SN}(d_i,t,t^*) = \sum_{\bdmi:T_i(\bdmi)=t}\sum_{\bdmi^*:T_i(\bdmi^*)=t^*}\sum_{\bmmi\in\{0,1\}^{n-1}}\sum_{m_i\in\{0,1\}} \\
  & \E[Y_i(d_i,\bdmi,m_i,\bmmi) \mid \bX,\bA] \P(M_i(d_i)=m_i \mid \bX,\bA) \P(\bMmi(\bdmi^*)=\bmmi \mid \bX,\bA) \\
  & \P(\bDmi=\bdmi^* \mid T_i=t^*,\bX,\bA) \P(\bDmi=\bdmi \mid T_i=t,\bX,\bA).
\end{align*}
Assumption \ref{ch5:ass:error-term-independence} implies $M_i(d_i) \indep \bMmi(\bdmi^*) \mid \bX,\bA$, and hence
\begin{equation}
  \begin{aligned}
    & \P(M_i(d_i)=m_i \mid \bX,\bA) \P(\bMmi(\bdmi^*)=\bmmi \mid \bX,\bA) \\
    = & \P(M_i(d_i)=m_i, \bMmi(\bdmi^*)=\bmmi \mid \bX,\bA);
  \end{aligned}
  \label{ch5:eq:proof-spillover-natural-mediator-joint-factorization}
\end{equation}
By condition (\ref{ch5:eq:spillover-natural-cross-world-independence}) in Assumption \ref{ch5:ass:spillover-natural-effects-independence} and consistency,
\begin{equation}
\begin{aligned}
  & \E[Y_i(d_i,\bdmi,m_i,\bmmi) \mid \bX,\bA] \\
  = & \E[Y_i(d_i,\bdmi,M_i(d_i),\bMmi(\bdmi^*)) \mid M_i(d_i)=m_i, \bMmi(\bdmi^*)=\bmmi,\bX,\bA].
\end{aligned}
\label{ch5:eq:proof-spillover-natural-cross-world-outcome-expectation}
\end{equation}
Substituting (\ref{ch5:eq:proof-spillover-natural-mediator-joint-factorization}) and (\ref{ch5:eq:proof-spillover-natural-cross-world-outcome-expectation}) into the preceding expansion, the weighted sum over $(m_i,\bmmi)$ is exactly the expectation with respect to the joint conditional distribution of \\$(M_i(d_i),\bMmi(\bdmi^*))$. Therefore,
\begin{align*}
  \mu_i^{SN}(d_i,t,t^*) = & \sum_{\bdmi:T_i(\bdmi)=t}\sum_{\bdmi^*:T_i(\bdmi^*)=t^*}\E[Y_i(d_i,\bdmi,M_i(d_i),\bMmi(\bdmi^*)) \mid \bX,\bA] \\
  & \P(\bDmi=\bdmi^* \mid T_i=t^*,\bX,\bA) \P(\bDmi=\bdmi \mid T_i=t,\bX,\bA) \\
  = & Y_i^{SN}(d_i,t,t^*).
\end{align*}
This completes the proof.

\subsection{Proof of Theorem \ref{ch5:thm:spillover-interventional-effects-identification}}
By the law of iterated expectations,
\[
\begin{aligned}
& \E(Y_i \mid D_i=d_i,T_i=t,S_i=s,\bX,\bA) \\
= & \sum_{m_i\in\{0,1\}} \E(Y_i \mid D_i=d_i,T_i=t,M_i=m_i,S_i=s,\bX,\bA)\\
&\hspace{3cm} \P(M_i=m_i \mid D_i=d_i,T_i=t,S_i=s,\bX,\bA).
\end{aligned}
\]
Assumption \ref{ch5:ass:error-term-independence} implies $(D_i,M_i) \indep (\bDmi,\bMmi) \mid \bX,\bA$. The definitions of $T_i$ and $S_i$ then imply $M_i \indep (T_i,S_i) \mid D_i,\bX,\bA$. Combining these results with consistency and condition (\ref{ch5:eq:spillover-interventional-treatment-mediator-independence}) in Assumption \ref{ch5:ass:spillover-interventional-effects-independence}, we have
\[
\begin{aligned}
& \P(M_i=m_i \mid D_i=d_i,T_i=t,S_i=s,\bX,\bA) = \P(M_i=m_i \mid D_i=d_i,\bX,\bA) \\
& = \P(M_i(d_i)=m_i \mid \bX,\bA).
\end{aligned}
\]
On the other hand, the argument used in the proof of Theorem \ref{ch5:thm:own-cde-identification} gives
\[
\E(Y_i \mid D_i=d_i,T_i=t,M_i=m_i,S_i=s,\bX,\bA) = Y_i^C(d_i,t,m_i,s).
\]
Substituting the preceding results into the expression for $\E(Y_i \mid D_i=d_i,T_i=t,S_i=s,\bX,\bA)$, then averaging over $s$ with weights $\P(S_i=s \mid T_i=t^*,\bX,\bA)$, and using the definition of $H_i(t)$, which gives $\P(H_i(t^*)=s \mid \bX,\bA)=\P(S_i=s \mid T_i=t^*,\bX,\bA)$, yields
{\small
\[
\begin{aligned}
  \mu_i^{SI}(d_i,t,t^*) & = \sum_{s\in\mathcal{S}_i} \sum_{m_i\in\{0,1\}} Y_i^C(d_i,t,m_i,s) \P(M_i(d_i)=m_i \mid \bX,\bA)\P(H_i(t^*)=s \mid \bX,\bA) \\
  & = Y_i^{SI}(d_i,t,t^*).
\end{aligned}
\]
}%
This completes the proof.

\subsection{Conditional \texorpdfstring{$\psi$}{psi}-Dependence and Proof of Theorem \ref{ch5:thm:own-cde-asymptotic-normality}}
To establish asymptotic normality, we introduce the notion of $\psi$-dependence from \citet{kojevnikov2021limit}. For any $H,H'\subseteq\mathcal{N}_n$, define
\[
\ell_{\bA}(H,H') \coloneq \min\{\ell_{\bA}(i,j): i\in H,\, j\in H'\}.
\]
Let $\{Z_i\}_{i=1}^n \subset \mathbb{R}$ be a triangular array and write $\mathcal{Z}_H=(Z_i)_{i\in H}$. Let $\mathcal{L}_d$ denote the class of bounded real-valued Lipschitz functions on $\mathbb{R}^d$, and let $\mathrm{Lip}(f)$ denote the Lipschitz constant of $f\in\mathcal{L}_d$. Further define
\[
\mathcal{P}_n(h,h',r)=\{(H,H'): H, H' \subseteq \mathcal{N}_n, |H|=h, |H'|=h', \ell_{\bA}(H,H') \geqslant r\}.
\]
\begin{definition}[$\psi$-dependence]
\label{ch5:def:conditional-psi-dependence}
If there exists a constant $C\in(0,\infty)$ and, for each $n\in\mathbb{N}$, an $\mathcal{F}_n$-measurable sequence $\{\psi_n(r)\}_{r\geqslant 0}$ satisfying $\psi_n(0)=1$, such that for every $n,h,h'\in\mathbb{N}$, $r>0$, $f\in\mathcal{L}_h$, $f'\in\mathcal{L}_{h'}$, and $(H,H')\in\mathcal{P}_n(h,h',r)$,
\begin{equation}
|\Cov(f(\boldsymbol{Z}_H), f'(\boldsymbol{Z}_{H'}) \mid \mathcal{F}_n)| \leqslant C h h' (\|f\|_\infty+\mathrm{Lip}(f))(\|f'\|_\infty+\mathrm{Lip}(f'))\psi_n(r) \quad \text{a.s.}
\end{equation}
then the triangular array $\{Z_i\}_{i=1}^n$ is said to be \emph{conditionally $\psi$-dependent} given the sequence of $\sigma$-fields $\{\mathcal{F}_n\}_{n\in\mathbb{N}}$. The function $\psi_n(r)$ is called the \emph{dependence coefficient} of $\{Z_i\}_{i=1}^n$.
\end{definition}

\begin{lemma}
\label{ch5:lem:own-cde-psi-dependence}
Under Assumptions \ref{ch5:ass:consistency}, \ref{ch5:ass:error-term-independence}, \ref{ch5:ass:own-cde-local-exposure-mappings}, \ref{ch5:ass:own-cde-moments-overlap}(ii), and \ref{ch5:ass:approximate-neighborhood-interference}, for any fixed $(t,m,s)$, let $\mathcal{F}_n\coloneq\sigma(\bX,\bA)$. Then the triangular array $\{\varphi_i^C(t,m,s)\}_{i\in\mathcal{N}_n}$ is conditionally $\psi$-dependent given $\{\mathcal{F}_n\}_{n\in\mathbb{N}}$, with dependence coefficient $\psi_n(r)$ defined in (\ref{ch5:eq:own-cde-psi-dependence-coefficient}).
\end{lemma}
\begin{proof}
Fix $(t,m,s)$ and write $I_i(d)\coloneq\1\{D_i=d,\ T_i=t,\ M_i=m,\ S_i=s\}$. For any $r>0$, use the local model in Assumption \ref{ch5:ass:approximate-neighborhood-interference} to define the local approximation to the observed outcome
\[
Y_i^{(r)} \coloneq f_{n,r}(i,\bD_{\mathcal{N}(i,r)},\bM_{\mathcal{N}(i,r)},\bX_{\mathcal{N}(i,r)},\bA_{\mathcal{N}(i,r)},\be_{\mathcal{N}(i,r)}),
\]
Replace $Y_i$ in the definition of $\varphi_i^C(t,m,s)$ with $Y_i^{(r)}$ and denote the resulting localized score by $\varphi_i^{C,(r)}(t,m,s)$. Take arbitrary $h,h'\in\mathbb{N}$, $r>0$, $f\in\mathcal{L}_h$, $f'\in\mathcal{L}_{h'}$, and $(H,H')\in\mathcal{P}_n(h,h',r)$, and let
\begin{gather*}
  \xi \coloneq f\big((\varphi_i^C(t,m,s))_{i\in H}\big), \quad \xi^{(r)} \coloneq f\big((\varphi_i^{C,(r)}(t,m,s))_{i\in H}\big), \\
  \zeta \coloneq f'\big((\varphi_j^C(t,m,s))_{j\in H'}\big), \quad \zeta^{(r)} \coloneq f'\big((\varphi_j^{C,(r)}(t,m,s))_{j\in H'}\big).
\end{gather*}

First consider $r>4K_0$. Then $K_0<r/4$ and
\[
\big(\bigcup_{i\in H}\mathcal{N}(i,r/4)\big) \cap \big(\bigcup_{i\in H'}\mathcal{N}(i,r/4)\big) = \varnothing;
\]
Under Assumption \ref{ch5:ass:own-cde-local-exposure-mappings}, $I_i(d)$ is measurable with respect to $\sigma\bigl((\varepsilon_j,\nu_j,\omega_j):j\in\mathcal{N}(i,K_0)\bigr)\vee\mathcal{F}_n$, while $Y_i^{(r/4)}$ is measurable with respect to $\sigma\bigl((\varepsilon_j,\nu_j,\omega_j):j\in\mathcal{N}(i,r/4)\bigr)\vee\mathcal{F}_n$. Therefore, by Assumption \ref{ch5:ass:error-term-independence}, $\big(\varphi_i^{C,(r/4)}(t,m,s)\big)_{i\in H}$ and $\big(\varphi_j^{C,(r/4)}(t,m,s)\big)_{j\in H'}$ are conditionally independent given $\bX,\bA$, and hence
\begin{equation}
  \Cov\big(\xi^{(r/4)}, \zeta^{(r/4)} \mid \mathcal{F}_n \big) = 0.
  \label{ch5:eq:proof-truncated-score-conditional-covariance-zero}
\end{equation}

On the other hand, by definition,
\[
\varphi_i^C(t,m,s)-\varphi_i^{C,(r/4)}(t,m,s) = \left(\frac{I_i(1)}{\pi_i^C(1,t,m,s)} - \frac{I_i(0)}{\pi_i^C(0,t,m,s)}\right) (Y_i-Y_i^{(r/4)}),
\]
Combining the overlap condition in Assumption \ref{ch5:ass:own-cde-moments-overlap}(ii) with $\E[|Y_i-Y_i^{(r/4)}|\mid\mathcal{F}_n]\leqslant\gamma_n(r/4)$, which follows from Assumptions \ref{ch5:ass:consistency} and \ref{ch5:ass:approximate-neighborhood-interference} and iterated expectations, gives
\begin{equation}
  \E[|\varphi_i^C(t,m,s)-\varphi_i^{C,(r/4)}(t,m,s)| \mid \mathcal{F}_n] \leqslant \frac{2}{\underline{\pi}}\gamma_n(r/4).
  \label{ch5:eq:proof-truncated-cde-score-approximation-bound}
\end{equation}

Decompose the conditional covariance as
\[
\begin{aligned}
  & \Cov(\xi,\zeta \mid \mathcal{F}_n) \\
  & = \Cov(\xi - \xi^{(r/4)},\zeta \mid \mathcal{F}_n) + \Cov(\xi^{(r/4)}, \zeta - \zeta^{(r/4)} \mid \mathcal{F}_n) + \Cov(\xi^{(r/4)}, \zeta^{(r/4)} \mid \mathcal{F}_n),
\end{aligned}
\]
The last term is zero by (\ref{ch5:eq:proof-truncated-score-conditional-covariance-zero}). Applying $|\Cov(U,V\mid\mathcal{F}_n)|\leqslant 2\|V\|_\infty\E(|U|\mid\mathcal{F}_n)$ for integrable $U$ and bounded $V$, together with the Lipschitz properties of $f$ and $f'$ and (\ref{ch5:eq:proof-truncated-cde-score-approximation-bound}), to the first two terms yields
\begin{align*}
|\Cov(\xi, \zeta \mid \mathcal{F}_n)| & \leqslant \frac{4}{\underline{\pi}}\big(h \mathrm{Lip}(f)\|f'\|_\infty + h' \mathrm{Lip}(f')\|f\|_\infty\big)\gamma_n(r/4) \\
& \leqslant C hh'(\|f\|_\infty + \mathrm{Lip}(f))(\|f'\|_\infty+ \mathrm{Lip}(f'))\gamma_n(r/4),
\end{align*}
where $C>0$ is a constant; the second inequality uses $h,h'\geqslant 1$.

Now consider $r\leqslant 4K_0$. In this case $\psi_n(r)=1$, and boundedness of $f$ and $f'$ gives
\[
\begin{aligned}
& |\Cov(\xi, \zeta \mid \mathcal{F}_n)| \leqslant \E[|\xi \zeta| \mid \mathcal{F}_n] + |\E[\xi \mid \mathcal{F}_n] \E[\zeta \mid \mathcal{F}_n]| \leqslant 2\|f\|_\infty\|f'\|_\infty \\
& \leqslant 2 hh'(\|f\|_\infty+ \mathrm{Lip}(f))(\|f'\|_\infty+ \mathrm{Lip}(f'))\psi_n(r).
\end{aligned}
\]
Combining the two cases shows that the triangular array $\{\varphi_i^C(t,m,s)\}_{i\in\mathcal{N}_n}$ satisfies the conditional $\psi$-dependence condition in Definition \ref{ch5:def:conditional-psi-dependence}, with dependence coefficient $\psi_n(r)$. This proves the lemma.
\end{proof}

\textbf{Proof of Theorem \ref{ch5:thm:own-cde-asymptotic-normality}.}\quad 
Fix $(t,m,s)$ and write $I_i(d)\coloneq\1\{D_i=d,\ T_i=t,\ M_i=m,\ S_i=s\}$. We use the decomposition
\begin{align*}
& \sqrt{j_n} \big(\widehat{OCDE}(t,m,s) - OCDE(t,m,s)\big) \\
= & \frac{1}{\sqrt{j_n}} \sum_{i\in \mathcal{J}_n} \big[\big(\tau_i^C(1,t,m,s) - \tau_i^C(0,t,m,s)\big)-OCDE(t,m,s)\big] \\
& + R_{1n}(1) - R_{1n}(0) + R_{2n}(1) - R_{2n}(0),
\end{align*}
where
\begin{gather*}
  R_{1n}(d) \coloneq \frac{1}{\sqrt{j_n}} \sum_{i\in \mathcal{J}_n} I_i(d) \big(Y_i - \mu_i^C(d,t,m,s)\big) \left(\frac{1}{\widehat{\pi}_i^C(d,t,m,s)}-\frac{1}{\pi_i^C(d,t,m,s)}\right), \\
  R_{2n}(d) \coloneq \frac{1}{\sqrt{j_n}} \sum_{i\in \mathcal{J}_n} \big(\widehat{\mu}_i^C(d,t,m,s)-\mu_i^C(d,t,m,s)\big) \left(1 - \frac{I_i(d)}{\widehat{\pi}_i^C(d,t,m,s)}\right).
\end{gather*}

We first show that
\begin{align*}
  & \sigma_{n,C}^{-1}\frac{1}{\sqrt{j_n}} \sum_{i\in \mathcal{J}_n} \big[\big(\tau_i^C(1,t,m,s) - \tau_i^C(0,t,m,s)\big)-OCDE(t,m,s)\big] \\
  = & \sigma_{n,C}^{-1}\frac{1}{\sqrt{j_n}} \sum_{i\in \mathcal{J}_n} \varphi_i^C(t,m,s) \xrightarrow{d} \mathcal{N}(0,1).
\end{align*}
Define the centered variable $\overline{\varphi}_i^C(t,m,s)\coloneq\varphi_i^C(t,m,s)-\E\big(\varphi_i^C(t,m,s)\mid\bX,\bA\big)$. Since
\[
\E\big(Y_i - \mu_i^C(d,t,m,s) \mid D_i=d,T_i=t,M_i=m,S_i=s,\bX,\bA\big)=0
\]
we have $\E\big(\tau_i^C(d,t,m,s)\mid\bX,\bA\big)=\mu_i^C(d,t,m,s)$, and therefore
\[
\E \big(\varphi_i^C(t,m,s) \mid \bX,\bA \big) = \mu_i^C(1,t,m,s) - \mu_i^C(0,t,m,s) - OCDE(t,m,s).
\]
By Theorem \ref{ch5:thm:own-cde-identification}, $\displaystyle OCDE(t,m,s)=\frac{1}{j_n}\sum_{i\in\mathcal{J}_n}[\mu_i^C(1,t,m,s)-\mu_i^C(0,t,m,s)]$. Hence \\$\sum_{i\in\mathcal{J}_n}\E\big(\varphi_i^C(t,m,s)\mid\bX,\bA\big)=0$, so
\[
\frac{1}{\sqrt{j_n}} \sum_{i\in \mathcal{J}_n} \varphi_i^C(t,m,s) = \frac{1}{\sqrt{j_n}} \sum_{i\in \mathcal{J}_n} \overline{\varphi}_i^C(t,m,s).
\]

Define the triangular array on the full sample $\mathcal{N}_n$ by
\[
\widetilde{\varphi}_{i}^C(t,m,s) \coloneq \sqrt{\frac{n}{j_n}} \1\{i\in\mathcal{J}_n\} \overline{\varphi}_i^C(t,m,s), \quad i\in\mathcal{N}_n.
\]
We now verify the conditions required by the network central limit theorem of \citet{kojevnikov2021limit}. It follows directly from the definition of $\psi$-dependence that if $\{Z_i\}_{i\in\mathcal{N}_n}$ is conditionally $\psi$-dependent given $\{\mathcal{F}_n\}_{n\in\mathbb{N}}$, $\{a_i\}_{i\in\mathcal{N}_n}$ is $\mathcal{F}_n$-measurable, and $\{w_i\}_{i\in\mathcal{N}_n}$ is $\mathcal{F}_n$-measurable and uniformly bounded, then $\{w_i(Z_i-a_i)\}_{i\in\mathcal{N}_n}$ remains conditionally $\psi$-dependent, with its dependence coefficient differing by at most a multiplicative constant (which can be absorbed into $C$). Since $\E\big(\varphi_i^C(t,m,s)\mid\bX,\bA\big)$ is $\mathcal{F}_n$-measurable and, by Assumption \ref{ch5:ass:own-cde-moments-overlap}(ii), the weight $\sqrt{n/j_n}\1\{i\in\mathcal{J}_n\}\leqslant 1/\underline{\pi}$ is uniformly bounded, Lemma \ref{ch5:lem:own-cde-psi-dependence} implies that $\{\widetilde{\varphi}_i^C(t,m,s)\}_{i\in\mathcal{N}_n}$ is conditionally $\psi$-dependent given $\mathcal{F}_n$, with dependence coefficient still taken to be $\psi_n(r)$.

We next verify uniform boundedness of the moments. By Assumption \ref{ch5:ass:own-cde-moments-overlap}(i), there exist constants $M<\infty$ and $p>4$ such that $\E\big(|Y_i(\bd,\bm)|^p\mid\bX,\bA\big)\leqslant M$ almost surely for every $n$, $i$, and $(\bd,\bm)$. By Theorem \ref{ch5:thm:own-cde-identification}, $\mu_i^C(d,t,m,s)=Y_i^C(d,t,m,s)$, which by definition is a convex combination of terms of the form $\E\big(Y_i(d,\bdmi,m,\bmmi)\mid\bX,\bA\big)$. Jensen's inequality therefore yields
\begin{equation} 
  |\mu_i^C(d,t,m,s)|^p \leqslant M
  \label{ch5:eq:proof-own-cde-regression-moment-bound}
\end{equation}
almost surely. Moreover, because $\pi_i^C(d,t,m,s)$ is uniformly bounded away from zero by Assumption \ref{ch5:ass:own-cde-moments-overlap}(ii), there exists a constant $C>0$ such that
\[
\begin{aligned}
|\tau_i^C(d,t,m,s)| & \leqslant |\mu_i^C(d,t,m,s)| + \frac{I_i(d)}{\pi_i^C(d,t,m,s)} \big|Y_i-\mu_i^C(d,t,m,s)\big|\\
& \leqslant C \big(|Y_i| + |\mu_i^C(d,t,m,s)|\big).
\end{aligned}
\]
Taking conditional expectations and using Assumptions \ref{ch5:ass:consistency} and \ref{ch5:ass:own-cde-moments-overlap}(i), (\ref{ch5:eq:proof-own-cde-regression-moment-bound}), and $(a+b)^p\leqslant 2^{p-1}(a^p+b^p)$ gives
\begin{equation}
  \E\big(|\tau_i^C(d,t,m,s)|^p\mid \bX,\bA \big) < \infty
  \label{ch5:eq:proof-own-cde-aipw-score-moment-bound}
\end{equation}
almost surely. In addition, (\ref{ch5:eq:proof-own-cde-regression-moment-bound}) implies $|OCDE(t,m,s)|\leqslant 2M^{1/p}$. Combining this with $|\varphi_i^C(t,m,s)|\leqslant|\tau_i^C(1,t,m,s)|+|\tau_i^C(0,t,m,s)|+|OCDE(t,m,s)|$, (\ref{ch5:eq:proof-own-cde-aipw-score-moment-bound}), and $(a+b+c)^p\leqslant 3^{p-1}(a^p+b^p+c^p)$, and then taking conditional expectations, yields
\[
\sup_n\max_{i\in\mathcal{N}_n} \E\big(|\widetilde{\varphi}_{i}^C(t,m,s)|^p \mid \bX,\bA \big) < \infty
\]
almost surely.

We have established conditional $\psi$-dependence and uniform boundedness of the relevant moments. Together with Assumptions \ref{ch5:ass:own-cde-moments-overlap}(iii) and \ref{ch5:ass:own-cde-weak-dependence}, Theorem 3.2 of \citet{kojevnikov2021limit} therefore gives
\[
\begin{aligned}
  & \sigma_{n,C}^{-1}\frac{1}{\sqrt{j_n}} \sum_{i\in \mathcal{J}_n} \varphi_i^C(t,m,s) = \sigma_{n,C}^{-1} \frac{1}{\sqrt{n}} \sum_{i=1}^n \overline{\varphi}_{i}^C(t,m,s) \\
  & = \sigma_{n,C}^{-1} \frac{1}{\sqrt{n}} \sum_{i=1}^n \widetilde{\varphi}_{i}^C(t,m,s) \xrightarrow{d} \mathcal{N}(0,1).
\end{aligned}
\]

We next show that $R_{1n}(d)=o_p(1)$. By the definition of $\Psi_{\pi^C}$,
\[
R_{1n}(d) = \sqrt{j_n} \frac{1}{j_n} \sum_{i\in\mathcal{J}_n} \Psi_{\pi^C}\big(Z_i,\widehat{\pi}_i^C(d,t,m,s),d,t,m,s\big),
\]
Hence, by Assumption \ref{ch5:ass:own-cde-gnn-convergence-rates}(ii) and $j_n\leqslant n$, $R_{1n}(d)=\sqrt{j_n}o_p(n^{-1/2})=o_p(1)$.

We next show that $R_{2n}(d)=o_p(1)$. Using the identity
\[
1 - \frac{I_i(d)}{\widehat{\pi}_i^C(d,t,m,s)} = \left(1-\frac{I_i(d)}{\pi_i^C(d,t,m,s)}\right) + I_i(d) \frac{\widehat{\pi}_i^C(d,t,m,s)-\pi_i^C(d,t,m,s)}{\widehat{\pi}_i^C(d,t,m,s)\pi_i^C(d,t,m,s)},
\]
we decompose $R_{2n}(d)$ as $R_{2n}(d)=R_{21,n}(d)+R_{22,n}(d)$, where
{\small
\begin{gather*}
  R_{21,n}(d) \coloneq \frac{1}{\sqrt{j_n}} \sum_{i\in\mathcal{J}_n} \big(\widehat{\mu}_i^C(d,t,m,s)-\mu_i^C(d,t,m,s)\big) \left(1-\frac{I_i(d)}{\pi_i^C(d,t,m,s)}\right), \\
  R_{22,n}(d) \coloneq \frac{1}{\sqrt{j_n}} \sum_{i\in\mathcal{J}_n}\big(\widehat{\mu}_i^C(d,t,m,s)-\mu_i^C(d,t,m,s)\big) I_i(d) \frac{\widehat{\pi}_i^C(d,t,m,s)-\pi_i^C(d,t,m,s)}{\widehat{\pi}_i^C(d,t,m,s)\pi_i^C(d,t,m,s)}.
\end{gather*}
}%
For $R_{21,n}(d)$, the definition of $\Psi_{\mu^C}$ gives
\[
R_{21,n}(d) = \sqrt{j_n} \frac{1}{j_n} \sum_{i\in\mathcal{J}_n} \Psi_{\mu^C}\big(Z_i,\widehat{\mu}_i^C(d,t,m,s),d,t,m,s\big),
\]
Thus, by the same argument as for $R_{1n}(d)$ and Assumption \ref{ch5:ass:own-cde-gnn-convergence-rates}(ii), $R_{21,n}(d)=o_p(1)$. 

For $R_{22,n}(d)$, Assumption \ref{ch5:ass:own-cde-moments-overlap}(ii) implies that there exists a constant $C>0$ such that $[\widehat{\pi}_i^C(d,t,m,s)\pi_i^C(d,t,m,s)]^{-1}\leqslant C$ almost surely for every $i\in\mathcal{N}_n$. Therefore, by the Cauchy--Schwarz inequality and Assumption \ref{ch5:ass:own-cde-gnn-convergence-rates}(i),
\begin{align*}
& |R_{22,n}(d)| \leqslant C \sqrt{j_n} \bigg(\frac{1}{j_n} \sum_{i\in\mathcal{J}_n}
\big(\widehat{\mu}_i^C(d,t,m,s) - \mu_i^C(d,t,m,s)\big)^2\bigg)^{1/2} \\
& \times \bigg(\frac{1}{j_n}\sum_{i\in\mathcal{J}_n}\big(\widehat{\pi}_i^C(d,t,m,s)-\pi_i^C(d,t,m,s)\big)^2\bigg)^{1/2} = C \sqrt{j_n} o_p(n^{-1/2}) = o_p(1).
\end{align*}
By Assumption \ref{ch5:ass:own-cde-moments-overlap}(iii), $\sigma_{n,C}^{-1}=O_p(1)$, so $\sigma_{n,C}^{-1}R_{1n}(d)=o_p(1)$ and $\sigma_{n,C}^{-1}R_{2n}(d)=o_p(1)$. Combining the preceding results and applying Slutsky's theorem yields
\[
\sigma_{n,C}^{-1} \sqrt{j_n} \big(\widehat{OCDE}(t,m,s) - OCDE(t,m,s)\big) \xrightarrow{d} \mathcal{N}(0,1).
\]
This completes the proof.

\subsection{Proof of Theorem \ref{ch5:thm:own-cde-variance-consistency}}
Fix $(t,m,s)$ and let $I_i(d)\coloneq\1\{D_i=d,\ T_i=t,\ M_i=m,\ S_i=s\}$ for $d\in\{0,1\}$. Define the oracle and feasible residual scores
\begin{gather*}
r_i(d) \coloneq \tau_i^C(d,t,m,s)-\mu_i^C(d,t,m,s) = \frac{I_i(d)}{\pi_i^C(d,t,m,s)}\bigl(Y_i-\mu_i^C(d,t,m,s)\bigr),\\
\widehat{r}_i(d) \coloneq \widehat{\tau}_i^C(d,t,m,s)-\widehat{\mu}_i^C(d,t,m,s) = \frac{I_i(d)}{\widehat{\pi}_i^C(d,t,m,s)}\bigl(Y_i-\widehat{\mu}_i^C(d,t,m,s)\bigr),\\
r_i\coloneq r_i(1)-r_i(0),\quad \widehat{r}_i \coloneq \widehat{r}_i(1) - \widehat{r}_i(0),
\end{gather*}
Then $\widetilde{\tau}=\bigl(\widehat r_i\bigr)_{i\in\mathcal{J}_n}$. For any $q\in\{U,PD\}$, define the corresponding oracle HAC quadratic form
\[
\widetilde \sigma_q^2 \coloneq \frac{1}{j_n} r'K_qr, \quad r \coloneq (r_i)_{i\in \mathcal{J}_n}.
\]

\textbf{Step 1: Re-express the target variance.} By the definition of the centered score,
\begin{equation}
  \varphi_i^C(t,m,s) = r_i+\mu_i^C(1,t,m,s)-\mu_i^C(0,t,m,s)-OCDE(t,m,s).
  \label{ch5:eq:proof-own-cde-centered-score-decomposition}
\end{equation}
By Theorem \ref{ch5:thm:own-cde-identification}, $\sum_{i\in\mathcal{J}_n}(\mu_i^C(1,t,m,s)-\mu_i^C(0,t,m,s)-OCDE(t,m,s))=0$. Hence \\$\sum_{i\in\mathcal{J}_n}\varphi_i^C(t,m,s)=\sum_{i\in\mathcal{J}_n}r_i$, and therefore
\[
\sigma_{n,C}^2 = \Var \bigg(\frac{1}{\sqrt{j_n}}\sum_{i\in \mathcal{J}_n}\varphi_i^C(t,m,s) \mid \bX,\bA\bigg) = \Var\bigg(\frac{1}{\sqrt{j_n}}\sum_{i\in \mathcal{J}_n}r_i \mid \bX,\bA\bigg).
\]

\textbf{Step 2: Establish equivalence of the oracle and feasible HAC estimators.} For any $q\in\{U,PD\}$, symmetry of $K_q$ gives $\displaystyle \widehat{\sigma}_q^2-\widetilde\sigma_q^2=\frac{1}{j_n}(\widehat r-r)'K_q(\widehat r+r)$. By the Cauchy--Schwarz inequality,
\[
|\widehat{\sigma}_q^2-\widetilde{\sigma}_q^2| \leqslant \bigg(\frac{1}{j_n}\sum_{i\in \mathcal{J}_n}(\widehat{r}_i-r_i)^2\bigg)^{1/2} \bigg(\frac{1}{j_n}\sum_{i\in \mathcal{J}_n}\bigg(\sum_{j\in \mathcal{J}_n}|K_{q,ij}| |\widehat{r}_j+r_j|\bigg)^2\bigg)^{1/2}.
\]
By the definition of the kernel matrices, for the positive-definite kernel $K_{PD}$, $K_{PD,ij}=0$ whenever $\ell_{\bA}(i,j)>b_n$, whereas for $\ell_{\bA}(i,j)\leqslant b_n$, the Cauchy--Schwarz inequality gives $0\leqslant K_{PD,ij}\leqslant 1$. Thus, uniformly over $q\in\{U,PD\}$, $|K_{q,ij}|\leqslant\1\{\ell_{\bA}(i,j)\leqslant b_n\}$, so $\sum_{j\in\mathcal{J}_n}|K_{q,ij}|\leqslant n(i,b_n)$. Therefore,
\begin{equation}
|\widehat{\sigma}_q^2- \widetilde{\sigma}_q^2| \leqslant \bigg(\frac{1}{j_n}\sum_{i\in \mathcal{J}_n}(\widehat{r}_i-r_i)^2\bigg)^{1/2}\bigg(\frac{1}{j_n}\sum_{i\in \mathcal{J}_n}n(i,b_n)^2\bigg)^{1/2} \max_{j\in \mathcal{J}_n}|\widehat{r}_j+r_j|.
\label{ch5:eq:proof-hac-feasible-oracle-difference-bound}
\end{equation}

We now bound the three factors on the right-hand side of (\ref{ch5:eq:proof-hac-feasible-oracle-difference-bound}). By Assumption \ref{ch5:ass:own-cde-hac-consistency}(i) and the overlap condition in Assumption \ref{ch5:ass:own-cde-moments-overlap}(ii), there exists a constant $C>0$ such that $|\widehat r_i(d)|\leqslant C$ and $|r_i(d)|\leqslant C$ almost surely for all $i\in\mathcal{J}_n$ and $d\in\{0,1\}$. Hence
\begin{equation}
\max_{j\in \mathcal{J}_n}|\widehat{r}_j+r_j| = O_p(1).
\label{ch5:eq:proof-hac-residual-max-bound}
\end{equation}
By Assumption \ref{ch5:ass:own-cde-hac-consistency}(iii) and $j_n/n\geqslant\underline\pi>0$,
\begin{equation}
\bigg(\frac{1}{j_n}\sum_{i\in \mathcal{J}_n}n(i,b_n)^2\bigg)^{1/2} \leqslant \bigg(\frac{n}{j_n}\frac{1}{n}\sum_{i=1}^n n(i,b_n)^2\bigg)^{1/2} = O_p(n^{1/4}).
\label{ch5:eq:proof-hac-neighborhood-growth-bound}
\end{equation}
Moreover, for any $d\in\{0,1\}$,
\begin{align*}
\widehat{r}_i(d)-r_i(d) & = I_i(d)\bigl(Y_i-\widehat{\mu}_i^C(d,t,m,s)\bigr)\frac{\pi_i^C(d,t,m,s)-\widehat{\pi}_i^C(d,t,m,s)}{\widehat{\pi}_i^C(d,t,m,s)\pi_i^C(d,t,m,s)} \\
& \quad - \frac{I_i(d)}{\pi_i^C(d,t,m,s)} \bigl(\widehat{\mu}_i^C(d,t,m,s)-\mu_i^C(d,t,m,s)\bigr),
\end{align*}
Therefore, by $(a-b)^2\leq 2a^2+2b^2$, Assumption \ref{ch5:ass:own-cde-hac-consistency}(i), and Assumption \ref{ch5:ass:own-cde-moments-overlap}(ii), there exists a constant $C>0$ such that
{\small
\[
\big(\widehat{r}_i(d)-r_i(d)\big)^2 \leqslant C\Big[\big(\widehat{\pi}_i^C(d,t,m,s)-\pi_i^C(d,t,m,s)\big)^2 + \big(\widehat{\mu}_i^C(d,t,m,s)-\mu_i^C(d,t,m,s)\big)^2\Big],
\]
}%
Applying $(a-b)^2\leq 2a^2+2b^2$ once more and Assumption \ref{ch5:ass:own-cde-gnn-convergence-rates}(i) yields
\begin{equation}
\bigg(\frac{1}{j_n}\sum_{i\in \mathcal{J}_n}(\widehat{r}_i-r_i)^2
\bigg)^{1/2} = o_p(n^{-1/4}).
\label{ch5:eq:proof-hac-residual-mse-rate}
\end{equation}
Substituting (\ref{ch5:eq:proof-hac-residual-max-bound}), (\ref{ch5:eq:proof-hac-neighborhood-growth-bound}), and (\ref{ch5:eq:proof-hac-residual-mse-rate}) into (\ref{ch5:eq:proof-hac-feasible-oracle-difference-bound}) gives
\begin{equation}
\widehat{\sigma}_q^2-\widetilde{\sigma}_q^2=o_p(1), \quad q\in\{U,PD\}.
\label{ch5:eq:proof-hac-feasible-oracle-equivalence}
\end{equation}

\textbf{Step 3: Establish consistency of the uniform-kernel oracle HAC estimator.} Define the triangular array
\[
u_i \coloneq \sqrt{\frac{n}{j_n}} \1\{i\in \mathcal{J}_n\}r_i, \quad i \in \mathcal{N}_n,
\]
Then
\[
\widetilde{\sigma}_U^2 = \frac{1}{j_n}\sum_{i\in \mathcal{J}_n}\sum_{j\in \mathcal{J}_n}r_ir_j \1\{\ell_{\bA}(i,j)\leqslant b_n\} = \frac{1}{n}\sum_{i=1}^n\sum_{j=1}^n u_i u_j \1\{\ell_{\bA}(i,j)\leqslant b_n\}.
\]
We verify the conditions of Proposition 4.1 in \citet{kojevnikov2021limit}. First, by the definition of $r_i(d)$, $\E[r_i(d)\mid\bX,\bA]=0$, so $\E[u_i\mid\bX,\bA]=0$. Second, by (\ref{ch5:eq:proof-own-cde-centered-score-decomposition}), $r_i$ is obtained from $\varphi_i^C(t,m,s)$ by subtracting an $\mathcal{F}_n=\sigma(\bX,\bA)$-measurable quantity, and $u_i$ is then obtained by multiplying by the bounded $\mathcal{F}_n$-measurable weight $\sqrt{n/j_n}\1\{i\in\mathcal{J}_n\}$. The conditional-$\psi$-dependence argument in the proof of Theorem \ref{ch5:thm:own-cde-asymptotic-normality} therefore implies that $\{u_i\}_{i\in\mathcal{N}_n}$ remains conditionally $\psi$-dependent given $\mathcal{F}_n$, with dependence coefficient $\psi_n(r)$. Hence Assumption 2.1(i) of \citet{kojevnikov2021limit} holds, while its Assumption 2.1(ii) is exactly Assumption \ref{ch5:ass:own-cde-weak-dependence}(i) here. Third, Assumption \ref{ch5:ass:own-cde-hac-consistency}(i) together with Assumption \ref{ch5:ass:own-cde-moments-overlap}(ii) implies that $u_i$ is uniformly bounded. Thus, for some $p>4$, $\sup_n\max_{i\in\mathcal{N}_n}\E(|u_i|^p\mid\bX,\bA)<\infty$ almost surely, so its Assumption 4.1(i) holds; its Assumptions 4.1(ii) and (iii) are precisely the first two conditions in Assumption \ref{ch5:ass:own-cde-hac-consistency}(ii). Finally, note that
\[
\sigma_{n,C}^{2} = \Var\bigg(\frac{1}{\sqrt{j_n}}\sum_{i\in \mathcal{J}_n}r_i \mid \bX,\bA \bigg) = \Var\bigg(\frac{1}{\sqrt{n}}\sum_{i=1}^{n}u_i \mid \bX,\bA\bigg),
\]
Proposition 4.1 of \citet{kojevnikov2021limit} therefore yields
\begin{equation}
\widetilde{\sigma}_U^2 = \sigma_{n,C}^2 + o_p(1).
\label{ch5:eq:proof-uniform-hac-oracle-consistency}
\end{equation}

\textbf{Step 4: Establish consistency of the positive-definite-kernel oracle HAC estimator.}
Using the triangular array $\{u_i\}_{i\in\mathcal{N}_n}$ from Step 3, we have
\[
\widetilde{\sigma}_{PD}^{2}
=
\frac{1}{j_n}\sum_{i\in \mathcal{J}_n}\sum_{j\in \mathcal{J}_n}r_ir_jK_{PD,ij}
=
\frac{1}{n}\sum_{i=1}^{n}\sum_{j=1}^{n}u_i u_j K_{PD,ij},
\]
Because $\E(u_i\mid\bX,\bA)=0$, $\displaystyle \sigma_{n,C}^{2}=\frac{1}{n}\sum_{i=1}^{n}\sum_{j=1}^{n}\E(u_i u_j\mid\bX,\bA)$. Decompose
\begin{equation}
\widetilde{\sigma}_{PD}^{2}-\sigma_{n,C}^{2} = \underbrace{\widetilde{\sigma}_{PD}^{2}-\E(\widetilde{\sigma}_{PD}^{2}\mid \bX,\bA)}_{\eqcolon A_{n,1}} + \underbrace{\E(\widetilde{\sigma}_{PD}^{2}\mid \bX,\bA)-\sigma_{n,C}^{2}}_{\eqcolon A_{n,2}}.
\label{ch5:eq:proof-positive-definite-hac-error-decomposition}
\end{equation}
The conditional $\psi$-dependence established in Step 3 (together with Assumption \ref{ch5:ass:own-cde-weak-dependence}(i), which verifies Assumption 2.1 of \citet{kojevnikov2021limit}) and finiteness of the conditional moments ensure that all conditions required for its Corollary A.2 are satisfied.

For $A_{n,1}$, by definition and using $0\leqslant K_{PD,ij}\leqslant 1$ and $K_{PD,ij}=0$ whenever $\ell_{\bA}(i,j)>b_n$,
\[
\begin{aligned}
  & \Var(A_{n,1} \mid \bX,\bA) \\
  & \leqslant \frac{1}{n^2} \sum_{i,j,k,l=1}^{n} \1\{\ell_{\bA}(i,j)\leqslant b_n\}\1\{\ell_{\bA}(k,l)\leqslant b_n\} |\Cov(u_i u_j,u_k u_l \mid \bX,\bA)|.
\end{aligned}
\]
Grouping by the set distance between $\{i,j\}$ and $\{k,l\}$, define
\[
H_n(r,b_n) \coloneq \# \big\{(i,j,k,l) \in \mathcal{N}_n^4: j\in \mathcal{N}(i,b_n),l \in \mathcal{N}(k,b_n),\ell_{\bA}(\{i,j\},\{k,l\})=r \big\}.
\]
By the preceding corollary, there exists a constant $C<\infty$ such that \\$|\Cov(u_i u_j,u_k u_l\mid\bX,\bA)|\leqslant C\psi_n\big(\ell_{\bA}(\{i,j\},\{k,l\})\big)^{1-4/p}$ almost surely. Hence
{\small
\[
  \Var(A_{n,1} \mid \bX,\bA) \leqslant \frac{C}{n^2}\sum_{r=0}^{\infty}H_n(r,b_n)\psi_n(r)^{1-4/p} \leqslant \frac{4C}{n}\sum_{r=0}^{\infty}c_n(r,b_n,2)\psi_n(r)^{1-4/p} = o(1),
\]
}%
where the second inequality follows from equation (A.13) of \citet{kojevnikov2021limit}, $H_n(r,b_n)\leqslant 4n c_n(r,b_n,2)$, and the final equality follows from Assumption \ref{ch5:ass:own-cde-hac-consistency}(ii). Thus $A_{n,1}=o_p(1)$.

For $A_{n,2}$, by definition $\displaystyle A_{n,2}=\frac{1}{n}\sum_{i=1}^{n}\sum_{j=1}^{n}(K_{PD,ij}-1)\E(u_i u_j\mid\bX,\bA)$. The same corollary implies that there exists a constant $C<\infty$ such that $|\E(u_i u_j\mid\bX,\bA)|=|\Cov(u_i,u_j\mid\bX,\bA)|\leqslant C\psi_n(\ell_{\bA}(i,j))^{1-2/p}$ almost surely. Therefore,
\[
  |A_{n,2}| \leqslant C \sum_{r=0}^{n}\psi_n(r)^{1-2/p} \frac{1}{n}\sum_{i=1}^{n}\sum_{j\in \mathcal{N}^B(i,r)}|K_{PD,ij}-1| = o(1),
\]
where the final equality follows from the third condition in Assumption \ref{ch5:ass:own-cde-hac-consistency}(ii). Combining this with (\ref{ch5:eq:proof-positive-definite-hac-error-decomposition}) gives
\begin{equation}
  \widetilde{\sigma}_{PD}^{2}-\sigma_{n,C}^{2} = o_p(1).
  \label{ch5:eq:proof-positive-definite-hac-oracle-consistency}
\end{equation}

\textbf{Step 5: Combine the preceding results.}
By (\ref{ch5:eq:proof-hac-feasible-oracle-equivalence}), (\ref{ch5:eq:proof-uniform-hac-oracle-consistency}), and (\ref{ch5:eq:proof-positive-definite-hac-oracle-consistency}), for every $q\in\{U,PD\}$, $\widehat{\sigma}_q^2=\sigma_{n,C}^2+o_p(1)$. Therefore,
\[
|\widehat{\sigma}_{C}^{2}-\sigma_{n,C}^{2}| = \big| \max\{\widehat{\sigma}_{U}^{2},\widehat{\sigma}_{PD}^{2}\} - \sigma_{n,C}^{2} \big| \leqslant \max_{q\in\{U,PD\}} \big|\widehat{\sigma}_{q}^{2}-\sigma_{n,C}^{2}\big| = o_p(1),
\]
Thus $\widehat{\sigma}_{C}^{2}=\sigma_{n,C}^{2}+o_p(1)$. This completes the proof.

\subsection{Proof of Theorem \ref{ch5:thm:own-nde-asymptotic-normality}}
We first state a lemma needed for the proof.
\begin{lemma}[$\psi$-dependence]
  \label{ch5:lem:own-nde-psi-dependence}
  Under Assumptions \ref{ch5:ass:consistency}, \ref{ch5:ass:error-term-independence}, \ref{ch5:ass:own-natural-local-exposure-mapping}, \ref{ch5:ass:own-natural-effects-moments-overlap}(ii), and \ref{ch5:ass:approximate-neighborhood-interference}, for any fixed $t$, let $\mathcal{F}_n\coloneq\sigma(\bX,\bA)$. Then the triangular array $\{\varphi_i^N(t)\}_{i\in\mathcal{N}_n}$ is conditionally $\psi$-dependent given $\{\mathcal{F}_n\}_{n\in\mathbb{N}}$, with dependence coefficient $\psi_n(r)$ defined in (\ref{ch5:eq:own-natural-effects-psi-dependence-coefficient}).
\end{lemma}
\begin{proof}
  The proof is identical to that of Lemma \ref{ch5:lem:own-cde-psi-dependence} and is therefore omitted.
\end{proof}

\textbf{Proof of Theorem \ref{ch5:thm:own-nde-asymptotic-normality}.}\quad 
Fix $t$. We use the decomposition
{\small
\begin{align*}
  & \sqrt{j_n}\big(\widehat{ONDE}(t) - ONDE(t)\big) \\
  = & \frac{1}{\sqrt{j_n}}\sum_{i\in\mathcal{J}_n}\big[\big(\tau_i^N(1,0,t) - \tau_i^N(0,0,t)\big) - ONDE(t)\big] + \sum_{\ell=1}^6 R_{\ell n}(1,0,t) - \sum_{\ell=1}^6 R_{\ell n}(0,0,t),
\end{align*}
}%
where
{\small
\begin{gather*}
R_{1n}(d,d^*,t) \coloneq \frac{1}{\sqrt{j_n}}\sum_{i\in \mathcal{J}_n} \Psi_{\mu^N}(Z_i,\widehat{\mu}_i^N(d,d^*,t),d,d^*,t),\\
R_{2n}(d,d^*,t) \coloneq \frac{1}{\sqrt{j_n}}\sum_{i\in \mathcal{J}_n} \Psi_{\eta^N}(Z_i,\widehat{\eta}_i^N(d,t,M_i),d,d^*,t),\\
R_{3n}(d,d^*,t) \coloneq \frac{1}{\sqrt{j_n}}\sum_{i\in \mathcal{J}_n} \Psi_{q^N}(Z_i,\widehat q_i^N(M_i\mid d^*),\widehat q_i^N(M_i\mid d),d,d^*,t),\\
R_{4n}(d,d^*,t) \coloneq \frac{1}{\sqrt{j_n}}\sum_{i\in \mathcal{J}_n} \Psi_{e^N}(Z_i,\widehat{e}_i^N(d,t),d,d^*,t),\\
R_{5n}(d,d^*,t) \coloneq \frac{1}{\sqrt{j_n}}\sum_{i\in \mathcal{J}_n} \Psi_{p^N}(Z_i,\widehat{p}_i^N(d^*),d,d^*,t),\\
R_{6n}(d,d^*,t) \coloneq \frac{1}{\sqrt{j_n}} \sum_{i\in \mathcal{J}_n}\bigg\{\1\{D_i=d^*\}\big(\widehat{\eta}_i^N(d,t,M_i)-\eta_i^N(d,t,M_i)\big) \left(\frac{1}{\widehat{p}_i^N(d^*)}-\frac{1}{p_i^N(d^*)}\right)\\
- \1\{D_i=d^*\}\big(\widehat{\mu}_i^N(d,d^*,t)-\mu_i^N(d,d^*,t)\big)\left(\frac{1}{\widehat{p}_i^N(d^*)}-\frac{1}{p_i^N(d^*)}\right)\\
+ \1\{D_i=d,T_i=t\}\bigl(Y_i-\eta_i^N(d,t,M_i)\bigr)
\left(\frac{\widehat q_i^N(M_i\mid d^*)}{\widehat q_i^N(M_i\mid d)}-\frac{q_i^N(M_i\mid d^*)}{q_i^N(M_i\mid d)}\right) \\
\hspace{3cm}\left(\frac{1}{\widehat{e}_i^N(d,t)}-\frac{1}{e_i^N(d,t)}\right)\\
- \1\{D_i=d,T_i=t\}\,\frac{q_i^N(M_i\mid d^*)}{q_i^N(M_i\mid d)} \big(\widehat{\eta}_i^N(d,t,M_i)-\eta_i^N(d,t,M_i)\big) \left(\frac{1}{\widehat{e}_i^N(d,t)}-\frac{1}{e_i^N(d,t)}\right)\\
- \frac{\1\{D_i=d,T_i=t\}}{e_i^N(d,t)} \left(\frac{\widehat q_i^N(M_i\mid d^*)}{\widehat q_i^N(M_i\mid d)}-\frac{q_i^N(M_i\mid d^*)}{q_i^N(M_i\mid d)}\right) \big(\widehat{\eta}_i^N(d,t,M_i)-\eta_i^N(d,t,M_i)\big)\\
- \1\{D_i=d,T_i=t\} \left(\frac{\widehat q_i^N(M_i\mid d^*)}{\widehat q_i^N(M_i\mid d)}-\frac{q_i^N(M_i\mid d^*)}{q_i^N(M_i\mid d)}\right) \big(\widehat{\eta}_i^N(d,t,M_i)-\eta_i^N(d,t,M_i)\big) \\
\hspace{3cm}\left(\frac{1}{\widehat{e}_i^N(d,t)}-\frac{1}{e_i^N(d,t)}\right)\bigg\}.
\end{gather*}
}%

We first show that
\begin{align*}
  & \sigma_{n,N}^{-1}\frac{1}{\sqrt{j_n}} \sum_{i\in \mathcal{J}_n} \big[\big(\tau_i^N(1,0,t) - \tau_i^N(0,0,t)\big)-ONDE(t)\big] \\
  = & \sigma_{n,N}^{-1}\frac{1}{\sqrt{j_n}} \sum_{i\in \mathcal{J}_n} \varphi_i^N(t) \xrightarrow{d} \mathcal{N}(0,1).
\end{align*}
Define the centered variable $\overline{\varphi}_i^N(t)\coloneq\varphi_i^N(t)-\E(\varphi_i^N(t)\mid\bX,\bA)$. By the definition of $\mu_i^N(d,d^*,t)$ and iterated conditional expectations,
\begin{gather*}
 \E\bigg[\frac{\1\{D_i=d^*\}}{p_i^N(d^*)}\big[\eta_i^N(d,t,M_i) - \mu_i^N(d,d^*,t)\big] \mid \bX,\bA \bigg] = 0, \\
\E\bigg[\frac{\1\{D_i=d,T_i=t\}}{e_i^N(d,t)}\frac{q_i^N(M_i \mid d^*)}{q_i^N(M_i \mid d)}\big[Y_i - \eta_i^N(d,t,M_i)\big] \mid \bX,\bA \bigg] = 0,
\end{gather*}
Hence $\E[\tau_i^N(d,d^*,t)\mid\bX,\bA]=\mu_i^N(d,d^*,t)$, so
\[
  \E(\varphi_i^N(t) \mid \bX,\bA) = \mu_i^N(1,0,t) - \mu_i^N(0,0,t) - ONDE(t).
\]
By Theorem \ref{ch5:thm:own-natural-effects-identification}, $\displaystyle ONDE(t)=\frac{1}{j_n}\sum_{i\in\mathcal{J}_n}[\mu_i^N(1,0,t)-\mu_i^N(0,0,t)]$. Therefore $\sum_{i\in\mathcal{J}_n}\E(\varphi_i^N(t)\mid\bX,\bA)=0$, and hence
\[
\frac{1}{\sqrt{j_n}} \sum_{i\in\mathcal{J}_n}\varphi_i^N(t) = \frac{1}{\sqrt{j_n}} \sum_{i\in\mathcal{J}_n}\overline{\varphi}_i^N(t).
\]

Define the triangular array
\[
\widetilde{\varphi}_i^N(t) \coloneq \sqrt{\frac{n}{j_n}} \1\{i\in\mathcal{J}_n\} \overline{\varphi}_i^N(t), \quad i\in\mathcal{N}_n.
\]
By Lemma \ref{ch5:lem:own-nde-psi-dependence} and the same argument used in the proof of Theorem \ref{ch5:thm:own-cde-asymptotic-normality}, $\{\widetilde{\varphi}_i^N(t)\}_{i\in\mathcal{N}_n}$ is conditionally $\psi$-dependent given $\mathcal{F}_n$, with dependence coefficient $\psi_n(r)$ as defined in (\ref{ch5:eq:own-natural-effects-psi-dependence-coefficient}).

Assumption \ref{ch5:ass:own-natural-effects-moments-overlap}(i), Theorem \ref{ch5:thm:own-natural-effects-identification}, and Jensen's inequality imply $|\eta_i^N(d,t,m)|^p\leqslant M$. By the definition of $\mu_i^N(d,d^*,t)$ and Jensen's inequality,
\begin{equation}
  |\mu_i^N(d,d^*,t)|^p \leqslant M
  \label{ch5:eq:proof-own-natural-regression-moment-bound}
\end{equation}
almost surely. By Assumption \ref{ch5:ass:own-natural-effects-moments-overlap}(ii), there exists a constant $C>0$ such that
\begin{align*}
|\tau_i^N(d,d^*,t)| & \leqslant |\mu_i^N(d,d^*,t)| + \frac{\1\{D_i=d^*\}}{p_i^N(d^*)}|\eta_i^N(d,t,M_i) - \mu_i^N(d,d^*,t)|\\
& \quad + \frac{\1\{D_i=d,T_i=t\}}{e_i^N(d,t)} \frac{q_i^N(M_i \mid d^*)}{q_i^N(M_i \mid d)}|Y_i-\eta_i^N(d,t,M_i)| \\
& \leqslant C \big(|Y_i|+|\eta_i^N(d,t,M_i)|+|\mu_i^N(d,d^*,t)|\big),
\end{align*}
Therefore, by Assumptions \ref{ch5:ass:consistency} and \ref{ch5:ass:own-natural-effects-moments-overlap}(i) and the inequality $(a+b+c)^p\leqslant 3^{p-1}(a^p+b^p+c^p)$, $\E[|\tau_i^N(d,d^*,t)|^p\mid\bX,\bA]<\infty$ almost surely. In addition, the definition of $ONDE(t)$ and (\ref{ch5:eq:proof-own-natural-regression-moment-bound}) imply $|ONDE(t)|\leqslant 2M^{1/p}$. The same argument as in the proof of Theorem \ref{ch5:thm:own-cde-asymptotic-normality} therefore gives \\$\sup_n\max_{i\in\mathcal{N}_n}\E[|\widetilde{\varphi}_i^N(t)|^p\mid\bX,\bA]<\infty$ almost surely. Combining this with Assumptions \ref{ch5:ass:own-natural-effects-moments-overlap}(iii) and \ref{ch5:ass:own-natural-effects-weak-dependence}, Theorem 3.2 of \citet{kojevnikov2021limit} yields
\begin{equation}
  \sigma_{n,N}^{-1}\frac{1}{\sqrt{j_n}} \sum_{i\in \mathcal{J}_n} \varphi_i^N(t) = \sigma_{n,N}^{-1} \frac{1}{\sqrt{n}} \sum_{i=1}^n \overline{\varphi}_{i}^N(t) = \sigma_{n,N}^{-1} \frac{1}{\sqrt{n}} \sum_{i=1}^n \widetilde{\varphi}_{i}^N(t) \xrightarrow{d} \mathcal{N}(0,1).
  \label{ch5:eq:proof-own-nde-oracle-clt}
\end{equation}

By Assumption \ref{ch5:ass:own-natural-effects-gnn-convergence-rates}(ii) and $j_n\leqslant n$, for $\ell=1,\ldots,5$,
\[
R_{\ell n}(d,d^*,t) = \sqrt{j_n}o_p(n^{-1/2}) = o_p(1).
\]

It remains to show that $R_{6n}(d,d^*,t)=o_p(1)$. By Assumptions \ref{ch5:ass:consistency} and \ref{ch5:ass:own-natural-effects-moments-overlap}(i), $|Y_i|\leqslant M$, $|\eta_i^N(d,t,m)|\leqslant M$, $|\mu_i^N(d,d^*,t)|\leqslant M$, and $|Y_i-\eta_i^N(d,t,m)|\leqslant 2M$. Assumption \ref{ch5:ass:own-natural-effects-moments-overlap}(ii) further implies that there exists a constant $C<\infty$ such that
\begin{gather}
  \left|\frac{1}{\widehat{p}_i^N(d^*)} - \frac{1}{p_i^N(d^*)}\right| \leqslant C |\widehat{p}_i^N(d^*)-p_i^N(d^*)|, \label{ch5:eq:proof-own-nde-treatment-propensity-inverse-bound}\\
  \left|\frac{1}{\widehat{e}_i^N(d,t)} - \frac{1}{e_i^N(d,t)} \right| \leqslant C|\widehat{e}_i^N(d,t) - e_i^N(d,t)|, \label{ch5:eq:proof-own-nde-exposure-propensity-inverse-bound}\\
  \begin{aligned}
  & \left|\frac{\widehat{q}_i^N(M_i \mid d^*)}{\widehat{q}_i^N(M_i \mid d)} - \frac{q_i^N(M_i \mid d^*)}{q_i^N(M_i \mid d)}\right|\\
  & \leqslant C\big(\big|\widehat{q}_i^N(M_i \mid d^*) - q_i^N(M_i\mid d^*)| + |\widehat{q}_i^N(M_i \mid d) - q_i^N(M_i \mid d)\big|\big). 
  \end{aligned}
  \label{ch5:eq:proof-own-nde-mediator-density-ratio-bound}
\end{gather}

Since $R_{6n}$ is the sum of six summations, it suffices to show that each is $o_p(1)$. For the first summation, by (\ref{ch5:eq:proof-own-nde-treatment-propensity-inverse-bound}), the Cauchy--Schwarz inequality, and Assumption \ref{ch5:ass:own-natural-effects-gnn-convergence-rates}(i),
\begin{align*}
  & \bigg|\frac{1}{\sqrt{j_n}} \sum_{i\in \mathcal{J}_n}\1\{D_i=d^*\}\big(\widehat{\eta}_i^N(d,t,M_i)-\eta_i^N(d,t,M_i)\big) \left(\frac{1}{\widehat{p}_i^N(d^*)}-\frac{1}{p_i^N(d^*)}\right)\bigg| \\
  \leqslant & C \sqrt{j_n}\bigg[\frac{1}{j_n}\sum_{i\in\mathcal{J}_n}\big(\widehat{\eta}_i^N(d,t,M_i)-\eta_i^N(d,t,M_i)\big)^2\bigg]^{1/2}\bigg[\frac{1}{j_n}\sum_{i\in\mathcal{J}_n}\big(\widehat{p}_i^N(d^*)-p_i^N(d^*))^2\bigg]^{1/2} \\
  & = o_p(1).
\end{align*}
For the second summation, replace $\widehat{\eta}_i^N(d,t,M_i)-\eta_i^N(d,t,M_i)$ in the preceding display by $\widehat{\mu}_i^N(d,d^*,t)-\mu_i^N(d,d^*,t)$; the same argument shows that it is also $o_p(1)$. For the third summation, by (\ref{ch5:eq:proof-own-nde-exposure-propensity-inverse-bound}), (\ref{ch5:eq:proof-own-nde-mediator-density-ratio-bound}), $|Y_i-\eta_i^N(d,t,M_i)|\leqslant 2M$, the Cauchy--Schwarz inequality, and Assumption \ref{ch5:ass:own-natural-effects-gnn-convergence-rates}(i),
{\small
\begin{align*}
& \bigg|\frac{1}{\sqrt{j_n}} \sum_{i\in \mathcal{J}_n}\1\{D_i=d,T_i=t\}\bigl(Y_i-\eta_i^N(d,t,M_i)\bigr)
\left(\frac{\widehat q_i^N(M_i\mid d^*)}{\widehat q_i^N(M_i\mid d)}-\frac{q_i^N(M_i\mid d^*)}{q_i^N(M_i\mid d)}\right) \\
& \hspace{6cm} \left(\frac{1}{\widehat{e}_i^N(d,t)}-\frac{1}{e_i^N(d,t)}\right)\bigg| \\
& \leqslant 2MC^2 \sqrt{j_n} \bigg[\frac{1}{j_n}\sum_{i\in\mathcal{J}_n}\big(\widehat{q}_i^N(M_i \mid d^*) - q_i^N(M_i\mid d^*)\big)^2\bigg]^{1/2}\bigg[\frac{1}{j_n}\sum_{i\in\mathcal{J}_n}\big(\widehat{e}_i^N(d,t) - e_i^N(d,t)\big)^2\bigg]^{1/2} \\
& \quad + 2MC^2 \sqrt{j_n} \bigg[\frac{1}{j_n}\sum_{i\in\mathcal{J}_n}\big(\widehat{q}_i^N(M_i \mid d) - q_i^N(M_i\mid d)\big)^2\bigg]^{1/2}\bigg[\frac{1}{j_n}\sum_{i\in\mathcal{J}_n}\big(\widehat{e}_i^N(d,t) - e_i^N(d,t)\big)^2\bigg]^{1/2} \\
& = o_p(1).
\end{align*}
}%
The fourth through sixth summations are handled analogously. By Assumption \ref{ch5:ass:own-natural-effects-moments-overlap}(ii), $q_i^N(M_i\mid d^*)/q_i^N(M_i\mid d)$, $1/e_i^N(d,t)$, and $1/\widehat{e}_i^N(d,t)-1/e_i^N(d,t)$ are uniformly bounded. Treating these as bounded factors and applying (\ref{ch5:eq:proof-own-nde-exposure-propensity-inverse-bound}), (\ref{ch5:eq:proof-own-nde-mediator-density-ratio-bound}), the Cauchy--Schwarz inequality, and Assumption \ref{ch5:ass:own-natural-effects-gnn-convergence-rates}(i) to the remaining two estimation-error factors shows that each of the three terms is $o_p(1)$. Hence $R_{6n}(d,d^*,t)=o_p(1)$.

By Assumption \ref{ch5:ass:own-natural-effects-moments-overlap}(iii), $\sigma_{n,N}^{-1}=O_p(1)$, so $\sigma_{n,N}^{-1}R_{\ell n}(d,d^*,t)=o_p(1)$ for $\ell=1,\ldots,6$. Combining this with (\ref{ch5:eq:proof-own-nde-oracle-clt}) and Slutsky's theorem gives
\[
\sigma_{n,N}^{-1}\sqrt{j_n}\big(\widehat{ONDE}(t)-ONDE(t)\big) \xrightarrow{d} \mathcal{N}(0,1).
\]
This completes the proof.

\subsection{Proof of Theorem \ref{ch5:thm:own-nde-variance-consistency}}
The proof is identical to that of Theorem \ref{ch5:thm:own-cde-variance-consistency} and is therefore omitted.

\end{appendices}

\spacingset{1.4}
\bibliography{c.bib}

@article{rubin1974estimating,
  title={Estimating causal effects of treatments in randomized and nonrandomized studies.},
  author={Rubin, Donald B},
  journal={Journal of educational Psychology},
  volume={66},
  number={5},
  pages={688},
  year={1974},
  publisher={American Psychological Association}
}

@article{kojevnikov2021limit,
  title={Limit theorems for network dependent random variables},
  author={Kojevnikov, Denis and Marmer, Vadim and Song, Kyungchul},
  journal={Journal of Econometrics},
  volume={222},
  number={2},
  pages={882--908},
  year={2021},
  publisher={Elsevier}
}

@article{leung2022graph,
  title={Graph neural networks for causal inference under network confounding},
  author={Leung, Michael P and Loupos, Pantelis},
  journal={arXiv preprint arXiv:2211.07823},
  year={2022}
}

@article{leung2024identifying,
  title={Identifying treatment and spillover effects using exposure contrasts},
  author={Leung, Michael P},
  journal={arXiv preprint arXiv:2403.08183},
  year={2024}
}

@article{aronow2017estimating,
 author = {Peter M. Aronow and Cyrus Samii},
 journal = {The Annals of Applied Statistics},
 number = {4},
 pages = {1912--1947},
 publisher = {Institute of Mathematical Statistics},
 title = {Estimating average causal effects under general interference, with application to a social network experiment},
 urldate = {2026-06-13},
 volume = {11},
 year = {2017}
}

@article{savje2024causal,
  title={Causal inference with misspecified exposure mappings: separating definitions and assumptions},
  author={S{\"a}vje, Fredrik},
  journal={Biometrika},
  volume={111},
  number={1},
  pages={1--15},
  year={2024},
  publisher={Oxford University Press}
}

@article{forastiere2021identification,
  title={Identification and estimation of treatment and interference effects in observational studies on networks},
  author={Forastiere, Laura and Airoldi, Edoardo M and Mealli, Fabrizia},
  journal={Journal of the American Statistical Association},
  volume={116},
  number={534},
  pages={901--918},
  year={2021},
  publisher={Taylor \& Francis}
}

@article{leung2022causal,
  title={Causal inference under approximate neighborhood interference},
  author={Leung, Michael P},
  journal={Econometrica},
  volume={90},
  number={1},
  pages={267--293},
  year={2022},
  publisher={Wiley Online Library}
}

@article{li2022random,
  title={Random graph asymptotics for treatment effect estimation under network interference},
  author={Li, Shuangning and Wager, Stefan},
  journal={The Annals of Statistics},
  volume={50},
  number={4},
  pages={2334--2358},
  year={2022},
  publisher={Institute of Mathematical Statistics}
}

@article{kundu2024network,
  title={Network structural equation models for causal mediation and spillover effects},
  author={Kundu, Ritoban and Song, Peter XK},
  journal={arXiv preprint arXiv:2412.05397},
  year={2024}
}

@article{corso2020principal,
  title={Principal neighbourhood aggregation for graph nets},
  author={Corso, Gabriele and Cavalleri, Luca and Beaini, Dominique and Li{\`o}, Pietro and Veli{\v{c}}kovi{\'c}, Petar},
  journal={Advances in neural information processing systems},
  volume={33},
  pages={13260--13271},
  year={2020}
}

@article{chernozhukov2018double,
  title={Double/debiased machine learning for treatment and structural parameters},
  author={Chernozhukov, Victor and Chetverikov, Denis and Demirer, Mert and Duflo, Esther and Hansen, Christian and Newey, Whitney and Robins, James},
  journal={The Econometrics Journal},
  volume={21},
  number={1},
  year={2018},
  publisher={The Econometrics Journal}
}

@article{farrell2021deep,
  title={Deep neural networks for estimation and inference},
  author={Farrell, Max H and Liang, Tengyuan and Misra, Sanjog},
  journal={Econometrica},
  volume={89},
  number={1},
  pages={181--213},
  year={2021},
  publisher={Wiley Online Library}
}

@article{klyne2026average,
  title={Average partial effect estimation using double machine learning},
  author={Klyne, Harvey and Shah, Rajen D},
  journal={The Annals of Statistics},
  volume={54},
  number={1},
  pages={176--200},
  year={2026},
  publisher={Institute of Mathematical Statistics}
}

@article{hoshino2024causal,
  title={Causal inference with noncompliance and unknown interference},
  author={Hoshino, Tadao and Yanagi, Takahide},
  journal={Journal of the American Statistical Association},
  volume={119},
  number={548},
  pages={2869--2880},
  year={2024},
  publisher={Taylor \& Francis}
}

@article{chernozhukov2022locally,
  title={Locally robust semiparametric estimation},
  author={Chernozhukov, Victor and Escanciano, Juan Carlos and Ichimura, Hidehiko and Newey, Whitney K and Robins, James M},
  journal={Econometrica},
  volume={90},
  number={4},
  pages={1501--1535},
  year={2022},
  publisher={Wiley Online Library}
}

@article{cai2015social,
  title={Social networks and the decision to insure},
  author={Cai, Jing and Janvry, Alain De and Sadoulet, Elisabeth},
  journal={American Economic Journal: Applied Economics},
  volume={7},
  number={2},
  pages={81--108},
  year={2015},
  publisher={American Economic Association 2014 Broadway, Suite 305, Nashville, TN 37203-2425}
}

@article{robins1992identifiability,
  author  = {Robins, James M. and Greenland, Sander},
  title   = {Identifiability and exchangeability for direct and indirect effects},
  journal = {Epidemiology},
  year    = {1992},
  volume  = {3},
  number  = {2},
  pages   = {143--155}
}

@inproceedings{pearl2001direct,
  author    = {Pearl, Judea},
  title     = {Direct and indirect effects},
  booktitle = {Proceedings of the Seventeenth Conference on Uncertainty in
               Artificial Intelligence},
  year      = {2001},
  pages     = {411--420}
}

@article{imai2010identification,
  author  = {Imai, Kosuke and Keele, Luke and Yamamoto, Teppei},
  title   = {Identification, inference and sensitivity analysis for causal
             mediation effects},
  journal = {Statistical Science},
  year    = {2010},
  volume  = {25},
  number  = {1},
  pages   = {51--71}
}

@article{tchetgen2012semiparametric,
  author  = {Tchetgen Tchetgen, Eric J. and Shpitser, Ilya},
  title   = {Semiparametric theory for causal mediation analysis: Efficiency
             bounds, multiple robustness and sensitivity analysis},
  journal = {The Annals of Statistics},
  year    = {2012},
  volume  = {40},
  number  = {3},
  pages   = {1816--1845}
}

@book{vanderweele2015explanation,
  author    = {VanderWeele, Tyler J.},
  title     = {Explanation in Causal Inference: Methods for Mediation and
               Interaction},
  publisher = {Oxford University Press},
  year      = {2015},
  address   = {New York}
}

@article{hudgens2008toward,
  author  = {Hudgens, Michael G. and Halloran, M. Elizabeth},
  title   = {Toward causal inference with interference},
  journal = {Journal of the American Statistical Association},
  year    = {2008},
  volume  = {103},
  number  = {482},
  pages   = {832--842}
}

@article{ogburn2014causal,
  author  = {Ogburn, Elizabeth L. and VanderWeele, Tyler J.},
  title   = {Causal diagrams for interference},
  journal = {Statistical Science},
  year    = {2014},
  volume  = {29},
  number  = {4},
  pages   = {559--578}
}

@article{ogburn2024causal,
  author  = {Ogburn, Elizabeth L. and Sofrygin, Oleg and D{\'i}az, Iv{\'a}n and
             van der Laan, Mark J.},
  title   = {Causal inference for social network data},
  journal = {Journal of the American Statistical Association},
  year    = {2024},
  volume  = {119},
  number  = {545},
  pages   = {597--611}
}

@article{vanderweele2013mediation,
  author  = {VanderWeele, Tyler J. and Hong, Guanglei and Jones, Stephanie M.
             and Brown, Joshua L.},
  title   = {Mediation and spillover effects in group-randomized trials: A case
             study of the {4R}s educational intervention},
  journal = {Journal of the American Statistical Association},
  year    = {2013},
  volume  = {108},
  number  = {502},
  pages   = {469--482}
}

@article{liu2021social,
  author  = {Liu, Haiyan and Jin, Ick Hoon and Zhang, Zhiyong and Yuan, Ying},
  title   = {Social network mediation analysis: A latent space approach},
  journal = {Psychometrika},
  year    = {2021},
  volume  = {86},
  number  = {1},
  pages   = {272--298}
}

@article{che2021network,
  author  = {Che, Chuanji and Jin, Ick Hoon and Zhang, Zhiyong},
  title   = {Network mediation analysis using model-based eigenvalue decomposition},
  journal = {Structural Equation Modeling: A Multidisciplinary Journal},
  year    = {2021},
  volume  = {28},
  number  = {1},
  pages   = {148--161}
}

@article{robins1994estimation,
  author  = {Robins, James M. and Rotnitzky, Andrea and Zhao, Lue Ping},
  title   = {Estimation of regression coefficients when some regressors are not
             always observed},
  journal = {Journal of the American Statistical Association},
  year    = {1994},
  volume  = {89},
  number  = {427},
  pages   = {846--866}
}

@article{bang2005doubly,
  author  = {Bang, Heejung and Robins, James M.},
  title   = {Doubly robust estimation in missing data and causal inference models},
  journal = {Biometrics},
  year    = {2005},
  volume  = {61},
  number  = {4},
  pages   = {962--973}
}

@inproceedings{kipf2017semi,
  author    = {Kipf, Thomas N. and Welling, Max},
  title     = {Semi-supervised classification with graph convolutional networks},
  booktitle = {International Conference on Learning Representations},
  year      = {2017}
}

@inproceedings{hamilton2017inductive,
  author    = {Hamilton, William L. and Ying, Rex and Leskovec, Jure},
  title     = {Inductive representation learning on large graphs},
  booktitle = {Advances in Neural Information Processing Systems},
  year      = {2017},
  volume    = {30},
  pages     = {1024--1034}
}

\end{document}